\documentclass[preprint,12pt]{elsarticle}

\newif\iflong
\newif\ifshort

 \longtrue

\iflong
\else
\shorttrue
\fi

\newcommand{\qedclaim}{\hfill $\diamond$ \medskip}
\newenvironment{proofclaim}{\noindent{\em Proof of the claim.}}{\qedclaim}
\usepackage{xspace}
\usepackage{vmargin,amsthm}
\usepackage{thm-restate,amsmath,amsfonts,amssymb}
\usepackage{color,url,hyperref,cleveref,todonotes}
\usepackage{caption,subcaption}

\usepackage{graphicx}
  \baselineskip 8 mm

\newcommand{\Pb}[4]{%
\begin{center}
  \begin{tabular}{|l|}%
  \hline
    \begin{minipage}[c]{0.96\textwidth}
      \smallskip%
      \par\noindent%
      #1%
      \par\noindent%
      \textbf{\textsf{Input}}: #2%
      \par\noindent%
      \textbf{\textsf{#3}}: #4 
      \smallskip%
      \par\noindent%
    \end{minipage}
  \\\hline
  \end{tabular}%
\end{center}
}%
\begin{document}

\newcommand{\np}{\textsf{NP}}
\newcommand{\p}{\textsf{P}}
\newcommand{\W}{\textsf{W}}
\newcommand{\FPT}{\tectsf{FPT}}
\newcommand{\C}{\mathcal{C}}
\newcommand{\Par}{\mathcal{P}}
\newcommand{\CCFw}{\textsc{Weighted $\C$-Coalition Formation}\xspace}
\newcommand{\CCF}{\textsc{$\C$-Coalition Formation}\xspace}
\newcommand{\GFF}{\textsc{General Factors}\xspace}
\newcommand{\GF}{\textsc{GF}\xspace}


\newcommand{\Copy}{4}
\newcommand{\sizeIu}{\textcolor{black}{30}} 
\newcommand{\sizeIvu}{\textcolor{black}{3}}
\newcommand{\sizeCar}{\textcolor{black}{42}}
\newcommand{\sizeOfListGadget}{\textcolor{black}{A}}
\newcommand{\sizeMissingOfVertexGadget}{\textcolor{black}{B}}
\newcommand{\sizeMissingOfEdgeGadget}{\textcolor{black}{D}}
\newcommand{\optimalValueWhenSatIsTrue}{\textcolor{black}{m_v|V(H)|  + m_\ell |V(H)|^2  + m_e |E(H)|   + 10 A  |V(H)| + 8 |E(H)| }}
\newcommand{\bestValue}{\textcolor{black}{|V(H)| \cdot m_v + |V(H)|^2 \cdot m_\ell + |E(H)| \cdot( m_e + 8 )}}

\newcommand{\weightStarsRight}{{4\C^3}}
\newcommand{\weightStarsColor}{{3\C^2}}
\newcommand{\weightXColor}{{2\C}}

\newcommand{\tw}{\text{\rm tw}}
\newcommand{\vc}{\text{\rm vc}}
\newcommand{\td}{\text{\rm td}}
\newcommand{\val}{\text{\rm val}}
\newcommand{\bO}{\mathcal{O}}

\newtheorem{theorem}{Theorem}[section]
\newtheorem{definition}[theorem]{Definition}
\newtheorem{lemma}[theorem]{Lemma}
\newtheorem{corollary}[theorem]{Corollary}
\newtheorem{claim}[theorem]{Claim}

\newtheorem{remark}[theorem]{Remark}
\newtheorem{observation}[theorem]{Observation}
\newtheorem{proposition}[theorem]{Proposition}
\newtheorem{conjecture}[theorem]{Conjecture}
\newtheorem{theorem1}{Theorem}
\newtheorem{RR}[theorem1]{Reduction Rule}
\newtheorem{question}[theorem]{Question}

\title{Exact Algorithms and Lower Bounds for Forming Coalitions of Constrained Maximum Size}


\author[label1]{Foivos Fioravantes}
\author[label2,label3]{Harmender Gahlawat}
\author[label4]{Nikolaos Melissinos}
\affiliation[label1]{organization={Department of Theoretical Computer Science, Faculty of Information Technology, Czech Technical University in Prague},
            city={Prague},
            country={Czech Republic}}

\affiliation[label2]{organization={Universit\'e Clermont Auvergne, CNRS, Clermont Auvergne INP, Mines Saint-\'Etienne, LIMOS},
            city={Clermont-Ferrand},
            postcode={63000},
            country={France}}

\affiliation[label3]{organization={G-SCOP, Grenoble-INP},
            city={Grenoble},
            postcode={38000},
            country={France}}

\affiliation[label4]{organization={Computer Science Institute, Faculty of Mathematics and Physics, Charles University},
            city={Prague},
            country={Czech Republic}}

\begin{abstract}
Imagine we want to split a group of agents into teams in the most \emph{efficient} way, considering that each agent has their own preferences about their teammates. This scenario is modeled by the extensively studied \textsc{Coalition Formation} problem. Here, we study a version of this problem where each team must additionally be of bounded size. 

We conduct a systematic algorithmic study, providing several intractability results as well as multiple exact algorithms that scale well as the input grows (FPT), which could prove useful in practice.

Our main contribution is an algorithm that deals efficiently with tree-like structures (bounded \emph{treewidth}) for ``small'' teams. We complement this result by proving that our algorithm is asymptotically optimal. Particularly, there can be no algorithm that vastly outperforms the one we present, under reasonable theoretical assumptions, even when considering star-like structures (bounded \emph{vertex cover number}). 
\end{abstract}

\maketitle

\section{Introduction}\label{S:Introduction}
{\let\thefootnote\relax\footnotetext{A preliminary version of this paper appeared in the proceedings of the 39th Annual AAAI Conference on Artificial Intelligence (AAAI 2025)~\cite{ourAAAISmallCars}.}}\par

Coalition Formation is a central topic in Computational Social Choice and economic game theory~\cite{handbook}. The goal is to partition a set of agents into coalitions to \textit{optimize} some \textit{utility} function. One well-studied notion in Coalition Formation is \textit{Hedonic Games}~\cite{dreze1980hedonic}, where the utility of an agent depends solely on the coalition it is placed in. Due to their extremely general nature that captures numerous scenarios, hedonic games are intensively studied in computer science~\cite{aziz2019fractional,barrot2019unknown,boehmer2020individual,brandt2023reaching,bullinger2021loyalty,fanelli2021relaxed,igarashi2019robustness,ohta2017core,sliwinski2017learning}, and are shown to have applications in social network analysis~\cite{social}, scheduling group activities~\cite{schedule}, and allocating tasks to wireless agents~\cite{wireless}.   

Due to its general nature, most problems concerning the computational complexity of hedonic games are hard~\cite{peters2015simple}. In fact, even encoding the preferences of agents, in general, takes exponential space, which motivates the study of succinct representations for agent preferences. One of the most-studied such class of games is Additive Separable Hedonic Games~\cite{bogomolnaia2002stability}, where the agents are represented by the vertices of a weighted graph and the weight of each edge represents the \textit{utility} of the agents joined by the edge for each other (see also Weighted Graphical Games model of~\cite{deng1994complexity}). Variants where the agent preferences are asymmetric are modeled using directed graphs. Here, the utility of an agent for a group of agents is \textit{additive} in nature. Additive Separable Hedonic Games are well-studied in the literature~\cite{aloisio2020impact,aziz2013computing,barrot2019stable}.

Most literature in the Additive Separable Hedonic Games considers the agents to be \textit{selfish} in nature and, hence, the notion used to measure the efficiency is that of \textit{stability}~\cite{peters2015simple}, including \textit{core stability}, \textit{Nash Stability}, \textit{individual stability}, etc. Semi-altruistic approaches where the agents are concerned about their \textit{relative}'s utility along with theirs are also studied~\cite{monaco2021additively}.  A standard altruistic approach in computational social choice is that of \textit{utilitarian social welfare}, where the goal is to maximize the total sum of utility of all the agents.

Observe that if all edge weights are positive, then the maximum utilitarian utility is achieved by putting all agents in the same coalition. But there are many practical scenarios, e.g., forming office teams to allocate several projects or allocating cars/buses to people for a trip, where we additionally require that each coalition should be of a bounded size. Coalition formations with constrained coalition size have recently been a focus of attention in ASHGs~\cite{levinger2023social} and in Fractional Hedonic Games~\cite{monaco2023nash}. Further, coalition formations where each coalition needs to be of a fixed size have also been studied~\cite{bilo2022hedonic,cseh2019pareto}.

We consider the Additive Separable Hedonic Games with an additional constraint on the maximum allowed size of a coalition (denoted by $\mathcal{C}$), with the goal to maximize the total sum of utility of all the agents. We denote this as the $\C$-\textsc{Coalition Formation} problem (\CCF~for short). We provide the formal problem definition, along with other preliminaries, in Section~\ref{S:prelim}. This game is shown to be \np-hard even when $\mathcal{C}=3$~\cite{levinger2023social} (and hence \W-hard parameterized by $\mathcal{C}$) via a  straightforward reduction from the \textsc{Partition Into Triangles}, which is \np-hard even for graphs with $\Delta\leq 4$~\cite{RNB13}. Therefore, we consider the \textit{parameterized complexity} of this problem through the lens of various structural parameters of the input graph and present a comprehensive analysis of its computational complexity. 

In parameterized complexity, the goal is to restrict the exponential blow-up of running time to some \textit{parameter} of the input (which is usually much smaller than the input size) rather than the whole input size. Due to its practical efficiency, this paradigm has been used extensively to study problems arising from Computational Social Choice and Artificial Intelligence~\cite{backstrom2012complexity,bessiere2008parameterized,bredereck2017parliamentary} (including hedonic games~\cite{ganian2023hedonic,lampis2022hedonic,P16,hedonicTreewidth,CCRS23,L21,HIO23,HKMO19}). 

It is worth mentioning that \CCF~has been studied from an approximation perspective and is shown to have applications in Path Transversals~\cite{soda}. Moreover,~\cite{approxGraphical} considered a Weighted Graphical Game to maximize social welfare and provided constant-factor approximation for restricted families of graphs. Finally, \cite{onlineAAMAS} considered the online version of several Weighted Graphical Games (aiming to maximize utilitarian social welfare), in one of which the authors also consider coalitions of bounded size.



\subsection*{Our contribution}

In this paper we study the parameterized complexity of $\C$-\textsc{Coalition Formation} problem, which is a version of the \textsc{Coalition Formation} problem with the added constraint that each coalition should be of size at most $\C$. We consider two distinct variants of this problem according to the possibilities for the utilities of the agents. In the \emph{unweighted} version, the utilities of all the pairs of agents are either $0$ (there is no edges connecting them) or $1$. In the \emph{weighted} version, the utilities of all pairs of agents are given by natural numbers. We will refer to the former as \CCF~and the latter as \CCFw~, respectively. In both cases, the underlying structure is assumed to be an undirected graph. In particular, this implies that the valuation are assumed to be \textit{symmetric}.

We begin by noting an interesting connection to the notion of \textit{Nash-stability}. Consider a solution $\Par=\{C_1,...,C_{\ell}\}$. Roughly speaking, $\Par$ is Nash-stable if any agent does not have any additional gain by leaving its coalition and joining another that can still accommodate it. In our setting, a solution is Nash-stable if for each agent $u$, if $u\in C_i$, then for each $j\in [\ell]$ such that $|C_j|<\ell$, $\sum_{x \in N_{C_i}} w(ux) \geq \sum_{y \in N_{C_i}}w(uy)$. Observe that an optimal solution for \CCF~(similarly, \CCFw) is also \textit{Nash-stable}. This is trivially true because if this condition is not satisfied for some $u\in C_i$, then we can move $u$ to $C_j$ while respecting the maximum size constraint and increase the valuation. Moreover, the notion of utilitarian social welfare captures the notion of Nash stability when the coalitions are required to be of bounded size and the valuations are symmetric. Hence, our positive results stand even when the goal is to obtain a Nash-stable coalition.

\begin{figure}[!t]
\centering
\includegraphics[scale=1]{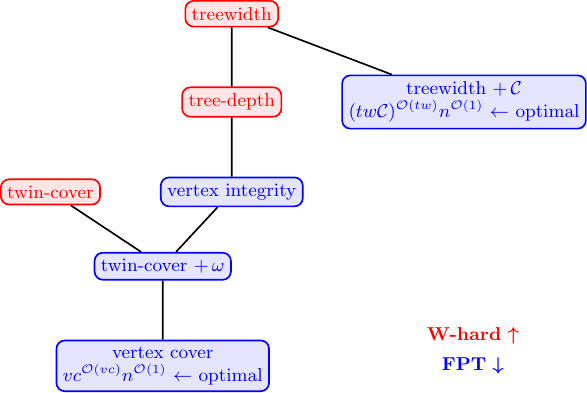}
\caption{Overview of our results. A parameter $A$ appearing linked to a parameter $B$ with $A$ being below $B$ is to be understood as ``there is a function $f$ such that $f(A)\geq f(B)$''. In blue (red resp.) we exhibit the FPT ($\W[1]$-hardness resp.) results we provide. The clique number of the graph is denoted by $\omega$. Note that our FPT results are for the more general, weighted version of the problem ($\C$-CFw), while our $\W[1]$-hardness results are for the more restricted, unweighted version of the problem ($\C$-CF). Finally, note that this figure does not include our results concerning the kernelization of $\C$-CF.}\label{fig:results}
\end{figure}

It follows from our previous discussion on the hardness of the \CCF problem that it is para-\np-hard parameterized by $\C+\Delta$. Thus, we consider other structural parameters of the input graph. The majority of our new results are summarized in Figure~\ref{fig:results}. We initiate our study with arguably the most studied structural graph parameter \textit{treewidth}, denoted by $tw$, which measures how tree-like a graph is. Treewidth is  a natural parameter of choice when considering problems admitting a graph structure, specifically in problems related to AI, because many real world networks exhibit bounded treewidth~\cite{maniu2019experimental} and tree decompositions  of ``almost optimal'' treewidth are easy to compute in FPT time~\cite{KorhonenL23}. We begin with showing that \CCFw is FPT when parameterized by $tw+\C$ by application of a bottom-up dynamic programming.

\begin{restatable}{theorem}{thmTwCFPT}\label{thm:TwC-FPT}
    The \CCFw~problem can be solved in time $(\tw\C)^{\bO(\tw)} n^{\bO(1)}$, where $\tw$ is the treewidth of the input graph. 
\end{restatable}

The complexity in the above algorithm has a non-polynomial dependency on $\C$. It is natural to wonder whether there can be an efficient algorithm that avoids this, i.e., if \CCFw (or even \CCF) is FPT parameterized by $tw$. We answer this question negatively in the following theorem by establishing that even \CCF remains W[1]-hard when parameterized by treedepth, a more restrictive parameter, justifying our choice of parameters for our FPT algorithm. In particular, we have the following theorem.

\begin{restatable}{theorem}{thmTreedepth}\label{thm:treedepth}
    The \CCF~problem is \W$[1]$-hard when parameterized by the tree-depth of the input graph. 
\end{restatable}

Nevertheless, we do achieve such an algorithm (having only polynomial dependence on $\C$) by allowing the input to have a star-like structure. In the following statements, $\vc$ denotes the vertex cover number of the input graph.

\begin{restatable}{theorem}{thmVcFPT}\label{thm:vc-FPT}
    The \CCFw~problem can be solved in time $\vc^{\bO(\vc)}n^{\bO(1)}$, where $\vc$ denotes the vertex cover number of the input graph. 
\end{restatable}

The next question we consider is whether the running time of our algorithms can be improved. In this regard, we prove that both of the above algorithms are, essentially, optimal, \emph{i.e.}, we do not expect a drastic improvement in their running times even if we restrict ourselves to the unweighted case. 

\begin{restatable}{theorem}{thmVcToTheVc}\label{thm:vc-to-the-vc}
    Unless the ETH fails, \CCF~does not admit an algorithm running in time $(\C\vc)^{o(\vc+\C)}n^{\mathcal{O}(1)}$, where $\vc$ denotes the vertex cover number of the input graph.
\end{restatable}

We then slightly shift our approach and attack this problem using the toolkit of \emph{kernelization}. Intuitively, our goal is to ``peel off'' the useless parts of the input (in polynomial time) and solve the problem for the ``small'' part of the input remaining, known as the \emph{kernel}. 
Due to its profound impact, kernelization was termed ``the lost continent of polynomial time''~\cite{kernelApplication}. It is specifically useful in practical applications as it has shown tremendous speedups in practice~\cite{gao2009data,guo2007invitation,niedermeier2000general,weihe1998covering}. We begin with providing a polynomial kernel for the unweighted case when parameterized by the vertex cover number of the input graph and $\C$.


\begin{restatable}{theorem}{thmPolyKernel}\label{thm:poly-kernel}
    \CCF~admits a kernel with $\mathcal{O} (\vc^2 \C)$ vertices, where $\vc$ denotes the vertex cover number of the input graph. 
\end{restatable}

It is well known that a problem admits a kernel iff it is FPT~\cite{CyganFKLMPPS15}. Hence, the notion of ``tractability'' in kernelization comes from designing polynomial kernels, and a problem is considered ``intractable'' from the kernelization point of view if it is unlikely to admit a polynomial kernel for the considered parameter. One may wonder if we can lift our kernelization algorithm for the weighted case. We answer this question negatively by proving that, unfortunately, there can be no such kernel for the weighted version. In some sense, this signifies that weights present a barrier for the kernelization of the problems considered here.

\begin{restatable}{theorem}{thmNoKernelVc}\label{thm:no-kernel-vc}
    \CCFw parameterized by $\vc + \C$ does not admit a polynomial kernel, unless polynomial hierarchy collapses, where $\vc$ denotes the vertex cover number of the input graph. 
\end{restatable}

We close our study by considering additional structural parameters for the unweighted case. We postpone the formal definition of these parameters until Section~\ref{subsec:structural-param}.

\begin{restatable}{theorem}{thmVi}\label{thm:vertex-integrity}
    The \CCF~problem can be solved in FPT time when parameterized by the vertex integrity of the input graph. 
\end{restatable}

\begin{restatable}{theorem}{thmTwinCover}\label{thm:twin-cover}
    The \CCF~problem is \W$[1]$-hard when parameterized by the twin-cover number of $G$. 
\end{restatable}

The choice to focus our attention to the above two parameters is not arbitrary. Let $G$ be a graph with vertex integrity $\textrm{vi}$, twin-cover number $\textrm{twc}$ and vertex cover number $\vc$. Then, $\textrm{vi}\leq \textrm{twc}+\omega(G)$ and $\textrm{twc}\leq \textrm{twc}+\omega(G)$, where $\omega(G)$ is the clique number of $G$. Finally, $\textrm{twc}+\omega(G)\leq f(\vc)$, for some computable function $f$. 
Taking the above into consideration, our Theorems~\ref{thm:vertex-integrity} and~\ref{thm:twin-cover} provide a clear dichotomy of the tractability of \CCF~when considering these parameters.

\section{Preliminaries}\label{S:prelim}

\iflong
\subsection{Graph Theory}
We follow standard graph-theoretic notation~\cite{D12}. In particular, we will use $V(G)$ and $E(G)$ to refer to the vertices and edges of $G$ respectively; if no ambiguity arises, the parenthesis referring to $G$ will be dropped. Moreover, we denote by $N_G(v)$ the \textit{neighbors} of $v$ in $G$ and  we use $d_G(v)$ to denote the \textit{degree} of $v$ in $G$. That is, $N_G(v)=\{u\in V(G)|uv\in E(G)\}$ and $d_G(v)=|N_G(u)|$. Note that the subscripts may be dropped if they are clearly implied by the context. The \textit{maximum degree} of $G$ is denoted by $\Delta(G)$, or simply $\Delta$ when clear by the context. Given a graph $G=(V,E)$ and a set $E'\subseteq E$, we use $G-E'$ to denote the graph resulting from the deletion of the edges of $E'$ from $G$. Finally, for any integer and $n$, we denote $[n]$ the set of all integers between $1$ and $n$. That is, $[n]=\{1,\dots,n\}$.

\subsection{Problem Formulation}
Formally, the input of the \CCFw~consists of a graph $G=(V,E)$ and an edge-weight function $w:E\rightarrow \mathbb{N}$. Additionally, we are given a \emph{capacity} $\C\in \mathbb{N}$ as part of the input. Our goal is to find a \emph{$\C$-partition} of $V$, that is, a partition $\Par= \{C_1,\dots,C_p\}$ such that $|C_i|\leq \C$ for each $i\in[p]$. For each $i\in [p]$, let $E_i$ denote the edges of $G[C_i]$.
Let $E(\Par)$ be the set of edges of the partition $\Par$, \emph{i.e.}, $E(\Par) = \bigcup_{i=1}^p E(G[C_i])$.
The \emph{value} of a $\C$-partition $\Par$ is: $v(\Par)=\sum_{i=1}^p\sum_{e\in E(C_i)} w(e)$.
We are interested in computing an \emph{optimal} $\C$-partition, \emph{i.e.}, a $\C$-partition of maximum value. Note that we will also use the defined notations for general (not necessarily $\C$-)partitions. 

We are also interested in the \emph{unweighted} version of the \CCF~problem, where each edge of the input graph has a weight of $1$; in such cases, the input of the problem will only consist of the graph and the required capacity. In Figure~\ref{fig:example} we illustrate an example for this unweighted version.

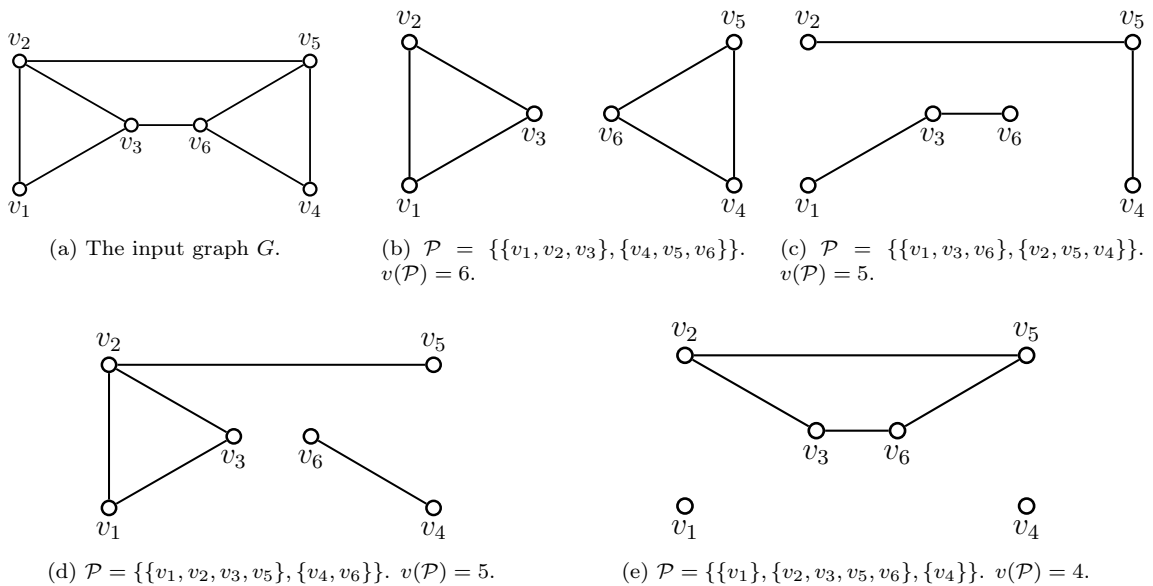
\begin{figure}[!t]
\centering

\begin{subfigure}[t]{0.3\textwidth}
\centering
\scalebox{0.85}{
\begin{tikzpicture}[inner sep=0.7mm]
    \node[draw, circle, line width=1pt, fill=white](v1) at (0,0)  [label=below: {$v_1$}]{};
    \node[draw, circle, line width=1pt, fill=white](v2) at (0,2)  [label=above: {$v_2$}]{};
    \node[draw, circle, line width=1pt, fill=white](v3) at (1.73,1)  [label=below: {$v_3$}]{};
    \node[draw, circle, line width=1pt, fill=white](v4) at (4.5,0)  [label=below: {$v_4$}]{};
    \node[draw, circle, line width=1pt, fill=white](v5) at (4.5,2)  [label=above: {$v_5$}]{};
    \node[draw, circle, line width=1pt, fill=white](v6) at (2.8,1)  [label=below: {$v_6$}]{};
    
    \draw[-, line width=0.8pt]  (v1) -- (v2);
    \draw[-, line width=0.8pt]  (v2) -- (v3);
    \draw[-, line width=0.8pt]  (v3) -- (v1);
    \draw[-, line width=0.8pt]  (v4) -- (v5);
    \draw[-, line width=0.8pt]  (v5) -- (v6);
    \draw[-, line width=0.8pt]  (v6) -- (v4);
    \draw[-, line width=0.8pt]  (v2) -- (v5);
    \draw[-, line width=0.8pt]  (v3) -- (v6);
\end{tikzpicture}
}
\caption{The input graph $G$.}
\end{subfigure}
\hspace{5pt}
\begin{subfigure}[t]{0.3\textwidth}
\centering
\scalebox{0.95}{
\begin{tikzpicture}[inner sep=0.7mm]
    \node[draw, circle, line width=1pt, fill=white](v1) at (0,0)  [label=below: {$v_1$}]{};
    \node[draw, circle, line width=1pt, fill=white](v2) at (0,2)  [label=above: {$v_2$}]{};
    \node[draw, circle, line width=1pt, fill=white](v3) at (1.73,1)  [label=below: {$v_3$}]{};
    \node[draw, circle, line width=1pt, fill=white](v4) at (4.5,0)  [label=below: {$v_4$}]{};
    \node[draw, circle, line width=1pt, fill=white](v5) at (4.5,2)  [label=above: {$v_5$}]{};
    \node[draw, circle, line width=1pt, fill=white](v6) at (2.8,1)  [label=below: {$v_6$}]{};
    
    \draw[-, line width=0.8pt]  (v1) -- (v2);
    \draw[-, line width=0.8pt]  (v2) -- (v3);
    \draw[-, line width=0.8pt]  (v3) -- (v1);
    \draw[-, line width=0.8pt]  (v4) -- (v5);
    \draw[-, line width=0.8pt]  (v5) -- (v6);
    \draw[-, line width=0.8pt]  (v6) -- (v4);
\end{tikzpicture}
}
\caption{$\Par=\{\{v_1,v_2,v_3\},\{v_4,v_5,v_6\}\}$. $v(\Par)=6$.}
\end{subfigure}
\hspace{5pt}
\begin{subfigure}[t]{0.3\textwidth}
\centering
\scalebox{0.95}{
\begin{tikzpicture}[inner sep=0.7mm]
    \node[draw, circle, line width=1pt, fill=white](v1) at (0,0)  [label=below: {$v_1$}]{};
    \node[draw, circle, line width=1pt, fill=white](v2) at (0,2)  [label=above: {$v_2$}]{};
    \node[draw, circle, line width=1pt, fill=white](v3) at (1.73,1)  [label=below: {$v_3$}]{};
    \node[draw, circle, line width=1pt, fill=white](v4) at (4.5,0)  [label=below: {$v_4$}]{};
    \node[draw, circle, line width=1pt, fill=white](v5) at (4.5,2)  [label=above: {$v_5$}]{};
    \node[draw, circle, line width=1pt, fill=white](v6) at (2.8,1)  [label=below: {$v_6$}]{};
    
    \draw[-, line width=0.8pt]  (v3) -- (v1);
    \draw[-, line width=0.8pt]  (v4) -- (v5);
    \draw[-, line width=0.8pt]  (v2) -- (v5);
    \draw[-, line width=0.8pt]  (v3) -- (v6);
\end{tikzpicture}
}
\caption{$\Par=\{\{v_1,v_3,v_6\},\{v_2,v_5,v_4\}\}$. $v(\Par)=5$.}
\end{subfigure}

\vspace{10pt}

\begin{subfigure}[t]{0.42\textwidth}
\centering
\scalebox{0.95}{
\begin{tikzpicture}[inner sep=0.7mm]
    \node[draw, circle, line width=1pt, fill=white](v1) at (0,0)  [label=below: {$v_1$}]{};
    \node[draw, circle, line width=1pt, fill=white](v2) at (0,2)  [label=above: {$v_2$}]{};
    \node[draw, circle, line width=1pt, fill=white](v3) at (1.73,1)  [label=below: {$v_3$}]{};
    \node[draw, circle, line width=1pt, fill=white](v4) at (4.5,0)  [label=below: {$v_4$}]{};
    \node[draw, circle, line width=1pt, fill=white](v5) at (4.5,2)  [label=above: {$v_5$}]{};
    \node[draw, circle, line width=1pt, fill=white](v6) at (2.8,1)  [label=below: {$v_6$}]{};
    
    \draw[-, line width=0.8pt]  (v1) -- (v2);
    \draw[-, line width=0.8pt]  (v2) -- (v3);
    \draw[-, line width=0.8pt]  (v3) -- (v1);
    \draw[-, line width=0.8pt]  (v6) -- (v4);
    \draw[-, line width=0.8pt]  (v2) -- (v5);
\end{tikzpicture}
}
\caption{$\Par=\{\{v_1,v_2,v_3,v_5\},\{v_4,v_6\}\}$. $v(\Par)=5$.}
\end{subfigure}
\hspace{20pt}
\begin{subfigure}[t]{0.42\textwidth}
\centering
\scalebox{1}{
\begin{tikzpicture}[inner sep=0.7mm]
    \node[draw, circle, line width=1pt, fill=white](v1) at (0,0)  [label=below: {$v_1$}]{};
    \node[draw, circle, line width=1pt, fill=white](v2) at (0,2)  [label=above: {$v_2$}]{};
    \node[draw, circle, line width=1pt, fill=white](v3) at (1.73,1)  [label=below: {$v_3$}]{};
    \node[draw, circle, line width=1pt, fill=white](v4) at (4.5,0)  [label=below: {$v_4$}]{};
    \node[draw, circle, line width=1pt, fill=white](v5) at (4.5,2)  [label=above: {$v_5$}]{};
    \node[draw, circle, line width=1pt, fill=white](v6) at (2.8,1)  [label=below: {$v_6$}]{};
    
    \draw[-, line width=0.8pt]  (v2) -- (v3);
    \draw[-, line width=0.8pt]  (v5) -- (v6);
    \draw[-, line width=0.8pt]  (v2) -- (v5);
    \draw[-, line width=0.8pt]  (v3) -- (v6);
\end{tikzpicture}
}
\caption{$\Par=\{\{v_1\},\{v_2,v_3,v_5,v_6\},\{v_4\}\}$. $v(\Par)=4$.}
\end{subfigure}

\caption{An example of possible solutions to the unweighted version of the \textsc{$\mathcal{C}$-Coalition Formation} problem. The input consists of the graph $G$ illustrated in subfigure (a), and the capacity $\mathcal{C}=4$. The $4$-partition in  subfigure (b) has the optimal value of 6. Observe that every possible $4$-partition that includes a set of $4$ vertices (some of which are not included here) is suboptimal.}
\label{fig:example}
\end{figure}

\fi

\subsection{Parameterized Complexity - Kernelization}
\emph{Parameterized complexity} is a computational paradigm that extends classical measures of time complexity. The goal is to examine the computational complexity of problems with respect to an additional measure, referred to as the \textit{parameter}.
Formally, a parameterized problem is a set of instances $(x,k) \in \Sigma^* \times \mathbb{N}$, where $k$ is called the parameter of the instance.
A parameterized problem is \emph{Fixed-Parameter Tractable} (FPT) if it can be solved in $f(k)|x|^{\bO(1)}$ time for an arbitrary computable function $f\colon \mathbb{N}\to\mathbb{N}$. 
According to standard complexity-theoretic assumptions, a problem is not in FPT if it is shown to be \W[1]-hard. This is achieved through a \emph{parameterized reduction} from another \W[1]-hard problem, a reduction, achieved in polynomial time, that also guarantees that the size of the considered parameter is preserved.

A \textit{kernelization algorithm} is a polynomial-time algorithm that takes as input an instance $(I,k)$ of a problem and outputs an \textit{equivalent instance} $(I',k')$ of the same problem such that the size of $(I',k')$ is bounded by some computable function $f(k)$. The problem is said to admit an $f(k)$ sized kernel, and if $f(k)$ is polynomial, then the problem is said to admit a polynomial kernel. It is known that a problem is FPT~if and only if it admits a kernel.

Finally, the \emph{lower bounds} we present are based on the so-called \textsc{Exponential Time Hypothesis} (ETH for short)~\cite{IP01}, a weaker version of which states that $3$-\textsc{SAT} cannot be solved in time $2^{o(n+m)}$, for $n$ and $m$ being the number of variables and clauses of the input formula respectively.

We refer the interested reader to classical monographs~\cite{CyganFKLMPPS15,Niedermeier06,FlumG06,DowneyF13,FLSZ19} for a more comprehensive introduction to this topic.

\subsection{Structural parameters}\label{subsec:structural-param}
Let $G=(V,E)$ be a graph. A set $U\subseteq V$ is a \emph{vertex cover} of $G$ if for every edge $e\in E$ it holds that $U\cap e \not= \emptyset$. The \emph{vertex cover number} of $G$, denoted $\vc(G)$, is the minimum size of a vertex cover of $G$.

A \emph{tree-decomposition} of $G$ is a pair $(T,\mathcal{B})$, where~$T$ is a tree, $\mathcal{B}$ is a family of sets assigning to each node $t$ of $T$ its \emph{bag} $B_t\subseteq V$, and the following conditions hold:
\begin{itemize}
	\item for every edge $\{u,v\}\in E(G)$, there is a node $t\in V(T)$ such that $u,v\in B_t$ and
	\item for every vertex $v\in V$, the set of nodes $t$ with $v\in B_t$ induces a connected subtree of $T$.
\end{itemize}
The \emph{width} of a tree-decomposition $(T,\mathcal{B})$ is $\max_{t\in V(T)} |B_t|-1$, and the treewidth $\tw(G)$ of a graph $G$ is the minimum width of a tree-decomposition of $G$.
It is well known that computing a tree-decomposition of minimum width is fixed-parameter tractable when parameterized by the treewidth~\cite{Kloks94,Bodlaender96}, and even more efficient algorithms exist for obtaining near-optimal tree-decompositions~\cite{KorhonenL23}.

A tree-decomposition $(T,\mathcal{B})$ is \emph{nice} if every node $t\in V(T)$ is exactly of one of the following four types:

\noindent\textbf{Leaf:} $t$ is a leaf of $T$ and $|B_t|=0$.    
 
\noindent\textbf{Introduce:} $t$ has a unique child $c$ and there exists $v\in V$ such that $B_t=B_{c}\cup \{v\}$.
    
\noindent\textbf{Forget:} $t$ has a unique child $c$ and there exists $v\in V$ such that $B_{c}=B_t\cup \{v\}$.
    
\noindent\textbf{Join:} $t$ has exactly two children $c_1,c_2$ and $B_t=B_{c_1}=B_{c_2}$.

Every graph $G=(V,E)$ admits a nice tree-decomposition that has width equal to $\tw(G)$~\cite{B98}.

The \emph{tree-depth} of $G$ can be defined recursively: if $|V|=1$ then $G$ has tree-depth $1$. Then, $G$ has tree-depth $k$ if there exists a vertex $v\in V$ such that every connected component of $G[V\setminus\{v\}]$ has tree-depth at most $k-1$.

The graph $G$ has \emph{vertex integrity} $k$ if there exists a set $U \subseteq V$ such that $|U| = k' \le k$ and all connected components of $G[V\setminus U]$ are of order at most $k - k'$. We can find such a set in FPT-time parameterized by $k$~\cite{DDH16}. 

A set $S$ is a \emph{twin-cover}~\cite{G11} of $G$ if $V$ can be partitioned into the sets $S,V_1,\dots,V_p$, such that for every $i\in [p]$, all the vertices of $V_i$ are twins. The size of a minimum twin-cover of $G$ is the \emph{twin-cover number} of $G$. 

Let $A$ and $B$ be two parameters of the same graph. We will write $A\leq_f B$ to denote that the parameter $A$ is upperly bounded by a function of parameter $B$. Let $G$ be a graph with treewidth $\tw$, vertex cover number $\vc$, tree-depth $\textrm{td}$, twin-cover number $\textrm{twc}$ and vertex integrity $\textrm{vi}$. We have that that $\textrm{twc}\leq_f\vc$. Moreover,  $\tw\leq_f\textrm{td}\leq_f \textrm{vi}\leq_f\vc$, but $\textrm{twc}$ is incomparable to $\tw$.

\section{Bounded Tree-width or Vertex Cover Number}

This section includes both the positive and negative results we provide for graphs of bounded tree-width or bounded vertex cover number. 

\subsection{FPT Algorithm parameterized by $\tw+\mathcal{C}$}
We begin with provided an FPT algorithm for \CCFw parameterized by $\tw+\mathcal{C}$, where $\tw$ is the treewidth of the input graph. Before we proceed to the main theorem of this section, allow us to briefly comment upon the \textsc{Max Utilitarian} ASHG (\textsc{MU} for short). Simply put, \textsc{MU} is a restriction of $\CCFw$ where $\C=n$ (i.e., the size of the coalitions is unbounded). The authors of~\cite{HKMO19} provide an FPT algorithm for \textsc{MU} parameterized by the treewidth of the input graph. We stress however that, as \textsc{MU} is a special case of $\CCFw$, the above FPT algorithm cannot be used to deal with our problem in the general case. We have the following result.

\thmTwCFPT*

\ifshort
\begin{sketch}
Assuming that we have a \emph{nice tree decomposition} $\mathcal{T}$ of the graph $G$ rooted at a node $r$, we are going to perform dynamic programming on the nodes of $\mathcal{T}$. For a node $t$ of $\mathcal{T}$, we denote by $B_t$ the bag of this node and by $B_t^{\downarrow}$ the set of vertices of the graph that appears in the bags of the nodes of the subtree with $t$ as a root. Observe that $B_t\subseteq B_t^{\downarrow}$.


For all nodes $t$ of $\mathcal{T}$ we will create
all the $\C$-partitions of $G[B_t^{\downarrow}]$ that are needed in order to find an optimal $\C$-partition; this will be achieved by storing only $(\tw\C)^{\bO{(\tw)}}$ $\C$-partitions for each bag.
We first define types of $\C$-partitions of
$G[B_t^{\downarrow}]$ based on their intersection with $B_t$ and the size of their sets.
In particular, we define a \emph{coloring function} $Col: B_t \rightarrow [tw+1]$ and let $S$ be a table of size $|tw+1|$ such that $0\le S[i]\le \C$ for all $i \in [\tw+1]$. We will say that a $\C$-partition $\Par = \{C_1,\ldots, C_p\}$ is of \emph{type} $(Col,S)_t$ if:
\begin{itemize}
    \item $\Par$ is a $\C$-partition of $G[B_t^{\downarrow}]$,
    \item for any $i \le \tw+1$ and $u \in B_t$, $Col(u) = i$ if and only if $u \in C_i \cap B_t $ and
    \item $S[i]=|C_i|$ for all $i\in[tw+1]$.
\end{itemize} 

\noindent Intuitively, for any $\C$-partition $\Par$ of type $(Col,S)_t$, the function $Col$ describes the way that $\Par$ partitions the set $B_t$. Also, the table $S$ gives us the sizes of the sets of $\Par$ that intersect with $B_t$.

Finally, for any node $t$, a $\C$-partition of type $(Col,S)_t$ will be called \textit{important} if it has a value greater or equal to the value of any other $\C$-partition of the same type.
Notice that any optimal $\C$-partition of the given graph is also an important $\C$-partition of $r$.
Therefore, to compute an optimal $\C$-partition of $G$, it suffices to find an important $\C$-partition of maximum value among the all important $\C$-partitions of $r$.

We now present the information we keep for each node.
Let $t$ be a node of $\mathcal{T}$, $Col: B_t \rightarrow [tw+1]$ be a function and $S$ be a table of size $|tw+1|$ such that $0\le S[i]\le \C$ for all $i \in [\tw+1]$. If there exists an important $\C$-partition of type $(Col,S)_t$, then we store a tuple $(Col,S,W,\Par)$ for $t$, where $\Par$ is an important $\C$-partition of type $(Col,S)_t$ and $W$ is its value. Observe that $W$ is the value of a partition of the whole subgraph induced by the vertices belonging to $B_t^{\downarrow}$.

We now explain how to deal with each kind of node of $\mathcal{T}$.

\noindent \textbf{Leaf Nodes.}
Since the leaf nodes contain no vertices, we do not need to keep any non-trivial coloring. Also, all the positions
of the tables $S$ are set to $0$.
We keep a $\C$-partition
$\Par=\{C_1,\ldots, C_{\tw+1}\}$ where $C_i=\emptyset$ for all $i \in [\tw+1]$.

\begin{figure}[!t]
\centering
\includegraphics[scale=0.7]{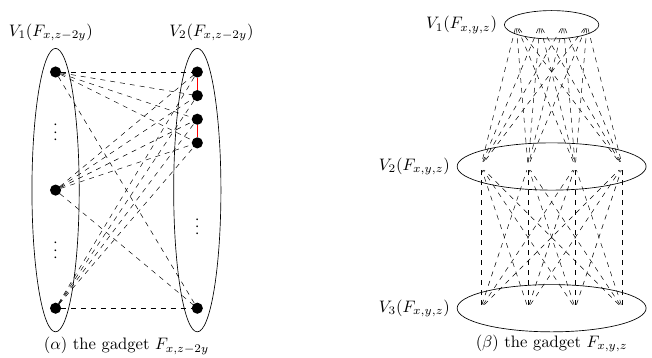}
\caption{The gadgets used in the proof of Theorem~\ref{thm:hard-td}}\label{fig:hard-td-gadgets}
\end{figure}

\noindent \textbf{Introduce Nodes.}
Let $t$ be an introduce node with $c$ being its child node and $u$ be the newly introduced vertex.
For each tuple $(Col,S,W,\Par)$ of $c$, we create at most $\tw+1$ tuples for $t$.
For each color $i \in [\tw+1]$ we consider two cases: either $0\le S[i]<\C$ or $S[i]=\C$.
If $0\le S[i]<\C$, then we set $Col(u)=i$, increase $S[i]$ by one, extend the $\C$-partition $\Par$ by adding $u$ into the set $C_{i}$ and increase $W$ by $\sum_{uv\in E, v \in C_i} w(uv)$.
If $S[i]=\C$ then we cannot color $u$ with the color $i$ as the corresponding set is already of size $\C$.

\noindent \textbf{Forget Nodes.}
Let $t$ be an forget node, with $c$ being its child node and $u$ be the newly introduced vertex.
For each tuple $(Col,S,W,\Par)$ of $c$ we create one tuple $(Col',S',W',\Par')$ for $t$.
Let $Col(u)=i$. We consider two cases: either $C_i\cap B_t = \emptyset$ or not.
In the former, we have that the color $i$ does not appear on any vertex of $B_c\setminus\{u\}=B_t$.
Therefore, we are free to reuse this color. To do so, we set $S'[i]=0$, and we modify $\Par$.
In particular, if $\Par = \{C_1, \ldots, C_k\}$, we create a new $\C$-partition $\Par'=\{C_1'\ldots,C_{k+1}'\}$ where 
$C_j' = C_j$ for all $j \in [k]\setminus \{i\}$, $C_i' = \emptyset$ and $C_{k+1}'=C_i$.
Also, we define $Col'$ as the restriction of the function $Col$ to the set $B_t$.
Finally, $W'=W$.
In the latter case, it suffices to restrict $Col$ to the set $B_t$.
We keep all the other information the same.

\noindent \textbf{Join Nodes.}
Let $t$ be a join node, with $c_1$ and $c_2$ being its children nodes.
For any pair of tuples $(Col_1,S_1,W_1,\Par_1)$ and $(Col_2,S_2,W_2,\Par_2)$
of $c_1$ and $c_2$ respectively, we will create a tuple $(Col,S,W,\Par)$ for $t$ if:
\begin{itemize}
    \item $Col_1(u) = Col_2(u)$ for all $u \in B_t$ and
    \item $S_1[i] + S_2[i] - |C_i \cap B_t| \le \C$ for all $i \in [\tw+1]$,
\end{itemize}
where $C_i$ is the $i^{th}$ set of $\Par_1$. The choice of $\Par_1$ here is arbitrary because of the first condition. Indeed, the first condition guarantees that $\Par_1$ and $\Par_2$ ``agree'' on the vertices of $B_t$. That is, the vertices of $B_t$ are partitioned in the same sets according to $\Par_1$ and $\Par_2$. The second condition guarantees that the sets created for $\Par$ are of size at most $\C$.
The tuple $(Col,S,W,\Par)$ is created as follows. We set:
\begin{itemize}
    \item $Col(u) = Col_1(u)$ for all $u \in B_t$,
    \item $S[i] = S_1[i] + S_2[i]- |C_i \cap B_t|$ for all $i \in [\tw+1]$, and
    \item $W = W_1 +W_2 - \sum_{uv \in E(G[B_t]), Col(u)=Col(v)} w(uv)$.
\end{itemize}
Once more, $C_i$ is chosen w.l.o.g. to be the $i^{th}$ set of $\Par_1$.
Finally, let $\Par_1 = \{C_1^1,\ldots ,C_{p}^1 \}$ and $\Par_2 = \{C_1^2,\ldots ,C_{p'}^2 \}$;
we create the $\C$-partition $\Par =\{C_1\ldots, C_{p+p' - \tw - 1} \}$ as follows.
For any $i \in [\tw+1]$, set $C_i=C_i^1 \cup C_i^2$. For any $i \in [p] \setminus [\tw+1]$, set $C_i = C_i^1$.
Last, for any $i \in [p']\setminus[\tw+1]$, set $C_{p+i} = C_i^2$.

This finishes the description of our algorithm. It remains to compute its running time.

First, we calculate the number of different types of $\C$-partitions for a node $t$.
We have at most $(\tw+1)^{\tw+1}$ different functions $Col$ and $(\C+1)^{\tw+1}$ different tables $S$.
Thus, we have $(\tw \C)^{\bO(\tw)}$ different types for each node. Since we keep only one tuple per type, we are storing $(\tw \C)^{\bO(\tw)}$ tuples for each node of $\mathcal{T}$.
Moreover, for the leaf nodes, we need to create just one tuple.
For the introduce and forget nodes, we need to consider each tuple of their children once.
Therefore, we can compute all tuples for these nodes in time $(\tw \C)^{\bO(\tw)}$.
For the join nodes, in the worst case, we need to consider all pairs of tuples of their children that share the same coloring function. This still does not result in more than $(\tw \C)^{\bO(\tw)}$ combinations.
Finally, all the other calculations remain polynomial to the number of vertices.
\end{sketch}

\fi

\iflong
\begin{proof}
As the techniques we are going to use are standard, we are sketching some of the introductory details. For more details on tree decompositions (definition and terminology), see \cite{DF13}. Assuming that we have a \emph{nice tree decomposition} $\mathcal{T}$ of the graph $G$ rooted at a node $r$, we are going to perform dynamic programming on the nodes of $\mathcal{T}$. For a node $t$ of $\mathcal{T}$, we denote by $B_t$ the bag of this node and by $B_t^{\downarrow}$ the set of vertices of the graph that appears in the bags of the nodes of the subtree with $t$ as a root. Observe that $B_t\subseteq B_t^{\downarrow}$.

In order to simplify some parts of the proof, we assume that the $\C$-partitions we look into are allowed to include empty sets.
In particular, whenever we consider a $\C$-partition $\Par = \{C_1,\ldots,C_p\}$ of a graph $G[B_t^{\downarrow}]$,
we assume that it is of the following form:
\begin{itemize}
    \item $p\ge \tw+1$,
    \item for any set $C_j \in \Par$, if $j \in [\tw+1]$ then either $C_j = \emptyset$ or $C_j\cap B_t \neq \emptyset$ and
    \item for any set $C_j \in \Par$, if $j > \tw+1$ then  $C_j \neq \emptyset$ and $C_j \cap B_t = \emptyset$.
\end{itemize}
Note that any $\C$-partition can be made to fit such a form without affecting its value.
Also, for any node $t$ of the tree decomposition and any $\C$-partition of $G[B_t^{\downarrow}]$,
no more than $\tw +1$ sets of the $\C$-partition can intersect with $B_t$. Thus, we do not need to store more sets of $\Par$ intersecting with $B_t$.

For all nodes $t$ of the tree decomposition, we will create
all the $\C$-partitions of $G[B_t^{\downarrow}]$ that are needed in order to find an optimal $\C$-partition; this will be achieved by storing only $(\tw\C)^{\bO{(\tw)}}$ $\C$-partitions for each bag.
In order to decide which $\C$-partitions we need to keep, we first define types of $\C$-partitions of
$G[B_t^{\downarrow}]$ based on their intersection with $B_t$ and the size of their sets.
In particular, we define a \emph{coloring function} $Col: B_t \rightarrow [tw+1]$ and a table $S$ of size $|tw+1|$ such that $0\le S[i]\le \C$ for all $i \in [\tw+1]$. We will say that a $\C$-partition $\Par = \{C_1,\ldots, C_p\}$ is of \emph{type} $(Col,S)_t$ if:
\begin{itemize}
    \item $\Par$ is a $\C$-partition of $G[B_t^{\downarrow}]$,
    \item for any $i \le \tw+1$ and $u \in B_t$, we have that $Col(u) = i$ if and only if $u \in C_i \cap B_t $ and
    \item $S[i]=|C_i|$ for all $i\in[tw+1]$.
\end{itemize} 

For any $\C$-partition $\Par$ of type $(Col,S)_t$, the function $Col$ describes the way that $\Par$ partitions the set $B_t$. Also, the table $S$ gives us the sizes of the sets of $\Par$ that intersect with $B_t$.

Finally, for any node $t$, a $\C$-partition of type $(Col,S)_t$ will be called \textit{important} if it has value greater or equal to the value of any other $\C$-partition of the same type.
Notice that any optimal $\C$-partition of the given graph is also an important $\C$-partition of the root of the tree decomposition.
Therefore, to compute an optimal $\C$-partition of $G$, it suffices to find an important $\C$-partition of maximum value among the all important $\C$-partitions of the root of the given tree decomposition of $G$.

We now present the information we will keep for each node.
Let $t$ be a node of the tree decomposition, $Col: B_t \rightarrow [tw+1]$ be a function and $S$ be a table of size $|tw+1|$ such that $0\le S[i]\le \C$ for all $i \in [\tw+1]$. If there exists an important $\C$-partition of type $(Col,S)_t$, then we store a tuple $(Col,S,W,\Par)$ for $t$, where $\Par$ is an important $\C$-partition of type $(Col,S)_t$ and $W$ is its value. Observe that $W$ is the value of a partition of the whole subgraph induced by the vertices belonging to $B_t^{\downarrow}$.

We now explain how to deal with each kind of node of the nice tree decomposition.

\smallskip 

\noindent \textbf{Leaf Nodes.}
Since the leaf nodes contain no vertices, we do not need to keep any non-trivial coloring. Also, all the positions
of the tables $S$ are equal to $0$.
Finally, we keep a $\C$-partition
$\Par=\{C_1,\ldots, C_{\tw+1}\}$ where $C_i=\emptyset$ for all $i \in [\tw+1]$.

\smallskip 

\noindent \textbf{Introduce Nodes.}
Let $t$ be an introduce node with $c$ being its child node and $u$ be the newly introduced vertex.
We will use the tuples we have computed for $c$ in order to build one important $\C$-partition for each type of $\C$-partition that exists for $t$.
For each tuple $(Col,S,W,\Par)$ of $c$, we create at most $\tw+1$ tuples for $t$ as follows.
For each color $i \in [\tw+1]$ we consider two cases: either $0\le S[i]<\C$ or $S[i]=\C$.
If $0\le S[i]<\C$, then we set $Col(u)=i$, increase $S[i]$ by one, extend the $\C$-partition $\Par$ by adding $u$ into the set $C_{i}$ and increase $W$ by $\sum_{uv\in E, v \in C_i} w(uv)$.
If $S[i]=\C$ then we cannot color $u$ with the color $i$ as the corresponding set is already of size $\C$.

First, we need to prove that, this way, we create at least one important $\C$-partition for $t$ for each type of $\C$-partition of $G[B_t^{\downarrow}]$. Assume that for a type $(Col, S)_t$ there exists an important $\C$- partition $\Par= \{C_1,\ldots, C_p\}$ of $B_t^{\downarrow}$.

Let $\Par_c $ be the $\C$-partition we defined by the restriction of $\Par$ on the vertex set $B_c^{\downarrow}$.
That is, $\Par_c = \{C_1^c, \ldots, C_p^c \} $ where $ C_i^c = C_i \cap B_c^{\downarrow} $ for all $i \in [p] $. 
Notice that, since $c$ is the child of an introduce node, there exists a $k \in [\ell]$ such that
$C_k^c = C_k \setminus \{u\}$ and $C_i^c = C_i$ for all $i \in [p]\setminus\{k\}$. Also, note that $C_k^c$ may be empty.
Since $\Par$ is a $\C$-partition of $G[B_t^{\downarrow}]$, we have that $\Par_c$ is a $\C$-partition of $G[B_c^{\downarrow}]$.
Furthermore, let $Col':B_c \rightarrow [\tw+1]$ such that $Col'(u) = Col(u)$ for all $u \in B_c$ and $S'$ be a table where
$S'[i] = S[i]$ for all $i \in [\tw+1] \setminus k$ and $S'[k] = S[k]-1$. Observe that $\Par_c$ is of type $(Col',S')_c$.

Since $\Par_c$ is of type $(Col',S')_c$,
we know that we have stored a tuple $(Col',S',W',\Par')$ for $c$, where $\Par' = \{C_1', \ldots, C_{p'}' \} $ is an important $\C$-partition of $G[B_c^{\downarrow}]$. Note that $\Par'$ is not necessarily the same as $\Par_c$, but both of these $\C$-partitions are of the same type. While constructing the tuples of $t$, at some point the algorithm will consider the tuple $(Col',S',W',\Par')$. At this stage, the algorithm will add the vertex $u$ on any set of $\Par'$ of size at most $\C-1$, creating a different tuple for each option. These options include the set colored by $k$; let $(Col_t,S_t,W_t,\Par_t)$ be the corresponding tuple, where $\Par_t=\{C^t_1,\dots,C^t_{p'}\}$. Observe that in this case, $u$ is colored $k$ (\emph{i.e.} $Col_t (u) = k = Col (u)$), $S'[k]$ is increase by one (\emph{i.e.} $S_t[k]=S'[k]+1 = S[k]$) and $u$ is added to $C'_k$ (\emph{i.e.} $C^t_k=C'_k\cup \{u\}$).
Notice that $Col'(v) = Col(v)$ for all $v \in B_t$ and $S'[i]=S[i]$ for all $i \in [\tw+1]$.
Therefore, it suffices to show that $\Par_t$ is also an important $\C$-partition of $G[B_t^{\downarrow}]$. Indeed, this would indicate that $\val(\Par)=\val(\Par_t)$, since $\Par$ and $\Par_t$ would both be important partitions of the same type.

On the one hand, we have that:
\[ 
\val(\Par) = \val(\Par_c) + \sum_{uv\in E, v \in C_k} w(uv) =  
\]
\[
= \val(\Par_c) + \sum_{uv\in E, v \in B_t\text{ and } Col(v)=k} w(uv)
\] 

On the other hand, we have that:
\[
\val(\Par_t) = \val(\Par') + \sum_{uv\in E, v \in C_k'} w(uv) =
\]
\[
=  W' + \sum_{uv\in E, v \in B_t\text{ and } Col'(v)=k} w(uv)
\]
Since $Col(v) = Col'(v)$ for all $v \in B_t$, we have that the two above sums are equal.
Therefore we need to compare $W'$ with $\val(\Par_c)$. Note that $\Par_c$ and $\Par'$ are both $\C$-partitions of
$G[B_c^{\downarrow}]$ of the same type. Thus, $W' = \val(\Par') \ge \val(\Par_c)$.
It follows that $ \val(\Par) \le \val(\Par_t)$, and since $\Par$ is important, we have that $ \val(\Par_t) = \val(\Par)$ and that $\Par_t$ is also important.

\smallskip 

\noindent \textbf{Forget Nodes.}
Let $t$ be an forget node, with $c$ being its child node and $u$ be the forgotten vertex.
We will use the tuples we have computed for $c$ in order to build one important $\C$-partition for each type of $\C$-partition that exists for $t$.
For each tuple $(Col,S,W,\Par)$ of $c$ we create one tuple $(Col',S',W',\Par')$ for $t$ as follows.
Let $Col(u)=i$. We consider two cases: either $C_i\cap B_t = \emptyset$ or not.
In the former, we have that the color $i$ does not appear on any vertex of $B_c\setminus\{u\}=B_t$.
Therefore, we are free to reuse this color. To do so, we set $S'[i]=0$ and we modify $\Par$.
In particular, if $\Par = \{C_1, \ldots, C_k\}$, we create a new $\C$-partition $\Par'=\{C_1'\ldots,C_{k+1}'\}$ where 
$C_j' = C_j$ for all $j \in [k]\setminus \{i\}$, $C_i' = \emptyset$ and $C_{k+1}'=C_i$.
Also, we define $Col'$ as the restriction of the function $Col$ to the set $B_t$.
Finally, $W'=W$.
In the latter case, it suffices to restrict $Col$ to the set $B_t$.
We keep all the other information the same.

We will now prove that, for any type of $\C$-partition of $t$, if there exists a $\C$-partition of that type, we have created an important $\C$-partition of that type.
Assume that for a type $(Col,S)_t$ there exists an important $\C$-partition $\Par =\{ C_1,\ldots,C_p\}$ of $B_t$ of value $W$.
We consider two cases: either $u \in C_\ell$ for some $\ell \le \tw+1$ or $u \in C_\ell$ for some $\ell > \tw+1$.

\smallskip 

\noindent\textbf{Case 1: $\boldsymbol{u \in C_\ell}$ for some $\boldsymbol{\ell \le \tw+1}$.} In this case, $C_\ell \cap B_t \neq \emptyset$. This follows from the assumption that any $\C$-partition $\Par = \{C_1,\ldots,C_p\}$ we consider is such that for any set $C_j \in \Par$, if $j \in [\tw+1]$ then either $C_j = \emptyset$ or $C_j\cap B_t \neq \emptyset$ and because $\{v \mid v \in B_c\setminus \{u\} \text{ and } Col(v) = \ell\} \neq \emptyset$. Let $Col_c: B_c \rightarrow [\tw+1]$ be such that $Col_c(u)=\ell$ and $Col_c(v) = Col(v)$ for all $v\in B_t$.
Notice that $\Par$ is of type $(Col_c, S)_c$. Let $(Col_c, S, W',\Par')$ be the tuple that is stored in $c$ for the $\C$-partition $\Par'=\{C_1',\ldots,C_{p'}'\}$ of type $(Col_c,S)_c$.

At some point while creating the tuples of $t$, the tuple $(Col_c, S, W',\Par')$ was considered.
Let $(Col_c', S', W', \Par')$ be the tuple that was created at that step. Notice that since $Col_c$ is an extension of $Col$ to the set $B_c$ and $C_\ell \cap B_t = \{v \in B_t\mid Col(v)=\ell \}\cap B_t \neq \emptyset$, we have that $\{v \in B_t\mid Col_c(v)=\ell \}\cap B_t \neq \emptyset$. Therefore, $\{v \in B_t\mid Col_c'(v)=\ell \}\cap B_t = \{v \in B_t\mid Col_c(v)=\ell \}\cap B_t \neq \emptyset$. It follows from the construction of $(Col_c', S', W',\Par')$ that $Col'_c (v) = Col(v)$ for all $v \in B_t$.
Also, since $\{v \in B_t\mid Col_c'(v)=\ell \}\cap B_t \neq \emptyset$, the vertex $u$ was not the only vertex colored with $\ell$. Therefore, $S'$ is the same as $S$.
This gives us that $\Par'$ and $\Par$ are of the same type in $t$. That is, $(Col_c', S')_t =(Col, S)_t$ and we have stored a tuple for this type.

It remains to show that $\Par'_t$ is an important partition of its type in $t$. This is indeed the case as $\Par$ and $\Par'$ have the same type in $c$ and $\Par'$ is an important partition of this type in $c$. 
Since the value of the two partitions does not change in $t$ and they remain of the same type, we have that $\Par'_t$ is an important partition of its type in $t$.

\smallskip 

\noindent\textbf{Case 2: $\boldsymbol{u \in C_\ell}$ for some $\boldsymbol{\ell > \tw+1}$.} In this case we have that $C_\ell \cap B_t = \emptyset$ and $C_\ell \cap B_c = \{u\}$. Notice that, at least one of the $C_i$s, $i \in [\tw +1]$, must be empty.
Indeed, since $C_i\cap B_c = C_i\cap B_t \neq \emptyset$, for all $i \in [\tw+1]$, we have $\tw +2$ sets intersecting $B_c$ (including $C_\ell$). This is a contradiction as these sets must be disjoint and $|B_c|\leq \tw+1$.

First, we need to modify the partition $\Par$ so that it respects the second item of the assumptions we have made for the $\C$-partitions in $c$.
To do so, select any $k \in [\tw+1]$ such that $C_k = \emptyset$ and set $C_k=C_i$.
Then, set $C_k =C_{k+1}$, for all $k \in [p-1] \setminus [i-1]$, and remove $C_p$.
Let $\Par_c =\{C_{c,1}, \ldots, C_{c,p-1}\}$ be the resulting $\C$-partition of $c$.
We define $Col_c:B_c\rightarrow [\tw+1]$ such that, for all $v \in B_c$, we have that $Col_c(v) = i$ if and only if
$v \in C_{c,i}$. Notice that $Col$ is the restriction of $Col_c$ on the vertex set $B_t$.
Also, we define $S_c$ to be the table of size $\tw+1$ such that $S_c[i] = |C_{c,i}|$ for all $i \in \tw+1$. Notice that for all $i \in [\tw+1]\setminus \{k\}$, we have $S[i]= S_c[i]$ and $S_c[k] \neq 0$ and $S[k] = 0$.

Observe that $\Par$ is of type $(Col, S)_t$ and $\Par_c$ is of type $(Col_c, S_c)_c$.
Therefore, let $(Col_c, S_c,W',\Par')$ be the tuple we have stored in $c$, where $\Par'$ is an important partition of type $(Col_c, S_c)_c$.
At some point while constructing the tuples of $t$, we consider the tuple $(Col_c, S_c,W',\Par')$ and create
a tuple $(Col_t, S_t,W',\Par_t)$ for $t$. We claim that $\Par_t$ is of the same type as $\Par$ and that $\Par_t$ is an important partition of that type.
Notice that $u$ is the only vertex of $B_c$ such that $Col_c(u)=k$. 
It follows that $(Col_t, S_t,W',\Par_t)$ was created by setting:
\begin{itemize}
    \item $Col_t$ to be the restriction of $Col_c$ on the set $B_t$,
    \item $S_t[k]=0$ and $S_t[i] = S_c[i]$, for $i \in [\tw+1]\setminus \{k\}$ and
    \item we modify the $\Par_c$ following the steps described by the algorithm.
\end{itemize}
Notice that $\Par'$ and $\Par_t$ are the same $\C$-partition, presented in a different way.
By the construction of $Col'_t$ and $S'_t$, we have that $(Col_t, S_t)_t$ is the same as $(Col, S)_t$. It follows that there exists a tuple $(Col,S,W',\Par')$ stored in $t$, where $\Par'$ is of type $(Col,S)_t$.

It remains to show that $\Par_t$ is an important partition of its type.
Notice that $\Par$ and $\Par_c$ are the same $\C$-partition. Therefore, they have the same value.
The same holds for $\Par'$ and $\Par_t$.
Finally, since $\Par'$ and $\Par_c$ have the same type in $c$ and $\Par'$ is an important partition, we have that $\val (\Par')\geq \val (\Par_c)$. So, $\val (\Par_t) = \val (\Par')\geq \val (\Par_c)= \val (\Par)$, from which follows that $\Par_t$ is also an important partition.

\smallskip 

\noindent \textbf{Join Nodes.}
Let $t$ be a join node, with $c_1$ and $c_2$ being its children nodes.
We will use the tuples we have computed for $c_1$ and $c_2$ in order to build one important $\C$-partition for each type of $\C$-partition that exists for $t$.
For any pair of tuples $(Col_1,S_1,W_1,\Par_1)$ and $(Col_2,S_2,W_2,\Par_2)$,
of $c_1$ and $c_2$ respectively, we will create a tuple $(Col,S,W,\Par)$ for $t$ if:
\begin{itemize}
    \item $Col_1(u) = Col_2(u)$ for all $u \in B_t$ (which is the same as $B_{c_1}$ and $B_{c_2}$),
    \item for all $i \in [\tw+1]$, we have $S_1[i] + S_2[i] - |C_i \cap B_t| \le \C$,
\end{itemize}
where $C_i$ is the $i^{th}$ set of $\Par_1$. Note that the choice of $\Par_1$ here is arbitrary because of the first condition. Indeed, the first condition guarantees that $\Par_1$ and $\Par_2$ ``agree'' on the vertices of $B_t$. That is, the vertices of $B_t$ are partitioned in the same sets according to $\Par_1$ and $\Par_2$. The second condition guarantees that the sets created for $\Par$ are of size at most $\C$.
The tuple $(Col,S,W,\Par)$ is created as follows. We set:
\begin{itemize}
    \item $Col(u) = Col_1(u)$ for all $u \in B_t$,
    \item $S[i] = S_1[i] + S_2[i]- |C_i \cap B_t|$ for all $i \in [\tw+1]$, and
    \item $W = W_1 +W_2 - \sum_{uv \in E(G[B_t]), Col(u)=Col(v)} w(uv)$.
\end{itemize}
Once more, $C_i$ is chosen w.l.o.g. to be the $i^{th}$ set of $\Par_1$.
We are now ready to define $\Par$.
Let $\Par_1 = \{C_1^1,\ldots ,C_{p}^1 \}$ and $\Par_2 = \{C_1^2,\ldots ,C_{p'}^2 \}$;
we create the $\C$-partition $\Par =\{C_1\ldots, C_{p+p' - \tw - 1} \}$ as follows.
For any $i \in [\tw+1]$, set $C_i=C_i^1 \cup C_i^2$. For any $i \in [p] \setminus [\tw+1]$, set $C_i = C_i^1$.
Last, for any $i \in [p']\setminus[\tw+1]$, set $C_{p+i} = C_i^2$.
This completes the construction of the tuple we keep for $t$, for each pair of tuples that are stored for $c_1$ and $c_2$.

We will now prove that for any type of $\C$-partition of $t$, if there exists a $\C$-partition of that type, we have created an important $\C$-partition of that type. We assume that for a type $(Col,S)_t$ of $t$,
there exists an important $\C$-partition $\Par = \{C_1,\ldots,C_p \}$ of $G[B_{t}^\downarrow]$.
Let $\Par_1 = \{C_1\cap B_{c_1}^\downarrow, \ldots ,C_p\cap B_{c_1}^\downarrow\}$ and $\Par_2= \{C_1\cap B_{c_2}^\downarrow, \ldots ,C_p\cap B_{c_2}^\downarrow\}$. Notice that $\Par_1$ and $\Par_2$ are $\C$-partitions of $G[B_{c_1}^\downarrow]$
and $G[B_{c_2}^\downarrow]$, respectively.
Let $(Col,S_1)_{c_1}$ and $(Col,S_2)_{c_2}$ be the types of $\Par_1$ and $\Par_2$, respectively (recall that, by construction, $Col_1=Col_2=Col$). The existence of $\Par_1$ (respectively $\Par_2$) guarantees that there is a tuple $(Col,S_1,W_1, \Par_1')$ (resp. $(Col,S_2,W_2, \Par_2')$) stored for the node $c_1$ (resp. $c_2$).
By the definition of $\Par_1$ and $\Par_2$, we have that $S[i] = S_1[i] + S_2[i]- |C_i \cap B_t| \le \C$ for all $i \in [\tw+1]$.
It follows that while constructing the tuples of $t$, at some point the algorithm considered the pair of tuples $(Col,S_1,W_1, \Par_1')$ and $(Col,S_2,W_2, \Par_2')$, and created the tuple $(Col, S',W', \Par')$ for $t$. Notice that, by the construction of $S'$, we have that
$S[i] = S'[i]$ for all $i \in [\tw+1]$. Therefore the type $(Col, S')_t$ is the same as $(Col, S)_t$.

It remains to show that $\Par'$ is an important partition of its type. Notice that
$\val(\Par') = W' = W_1 +W_2 - \sum_{uv \in E(G[B_t]), Col(u)=Col(v)} w(uv) \text{ and }\\
  \val(\Par) = \val(\Par_1) +\val(\Par_2) - \sum_{uv \in E(G[B_t]), Col(u)=Col(v)} w(uv)$.
Since $W_1$ is the weight of an important partition of the same type as $\Par_1$ in $c_1$, we have that $\val(\Par_1)\le W_1$.
Also, $W_2$ is the weight of an important partition of the same type as $\Par_2$ in $c_2$. It follows that $\val(\Par_2)\le W_2$.
Overall:
$\val(\Par') = W_1 +W_2 - \sum_{uv \in E(G[B_t]), Col(u)=Col(v)} w(uv) \ge \val(\Par_1) +\val(\Par_2) - \sum_{uv \in E(G[B_t]), Col(u)=Col(v)} w(uv) = \val(\Par) $.

Thus, $\Par'$ is an important partition of its type in $t$. This finishes the description of our algorithm, as well as the proof of its correctness.

\medskip

It remains to compute the running time of our algorithm. First we calculate the number of different types of $\C$-partitions for a node $t$.
We have at most $(\tw+1)^{\tw+1}$ different functions $Col$ and $(\C+1)^{\tw+1}$ different tables $S$.
Therefore, we have $(\tw \C)^{\bO(\tw)}$ different types for each node. Since we are storing one tuple per type, we are storing $(\tw \C)^{\bO(\tw)}$ tuples for each node of the tree decomposition.
Moreover, we create just one tuple for each leaf node. Also, the tuples of the child of each
introduce and forget node are considered once.
Therefore, we can compute all tuples for these nodes in time $(\tw \C)^{\bO(\tw)}$.
As for the join nodes, in the worst case we may need to consider all pairs of tuples of their children that share the same coloring function. This still does not result in more than $(\tw \C)^{\bO(\tw)}$ combinations.
Finally, as all the other calculations remain polynomial to the number of vertices,
the total time that is required is $(\tw \C)^{\bO(\tw)} |V(G)|^{\bO(1)}$.
\end{proof}
\fi

\subsection{Intractability for tree-depth}
In this section, we establish that \CCF is \W[1]-hard when parameterized by the tree-depth of the input graph, complementing our algorithm from \Cref{thm:TwC-FPT}. To this end, we provide a rather involved reduction from \textsc{General Factors}, defined below.

\Pb{\GFF}{A graph $H$ and a list function $L: V(H) \rightarrow \mathcal{P}(\{0,\dots,\Delta(H)\})$ that specifies the available degrees for each vertex $u \in V$.}{Question}
{Does there exists a set $S\subseteq E(H)$ such that $d_{H-S}(u) \in L(u)$ for all $u \in V(H)$?}

\begin{proposition}[\cite{GKSSY12}]\label{P:W-hard}
    \textsc{General Factors}  
    is $\W[1]$-hard even on bipartite graphs when parameterized by the size of the smallest bipartition.
\end{proposition}

To ease the exposition, we will first present the construction of our reduction and a high level idea of the proof that will follow before we proceed with our proof. We will then show that any optimal $\C$-partition of the constructed graph verifies a set of important properties that will be utilized in the reduction.

\iflong

\medskip 

\noindent \textbf{The construction.} 
    Let $(H,L)$ be an instance of the \textsc{General Factors} problem where $H=(V_L,V_R, E)$ is a bipartite graph ($V(H)=V_L\cup V_R$ and $E(H) = E$)
    and $L: V_L\cup V_R \rightarrow \mathcal{P}([|V(H)|])$ gives the list of degrees for each vertex.
    Notice that, normally, $|L(u)| \le d(u)\le |V(H)|$. 
    Nevertheless, we can assume that $|L(u)| = |V(H)|$ as we can allow $L(u)$ to be a multiset. Hereafter, we assume that the size for the smallest bipartition is $m$ and the total number of vertices is $n =|V(H)|$.
    Note that $m \le n/2$. We can also assume that $m>2$ as otherwise we could answer whether $(H,L)$ is a yes-instance of the \textsc{General Factors} problem in polynomial time.

\begin{figure}[!t]
\centering
\includegraphics[scale=1.2]{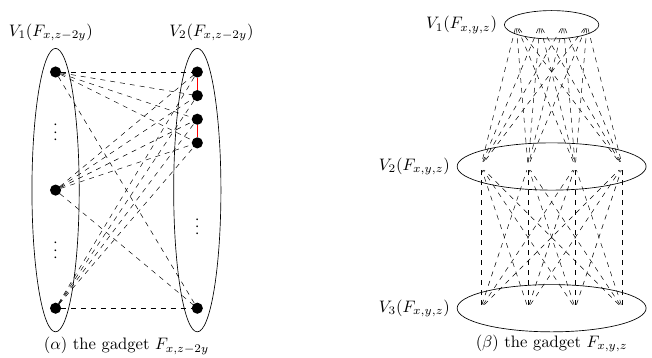}
\caption{The gadgets used in the proof of Theorem~\ref{thm:treedepth}}\label{fig:hard-td-gadgets}
\end{figure}

    Starting from $(H,L)$, we will construct a graph $G$ such that any $\C$-partition of $G$, for $\C=100n^3$, has a value exceeding a threshold if and only if $(H,L)$ is a yes-instance of the \textsc{General Factors} problem.
    We start by carefully setting values so that our reduction works. We define the values
    $ A  = n^2$, $\sizeMissingOfVertexGadget = 5n^2 + 3m +4$ and $\sizeMissingOfEdgeGadget = 2m+5 $ 
    which will be useful for the constructions and calculations that follow. 

    We now describe the two different gadgets denoted by $F_{x,5 A  -2y}$ and $F_{x,\C - y,z}$.
    The $F_{x,5 A  -2y}$ gadget is defined for $4xy < 5 A -2y $. It is constructed as follows (illustrated in Figure~\ref{fig:hard-td-gadgets}(a)):
    \begin{itemize}
        \item We create two independent sets $U$ and $V$ of size $x$ and $5 A  -2y$ respectively,
        \item we add all edges between vertices of $U$ and $V$ and 
        \item we add $2xy$ edges between vertices of $V$ such that the graph induced by the vertices incident to these edges is an induced matching (we have enough vertices because we assumed that $4xy < 5 A -2y $).
    \end{itemize}
    Hereafter, for any gadget $F_{x,5 A  -2y}=F$ we will refer to $U$ as $V_1(F)$ 
    and to $V$ as $V_2(F)$.
    
    The construction of $F_{x,\C- y,z}$ is as follows (illustrated in Figure~\ref{fig:hard-td-gadgets}(b)):
    \begin{itemize}
        \item We create three independent sets $U$, $V$ and $W$ of size $x$, $\C- y$ and $z$ respectively,
        \item we add all edges between vertices of $U$ and $V$ and all edges between vertices of $V$ and $W$,
    \end{itemize}
    Hereafter, for any gadget $F_{x,\C- y, z}=F$ we will refer to $U$ as $V_1(F)$, to $V$ as  $V_2(F)$
    and to $W$ as $V_3(F)$.
    Before we continue, notice that $|E(F_{x,5 A  -2y})| = 5x A $ and
    $|E(F_{x,\C- y, z})| = (x+z)(\C-y)$.
    
    We are now ready to describe the construction of the graph $G$, illustrated in Figure~\ref{fig:hard-td}.
    First, for each vertex $v\in V(H)$, we create a copy $F^v$ of the $F_{4,\C-\sizeMissingOfVertexGadget, 2 m+10}$ gadget; we say that this is a \emph{vertex-gadget}.
    We also fix a set $U(F^v) \subset V_1(F^v)$ such that $|U(F^v))|=2$.  
    Now, for any vertex $v \in V(H)$ and integer $\alpha \in L(v)$, we create a copy $F^{\alpha(v)}$ of the $F_{m+6,5 A  -2\alpha}$ gadget; we say that this is a \emph{list-gadget}. We add all the edges between $V_2(F^{\alpha(v)})$ and $U(F^v)$.
    Recall that we have assumed $|L(v)| = |V(H)|$, for all $v \in V(H)$. So, for each vertex $v$ of $H$, in addition to $F^v$, we have created
    $|V(H)|=n$ gadgets (one for each element in the list).
    Finally for each edge $e = uv \in E(H)$, where $u \in V_L$ and $v \in V_R$, we create a copy $F^e$ of the $F_{1,\C-\sizeMissingOfEdgeGadget, 2 m}$ gadget; we say that this is an \emph{edge-gadget}. Then, we add a set of vertices $V_e = \{w_{L1}^e,w_{L2}^e,w_{R1}^e,w_{R2}^e\}$. We add all the edges between $V_1(F^u)$ and $V_e$, all the edges between $V_1(F^u)$ and $\{w_{L1}^e,w_{L2}^e\}$, all the edges between $V_1(F^v)$ and $\{w_{R1}^e,w_{R2}^e\}$ and the edges $w_{Li}^e,w_{Rj}^e$ for all $i,j \in [2]$ (\emph{i.e.}, $V_e$ induces a $K_{2,2}$). Hereafter, let $V_E=\bigcup_{e \in E(H)}V_e$ and by $U_E=\bigcup_{e \in E(H)}V(F^e)$. This completes the construction of $G$.

\medskip 

\noindent \textbf{High-level idea.} The reduction works for a carefully chosen value for $\C$. Also, each gadget that is added has the number of its vertices carefully tweaked through changing the values of the $x,y$ and $z$. We proceed by showing that in any optimal $\C$-partition of $G$, for every gadget, its vertices belong in the same set of the partition. Moreover, each gadget $F^v$ will be in the same set as exactly one of the $F^{a}$s. Then, we are left with the vertices of $V^{uv}$, which will serve as translators between the two problems. Intuitively, if an edge $uv$ does not appear in the set $S$ of the solution of the \textsc{General Factors} problem, then any optimal $\C$-partition $\Par$ of $G$ will be such that the vertices of $V_{uv}$ will be in the same set as the vertices of $F^{uv}$. In particular, every $\C$-partition $\Par$ of $G$ that has a value exceeding a threshold, will be such that the vertices of $V_{uv}$ will be split to the different sets that contain the vertices of $F^v$ and $F^u$ if and only if the edge $uv$ belongs in the solution $S$ of the \textsc{General Factors} problem.

\medskip 

\noindent \textbf{Properties of optimal $\C$-partitions of $G$.}
    Before we continue let us introduce some notation.  
    Observe that all the vertex-gadgets contain the same number of edges. For every vertex $v\in V(H)$, let $m_v = |E(F)|$, where $F$ is any vertex-gadget.
    Similarly, all edge-gadgets contain the same number of edges. For every edge $e\in E(H)$, let $m_e = |E(F)|$, where $F$ is any edge-gadget.
    Finally, the same holds for the list-gadgets; let $m_\ell = |E(F)|$, where $F$ is any list-gadget. 
    
    Our goal is to show that an optimal $\C$-partition $\Par$ of $G$ has value $v(\Par) = m_v|V(H)|  + m_\ell |V(H)|^2  + m_e |E(H)|   + 10 A  |V(H)| + 8 |E(H)| $
    if and only if $(H,L)$ is a yes-instance of the \textsc{General Factors} problem. 

\begin{figure}[!t]
\centering
\includegraphics[scale=1]{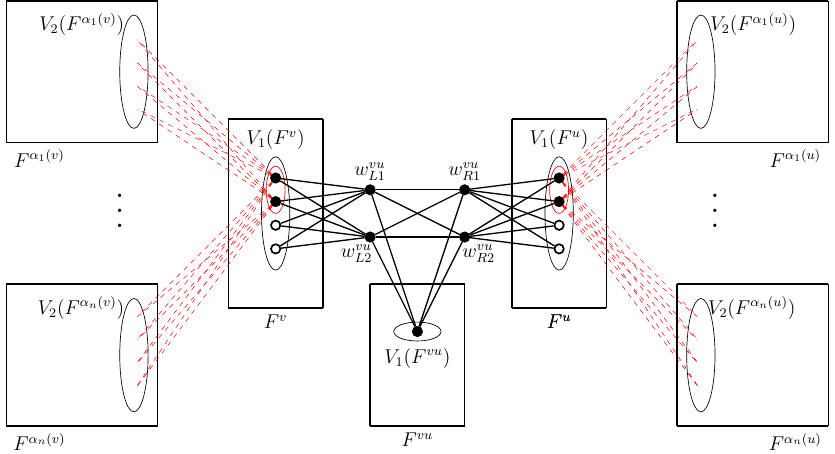}
\caption{The graph $G$ constructed in the proof of Theorem~\ref{thm:treedepth}.}\label{fig:hard-td}
\end{figure}
    
    Assume that $\Par$ is an optimal $\C$-partition of $G$. 
    First, we will show that for every gadget $F$, there exists a $C \in \Par$ such that $V(F)\subseteq C$.
    Then we will prove that for any vertex-gadget $F^v$, there exists one list-gadget $F$ that represents an element of the
    list $L(v)$ (\emph{i.e.} any $u \in U(F^v)$ and $w \in V_2(F)$ are adjacent) and there exists a $C \in \Par$ such that $V(F^v) \cup V(F)\subseteq C$.
    Finally, we will show that in order for $\Par$ to be optimal, \emph{i.e.}, $v(\Par)=m_v|V(H)|  + m_\ell |V(H)|^2  + m_e |E(H)|   + 10 A  |V(H)| + 8 |E(H)| $,
    the vertices of $V_E$ will be partitioned such that:
    \begin{itemize}
        \item the set that includes $V(F^e)$ either includes all the vertices of $V_e$ or none of them and
        \item the set that includes $V(F^v)$ and a gadget $V(F)$, for a list-gadget $F$ representing the value $\alpha \in L(v)$, will also include
        $2\alpha$ vertices from $V_E$.
    \end{itemize} 
    We will show that if both the above conditions hold then $\Par$ is optimal and $(H,L)$ is a yes-instance of the \textsc{General Factors} problem.
    In particular, the edges $E'$ of the solution of the \textsc{General Factors} problem are exactly the edges $e \in E(H)$ such that $F^e$ and $V_e$ are in the same set of $\Par$.

    Let $\Par$ be a $\C$-partition of $G$. For every $C\in \Par$, we can assume that $G[C]$ is connected as otherwise we could consider each connected component of $G[C]$ separately. We start with the following lemma:
    \begin{lemma} \label{treedepth:lemma:vertex-egde:gadgets:in:one:set}
    Let $\Par = \{C_1,\ldots, C_p\}$ be an optimal $\C$-partition of $G$ and $F$ be a vertex or edge-gadget.
    There exists a set $C \in \Par$ such that $C \supseteq V(F)$.
    \end{lemma}

    \begin{proof}
        Assume that this is not true and let $F$ be a vertex or edge-gadget such that $C \cap V(F) \neq V(F)$ for all $C \in \Par$.
        We first show that $\max_{C \in \Par} \{|C\cap V_2(F)|\} = x \ge 2 |V_2(F)| /3$.
        Assume that $\max_{C \in \Par} \{|C\cap V_2(F)|\} = x < 2|V_2(F)| /3$.
        We consider the partition $\Par' = \{V(F), C_1 \setminus V(F), \ldots , C_p \setminus V(F)\}$.
        We will show that $v(\Par) < v(\Par')$. Notice that any edge that is not incident to a vertex of $V(F)$ is either in both sets
        or in none of them. Therefore, we need to consider only edges incident to at least on vertex of $V(F)$.
        Also, since all edges in $E(F)$ are included in $\Par'$ we only need to consider edges incident to $V_1(F)$ 
        (as any other vertex is incident to edges in $E(F)$).

        For any vertex $v \in V_1(F)$, let $d = d(v) - |V_2(F)|$ ($d$ is the same for any $v \in V_1(F)$) and $C \in \Par $ be the set such that $v\in \C$.
        We have that $|C \cap N(v)|\le d +x$, from which follows that $|E(\Par) \setminus E(\Par')| \le 4d$.
        Notice that, regardless of which gadget $F$ and vertex $v \in  V_1(F)$ we consider, we have that $d \le 5n  A  + 2 n < 6n A $ (since $ A =n^2$).
        Indeed, if $F$ is an edge-gadget then $d=4$. Also, if $F$ is a vertex-gadget, then any $v \in V_1(F)$ has at most $5 A $ neighboring vertices in each of the $n$ list-gadgets related to it (if it is in $U(F)$) and at most $2n$ in $V_E$.
        
        We will now calculate $|E(\Par') \setminus E(\Par)|$.
        Consider a $v \in V_1(F) \cup V_3(F)$ and let $C \in \Par$ such that $v \in C$.
        Notice that $|C \cap V_2(F)| = x$. Therefore, we have at least $|V_2(F)| - x$ edges incident to $v$, which belong in $E(\Par') \setminus E(\Par)$.
        Since $ V_1(F)\cup V_3(F) $ is an independent set, it follows that
        $|E(\Par') \setminus E(\Par)| \ge |V_1(F)\cup V_3(F) | ( |V_2(F)| - x )> (2m+4)|V_2(F)| / 3 >|V_2(F)|$.
        Now, in order to show that $v(\Par')> v(\Par)$, it suffices to show that $4 d< |V_2(F)| $.
        This is indeed the case as $|V_2(F)| = \C - \sizeMissingOfVertexGadget = 100n^3 - (5n^2 + 3m +4) > 24 n^3 = 24n  A  > 4 d$. 
        Thus, we can assume that $\max_{C \in \Par} \{|C\cap V_2(F)|\} = x \ge 2|V_2(F)| /3$. 
    
        Let $C \in \Par$ be the set such that $|C\cap V_2(F)| \ge 2|V_2(F)| /3$.
        We will show that $C \cap V_3(F) = V_3(F)$.
        Assume that this is not true and let $v \in V_3(F)$ such that $v \notin C$. 
        Notice that at most $y = |V_2(F)| - x \le |V_2(F)| /3$ edges incident to $v$ are included in $E(\Par)$.
        If $|C|< \C$, then moving $v$ from its set to $C$ increases the number of edges in $E(\Par)$ by $x-y \ge|V_2(F)|/3 $. Therefore, we can assume that
        $|C| = \C$.
        Since $G[C]$ is connected and $|C| = \C$, we have that $C$ must include at least one vertex from $V_1(F)$ and at least one vertex from $N(V_1(F))\setminus V_2(F)$.
        Notice that any vertex $u \in N(V_1(F))\setminus V_2(F)$ has degree at most $m + 10$ (regardless of the value of $m$).
        Therefore, by replacing a vertex $u \in C \cap  ( N(V_1(F))\setminus V_2(F) )$ in $C$ by $v$,
        we increase the number of edges in $E(\Par)$ by at least $x-y - d(u) \ge |V_2(F)|/3 - d(u) > 0$.
        This contradicts the optimality of $\Par$. Thus, we can assume that $C \cap V_3(F) = V_3(F)$.

        We will show that $C \cap V_2(F) = V_2(F)$.
        Assume that there exists a vertex $v \in V_2(F) \setminus C$. 
        Since $C \cap V_3(F) = V_3(F)$, we have that $|N(v) \cap C| \ge |V_3(F)|= 2m+14$ and $N(v) \setminus C \le 4$. 
        If $|C|<\C$, then moving $v$ from its set to $C$ increases the number of edges in $E(\Par)$ (recall that $m>2$).
        Thus we can assume that $|C|=\C$.
        Since $G[C]$ is connected and $|C| = \C$, we have that $C$ must include at least one vertex in $V_1(F)$ and one from $N(V_1(F))\setminus V_2(F)$.
        Notice that any vertex $u \in N(V_1(F))\setminus V_2(F)$ has degree at most $m + 10$.
        Therefore, by replacing a vertex $u \in C \cap  ( N(V_1(F))\setminus V_2(F) )$ in $C$ by $v$,
        we increase the number of edges in $E(\Par)$ by at least $2m +10 - (m +10) =m$.
        This contradicts the optimality of $\Par$. Thus, we can assume that $C \cap V_2(F) = V_2(F)$.

        We now show that $C \cap V_1(F) = V_1(F)$.
        Assume that this is not true and let $v \in V_1(F)$ such that $v \notin C$. 
        Notice that we may have up to $d(v) - |V_2(F)|$ edge in $E(\Par)$ that are incident to $v$.
        If $|C|< \C$, then moving $v$ from its set to $C$ increases the number of edges in $E(\Par)$
        by at least $2|V_2(F)| - d(v)>0$ (since $V_2(F) \subset C$ and $|V_2(F)|=\C - \sizeMissingOfVertexGadget$).
        Thus, we can assume that $|C |= \C$.
        Since $G[C]$ is connected and $|C|=\C$, we have that $C$ must include at least one vertex in $V_1(F)$ and one from $N(V_1(F))\setminus V_2(F)$.
        Any vertex $u \in N(V_1(F))\setminus V_2(F)$ can contribute at most $m + 10$ edges in $E(\Par)$.
        Therefore, by replacing $u$ in $C$ by $v$, we increase the number of edges in $E(\Par)$
        by at least $2|V_2(F)| - d(v) - (m + 10)>0$. This contradicts the optimality of $\Par$, finishing the proof of the lemma.
        \end{proof}

    Next, we will show that the same holds for the list-gadgets. In order to do so, we first need the two following intermediary lemmas.
    \begin{lemma} \label{treedepth:lemma:list:gadgets:3V_2/4:in:one:set}
        Let $\Par = \{C_1\ldots, C_p\}$ be an optimal $\C$-partition of $G$ and $F$ be a list-gadget in $G$.
        There exists a set $C \in \Par$ such that $|C \cap V_2(F)| \ge 3 |V_2(F)|/4 $.
    \end{lemma}
    \begin{proof}
        Assume that this is not true and let $max_{C \in \Par} \{|C \cap  V_2(F)|\} = x <  3 |V_2(F)|/4 $. 
        We create a new partition $\Par' = \{V(F), C_1\setminus V(F),\ldots, C_p\setminus V(F) \}$.
        We will show that $v(\Par) < v(\Par')$. 
        Notice that any edge that is not incident to a vertex of $V_2(F)$ is either in both $\Par$  and $\Par'$ or in neither of them.
        Therefore, we need to consider only the edges that are incident to a vertex of $V_2(F)$.
        Observe that any edge in $G[V(F)]$ is included in $E(\Par')$. Thus, $E(\Par)\setminus E(\Par') \subseteq E(G[V_2(F)\cup U(F^v)] \setminus E(G[V_2(F)])$
        (recall that $N[V_2(F)] \cap V_1(F^v) = U(F^v)$ and $N[V_2(F)] \setminus V_1(F^v) \subseteq V(F)$).
        Since $max_{C \in \Par} \{|C \cap  V_2(F)|\} = x <  3 |V_2(F)|/4 $, we have at most $3  |V_2(F)|/2$ edges of $E(G[V_2(F)\cup V_1(F^v)] \setminus E(G[V_2(F)])$
        in $E(\Par)\setminus E(\Par')$. Thus, $E(\Par)\setminus E(\Par') \le 3  |V_2(F)|/2 $.
        We will now calculate the size of $E(\Par')\setminus E(\Par)$. Since $max_{C \in \Par} \{|C \cap  V_2(F)|\} = x <  3 |V_2(F)|/4 $,
        for each vertex $v \in V_1(F)$ there are at least $ |V_2(F)|/4$ edges incident to $v$ that are included in $E(\Par')\setminus E(\Par)$.
        Therefore, $|E(\Par')\setminus E(\Par)| \ge |V_1(F)| |V_2(F)|/4$.
        Since $|V_1(F)| = m +6 > 6 $ we have that $v(\Par) < v(\Par')$, which contradicts to the optimality of $\Par$.
    \end{proof}

    \begin{lemma} \label{treedepth:lemma:list:gadgets:V_1:in:one:set}
        Let $\Par = \{C_1\ldots, C_p\}$ be an optimal $\C$-partition of $G$ and $F$ a list-gadget in $G$.
        There exists a set $C \in \Par$ such that $|C \cap V_2(F)| \ge 3 |V_2(F)|/4 $ and $V_1(F) \subseteq C$.
    \end{lemma}

    \begin{proof}
        By Lemma~\ref{treedepth:lemma:list:gadgets:3V_2/4:in:one:set}, we have that there exists a $C \in \Par$ such that $|C \cap  V_2(F)|  \ge  3 |V_2(F)|/4 $.
        Assume that there exists a $v \in V_1(F) \setminus C$.
        We can assume that $|C| = \C$, as otherwise we could move $v$ into $C$ which would result in a $\C$-partition with higher value.
        Since $|C|=\C$ and $G[C]$ is connected, we know that $C$ includes vertices from $V_1(F^v)$, where $F^v$ is a vertex-gadget in $G$.
        Also, by Lemma~\ref{treedepth:lemma:vertex-egde:gadgets:in:one:set}, we know that $C \supseteq V(F^v)$. Since $|C|=\C$ and $G[C]$ is connected, we also have a vertex $u \in C \cap N[V_1(F^v)] \setminus ( V_2(F^v) \cup V_2(F) )$. Notice that  $d(u) \le m+10$.
        We claim that replacing $u$ by $v$ in $C$ will result in a $\C$-partition with higher value.
        Indeed, since $d(u) \le m+10$, removing $u$ from $C$ reduces the value of the partition by at most $m+10$.
        Moreover, $v$ has $3 |V_2(F)|/4$ neighbors in $C$. Therefore, moving $v$ into $C$ increases the value of $\Par$ by at least $|V_2(F)|/2$.
        Since $|V_2(F)|/2 > m+10$, this is a contradiction to the optimality of $\Par$.
        Thus $V_1(F) \subseteq C$.
    \end{proof}

    We are now ready to show that the vertices of any list-gadget will belong to the same set in any optimal $\C$-partition of $G$.
    \begin{lemma} \label{treedepth:lemma:list:gadgets:F:in:one:set}
        Let $\Par = \{C_1, \ldots , C_p\}$ be an optimal $\C$-partition of $G$ and $F$ be a list-gadget in $G$.
        There exists a set $C \in \Par$ such that $V(F) \subseteq C$.
    \end{lemma}

    \begin{proof}

        By Lemma~\ref{treedepth:lemma:list:gadgets:V_1:in:one:set} we have that there exists a $C \in \Par$ such that $|C \cap  V_2(F)|  \ge  3 |V_2(F)|/4 $ and $V_1(F) \subseteq C$.
        Assume that there exists a vertex $u \in C\setminus V_2(F)$. We can assume that $|C| = \C$ as otherwise we could include $u$ into $C$
        and this would result in a $\C$-partition with a higher value (as most of the neighbors of $u$ are in $C$).
        Since $|C|=\C$ and $G[C]$ is connected, we know that $C$ includes vertices from $V_1(F^v)$, where $F^v$ is a vertex-gadget in $G$.
        Also, by Lemma~\ref{treedepth:lemma:vertex-egde:gadgets:in:one:set}, we know that $C \supseteq V(F^v)$.

        Since $C \supseteq V(F^v)$, we can conclude that there is no other list-gadget $F'$ in $G$ such that $V_1(F') \cap  C \neq \emptyset$.
        Indeed, since $V_1(F') \cap  C \neq \emptyset$  and by Lemma~\ref{treedepth:lemma:list:gadgets:V_1:in:one:set}, we have that $|C \cap  V_2(F')|  \ge  3 |V_2(F')|/4 $ and, thus, that $|C|> \C$.
        Since $|C|=\C$ and $G[C]$ is connected, we need to include vertices from $N(V_1(F^v)) \setminus (V_2(F^v) \cup V_2(F) ) $ in $C$. 
        Also, since we have concluded that there is no list-gadget $F'$ in $G$ such that $V_1(F') \cap  C \neq \emptyset$,
        any vertex $w \in C$ such that $ w \in   N(V_1(F^v)) \setminus (V_2(F^v) \cup V_2(F))$ has $|N(w) \cap C| \le 6 $.
        We claim that replacing $u$ in $C$ by $v$ will result in a $\C$-partition with a higher value.
        Indeed, since $|N(u) \cap C| \le 6$, removing $u$ from $C$ will reduce the value of the partition by at most $4$.
        Also, since $v$ has at least $d(v)-1$ of its neighbors in $C$, moving it into $C$ increases the value of the partition by at least $d(v) - 1 \ge  m+6+2 - 1 > 6$. This is a contradiction to the optimality of $\Par$.
        Thus, $V_1(F) \subseteq C$.
    \end{proof}

    As we already mentioned, it follows from Lemmas~\ref{treedepth:lemma:vertex-egde:gadgets:in:one:set} and~\ref{treedepth:lemma:list:gadgets:F:in:one:set} that for any optimal $\C$-partition $\Par$ of $G$,
    any set $C\in \Par$ that includes a vertex-gadget $F^v$ can also include at most one list-gadget $F$.
    We will show that any such set $C$ must, actually, include exactly one list-gadget.
    \begin{lemma} \label{treedepth:lemma:vertex:gadgets:together:list:gadgets}
        Let $\Par = \{C_1\ldots, C_p\}$ be an optimal $\C$-partition of $G$ and $F^v$ a vertex-gadget in $G$.
        Let $C \in \Par$ be the set such that $V(F^v) \subseteq C$. 
        There exists a list-gadget $F$ such that $N(V_1(F^v)) \cap V(F) \neq \emptyset$ and $V(F)\subseteq C$.
    \end{lemma}

    \begin{proof}
        By Lemma~\ref{treedepth:lemma:vertex-egde:gadgets:in:one:set} we have that for any vertex-gadget $F^v$, there exists a set $C \in \Par$ such that $V(F^v) \subseteq C$.
        We will show that $C$ also includes a list-gadget $F$ and that $N(V_1(F^v)) \cap V(F) \neq \emptyset$.
        Assume that this is not true, and let $F$ be any list-gadget such that $N(V_1(F^v)) \cap V(F) \neq \emptyset$.
        We can assume that $|C| \ge \C - |V(F)|$ as otherwise we could include $V(F)$ in $C$ and create a $\C$-partition
        of higher value than $\Par$.
        By the size of $F^v$, the assumption that $|C| \ge \C - |V(F)|$, and
        Lemma~\ref{treedepth:lemma:vertex-egde:gadgets:in:one:set}, we have that $C\setminus V(F^v) \subseteq V_E $.
        Let $S = V(F) \cup V(F^v) $ and $\Par' =\{ S, C_1 \setminus S,\ldots, C_p\setminus S \}$. 
        We claim that $v(\Par)< v(\Par')$. 
        We will calculate the values $|E(\Par)\setminus E(\Par')|$ and $|E(\Par')\setminus E(\Par)|$. 
        Notice that the only edges that may belong in $E(\Par)\setminus E(\Par') $ are the edges between $V_1(F^v)$ and $V_E$.
        This means that $|E(\Par)\setminus E(\Par')| \le 8n$ (since there are less than $n$ edges incident to $v$ in $H$ and $8$ edges between $V(F^v)$ and $V_e$, for any $e$ incident to $v$).
        As for $|E(\Par')\setminus E(\Par)|$, since the edges between  $|U(F^v)| $ and $ |V_2(F)|$ do not contribute to $\Par$,
        we have that $|E(\Par')\setminus E(\Par)| \ge |U(F^v)| \cdot |V_2(F)|$.
        Since $|V_2(F)| > 5  A  - 2n > 3n $ (for any list-gadget, and sufficiently large $n$) and $|V_1(F^v)| =4$, we can conclude that $|V_1(F^v)| \cdot |V_2(F)| > 8n$. Therefore, $|E(\Par)\setminus E(\Par')| < |E(\Par')\setminus E(\Par)|$, which contradicts the optimality of $\Par$.
    \end{proof}

    Finally, we will show that any vertex $u \in V_e$ must be in a set that includes either vertices from
    $V(F^e)$ or vertices from $V(F^u) \cup V(F^v)$, where $e=uv$. 
    Formally:
    \begin{lemma} \label{treedepth:lemma:edge:sets:with:vertex:or:edge:gadgets}
        Let $\Par = \{C_1\ldots, C_p\}$ be an optimal $\C$-partition of $G$, $w \in V_{e}$, for some $e=uv \in E(H)$, and $w \in C$ for some $C\in\Par$. 
        If $V(F^u) \cap  C = \emptyset$ and $V(F^v) \cap  C = \emptyset$  then $V(F^e) \cup \{w\} \subseteq C$. 
    \end{lemma}

    \begin{proof}
        It follows by Lemma~\ref{treedepth:lemma:vertex-egde:gadgets:in:one:set} that there exists a $C'\in \Par$ such that
        $V(F^e) \subseteq C' \subseteq V(F^e ) \cup V_e $. 
        Indeed, assuming otherwise, $C'$ would include vertices from a vertex-gadget. In this case we would have that $|C'| > \C$, a contradiction.
        Assume that $V(F^u) \cap  C = \emptyset$ and $V(F^v) \cap  C = \emptyset$.
        If $C' \neq C$ then $w$ contributes $0$ edges to the value of $\Par$ since $N(w) \cap C= \emptyset$. 
        Now since $C' \subseteq V(F^e ) \cup V_e$, and $|V(F^e )| = \C - 4$, 
        we know that we always can move $w$ to $C'$ and increase the value of the partition. 
        Therefore, $C' = C$.
    \end{proof}
    
    Next, we will calculate the absolute maximum value of any $\C$-partition of $G$.
    Notice that in any optimal $\C$-partition, we have two kind of sets: those that include vertices of vertex or list-gadgets and
    those that include vertices from edge-gadgets. We separate the sets of any optimal $\C$-partition of $G$ based on that.
    In particular, for an optimal $\C$-partition $\Par $, we define $\Par_V$ and $\Par_E$ as follows.
    We set $\Par_E \subseteq \Par$ such that $C\in \Par_E$ if and only if there exists an edge-gadget $F^e$ such that $V(F^e) \subseteq C$.
    Then we set $\Par_V = \Par \setminus \Par_E$. 

    It is straightforward to see that the previous lemmas also hold for optimal $\C$-partitions $\Par'$ of $G[V(\Par_E)]$ and $\Par''$ of $G[V(\Par_V)]$.
    Indeed, assuming otherwise, we could create a $\C$-partition for $G$ of higher value since $\Par$ is the concatenation of $\Par'$ and $\Par''$.

    Notice now that for any vertex in $V(G) \setminus V_E$, we know whether it belongs in $V(\Par_V)$ or in $V(\Par_E)$.
    However, this is not true for the vertices of $V_E$. 
    We will assume that $V(\Par_V)$ includes $x$ vertices from $V_E$ and we will use this in order to provide an
    upper bound to the value of $|E(\Par_V)|$ and $|E(\Par_E)|$.

    Let us now consider an optimal partition $\Par$, let  $S = V(\Par_E) \cap V_E $, $x = |S|$ and $y = |V_E \setminus S|$.

    We start with the upper bound of $|E(\Par_V)|$.
    Let $F^v$ be a vertex-gadget. Recall that, there are $n$ list-gadgets adjacent to $F^v$ and, by Lemma~\ref{treedepth:lemma:vertex:gadgets:together:list:gadgets},
    exactly one of them is in the same set as $F^v$ in any optimal $\C$-partition.
    Let $F$ be a list-gadget that is not in the same set as $F^v$ in $\Par$ (and thus in $\Par_V$).
    By Lemma~\ref{treedepth:lemma:list:gadgets:F:in:one:set}, we know that all vertices of $F$ are in the same set of $\Par_V$.
    Thus, for each one of them, we have $m_\ell$ edges in $E(\Par_V)$. 
    Since there are $n-1$ such list-gadgets for each one of the $n$ vertex-gadgets, in total we have  $ n(n-1)m_\ell$
    edges that do not belong in the same set as a vertex-gadget.

    Now, let $F$ be the list-gadget such that the vertices of $V(F)$ and $V(F^v)$ are in the same set $C\in \Par_V$.
    Let $\alpha_v$ be the value represented by the list-gadget $F$.
    Since $|V(F^v) \cup V(F)| = \C -2\alpha_v$, at most $2\alpha_v$ of the $y$ vertices of $V_E$ can be in $C$.
    Let $|C\cap V_E|=y_v \le 2\alpha_v \le y$. Since these vertices must be incident to $V_1(F^v)$, we have that
    $|E(G[C])| = m_v +  m_\ell +4 y_v + 2(5 A -2\alpha_v) = 
    m_v +  m_\ell  + 10  A  +4 y_v - 4 \alpha_v  $; the $2(5 A -2\alpha_v)$ term
    comes from the fact that exactly $2$ vertices of $V_1(F^v)$ are adjacent to all the vertices of $V_2(F)$. 
    By counting all sets that include vertices from vertex-gadgets we have that
    \[ m_v n + m_\ell n + 10  A  + 4 \sum_{v \in V(H)} y_v - 4\sum_{v \in V(H)} \alpha_v
    \]
    In total:
    \[ |E(\Par_V)| = m_v n + m_\ell n^2 +  10n  A  + 4 \sum_{v \in V(H)} y_v - 4\sum_{v \in V(H)} \alpha_v,
    \]
    where $n = |V(H)|$.

    Now, we will calculate an upper bound of $|E(\Par_E)|$ and we will give some properties that must be satisfied in order to achieve this maximum.
    Let $S = V_E \cap V(\Par_E)$. By Lemma~\ref{treedepth:lemma:edge:sets:with:vertex:or:edge:gadgets}, we have that $\Par_E$ consists of
    the vertex sets of the connecting components of $G[V(\Par_E)]$.
    Thus, in order compute an upper bound of $|E(\Par_E)|$, it is suffices to find an upper bound of the number of edges in $G[S' \cup \bigcup_{e \in E(H)} V(F^e)]$,
    for any set $S' \subseteq V_E$ where $|S'| = |S|$.

    For any $G[S' \cup \bigcup_{e \in E(H)} V(F^e)]$, where $S' \subseteq V_E$, we define types of its connected components based on the size of 
    their intersection with $V_E$. 
    In particular, let $\mathcal{X} = \{C_1\ldots,C_p\}$ be the vertex sets of the connected component of 
    $G[S \cup \bigcup_{e \in E(H)} V(F^e)]$. 
    For any set $C \in \mathcal{X}$, we have $| C \cap V_E|  = i$, where $i\in \{0,1,2,3,4\}$.
    We set $\mathcal{X}_i = \{ C\in \mathcal{X} \mid | C \cap V_E|  = i\}$. Notice that $\sum_{i=0}^4 |\mathcal{X}_i| = |E(H)|$. 

    We claim that, in order to maximize the number of edges in $G[S \cup \bigcup_{e \in E(H)} V(F^e)]$,
    we would like to have as many sets in $\mathcal{X}_4$ as possible. Formally, we have the following lemma.

    \begin{lemma} \label{treedepth:lemma:max:value:edge:gadgets}
        For any $0\le x \le |V_E|$, let $S$ be a subset of $V_E$ such that $|S|=x$ and 
        $|E(G[U_E \cup S])| = \max_{S'\subseteq V_E, |S'|=x} |E(G[U_E \cup S'])| $.
        Assume that $\mathcal{X} = \{C_1\ldots,C_p\}$ are the vertex sets of the connected components of $G[U_E \cup S]$.
        We have that $|\mathcal{X}_1|+|\mathcal{X}_2|+|\mathcal{X}_3| =0 $ if $x\ mod\ 4=0$ and $|\mathcal{X}_1|+|\mathcal{X}_2|+|\mathcal{X}_3| =1 $ otherwise.
    \end{lemma}
    \begin{proof}
        First, we will prove that $|\mathcal{X}_1|+|\mathcal{X}_2|+|\mathcal{X}_3| \le 1 $. 
        Assume that $|\mathcal{X}_1|+|\mathcal{X}_2|+|\mathcal{X}_3| > 1 $. We will show that there exists a set $S'$ such that $|S'|=x$ and $|E(G[U_E \cup S])| < |E(G[U_E \cup S'])|$.
        Let $C^1$ and $C^2$ be two sets in $\mathcal{X}$ such that $C^1,C^2 \notin \mathcal{X}_0 \cup \mathcal{X}_4$. 
        Let $C^1 \in \mathcal{X}_{\ell_1}$ and $C^2 \in \mathcal{X}_{\ell_2}$. We distinguish two cases: $\ell_1+\ell_2\le 4$ and $\ell_1+\ell_2> 4$.

        \medskip 
        
        \noindent\textbf{Case 1. $\boldsymbol{\ell_1+\ell_2\le 4}$.} Let $F^{e_1}$ and $F^{e_2}$ be the edge-gadgets such that $V(F^{e_1}) \subseteq C_1$ and $V(F^{e_2}) \subseteq C_2$.
        We modify the sets $C^1$ and $C^2$ as follows:
        \begin{itemize}
            \item We replace $C^1$ with $C^{1'} = C^1\setminus V_{e_1}$ and
            \item we replace $C^2$  with $C^{2'} = C^2\setminus V_{e_2} \cup Y$ where $Y \supseteq \{ w_{L1}^e,w_{R1}^e \}$ and $|Y| = \ell_1+\ell_2$.
        \end{itemize}
        Let $S'=Y \cup ( S \setminus (V_{e_1} \cup V_{e_2}) ) $ and let us denote the resulting partition by $\mathcal{X}'$.
        Notice that $\ell_1+\ell_2 \ge 2$. So we can always have $Y \supseteq \{ w_{L1}^e,w_{R1}^e  \}$ and $|Y| = \ell_1+\ell_2$.
        Also, $|V(\mathcal{X}) \cap V_E| = |V(\mathcal{X}') \cap V_E| =x$. It remains to show that $|E(\mathcal{X}')| > |E(\mathcal{X})| $.
        It suffices to show that $|E(G[C^{1'}])|+|E(G[C^{2'}])| > |E(G[C^1])|+|E(G[C^2])|$. To achieve that, we again distinguish three sub-cases: $\ell_1+\ell_2 = 2$, $\ell_1+\ell_2 =3$, or $\ell_1+\ell_2 =4$.

        \smallskip 
        
        \noindent\textbf{Case 1.a. $\boldsymbol{\ell_1+\ell_2 = 2}$.} Since $\ell_1\ge 1$ and $\ell_2\ge 1$, we have that $\ell_1= \ell_2= 1$.
        Thus, by the construction of $G$, $|E(G[C^1])| = |E(G[C^2])| = |E(G[F^{e_1}])| +1 $ (as all edge-gadgets have the same number of edges).
        Also, since $\ell_1+\ell_2 = 2$, we get that $Y = \{ w_{L1}^e,w_{R1}^e  \}$. This, by the construction of $G$, gives us that $|E(G[C^{1'}])| = |E(G[F^{e_1}])|$ and $|E(G[C^{2'}])| = |E(G[F^{e_2}])| +3 $.
        Therefore, $|E(G[C^{1'}])|+|E(G[C^{2'}])| > |E(G[C^1])|+|E(G[C^2])|$.

        \smallskip 
        
        \noindent\textbf{Case 1.b. $\boldsymbol{\ell_1+\ell_2 = 3}$.} Since $\ell_1\ge 1$ and $\ell_2\ge 1$ we have that either $\ell_1=2 $ and $ \ell_2= 1$ or $\ell_1=1 $ and $ \ell_2= 2$.
        Assume, w.l.o.g., that $\ell_1=2 $ and $ \ell_2= 1$.
        By the construction of $G$ we have that $|E(G[C^1])| \le  |E(G[F^{e_1}])| +3 $ and $|E(G[C^2])| \le  |E(G[F^{e_2}])| +1 $.
        Also, since $\ell_1+\ell_2 = 3$ and $Y \supseteq \{ w_{L1}^e,w_{R1}^e  \}$, we have that $|E(G[C^{1'}])| = |E(G[F^{e_1}])|$ and $|E(G[C^{2'}])| = |E(G[F^{e_2}])| +5 $.
        Therefore, $|E(G[C^{1'}])|+|E(G[C^{2'}])| > |E(G[C^1])|+|E(G[C^2])|$.

        \smallskip 
        
        \noindent\textbf{Case 1.c. $\boldsymbol{\ell_1+\ell_2 = 4}$.} Since $\ell_1\ge 1$ and $\ell_2\ge 1$, we have that either $\ell_1 = \ell_2 = 2$ or one of the $\ell_1$ and $ \ell_2$ is $1$ and the other $3$. In the first case, $|E(G[C^1])| = |E(G[C^2])| \le  |E(G[F^{e_1}])| +3 $ while in the second, $|E(G[C^1])| =  |E(G[F^{e_1}])| +1 $ and $|E(G[C^1])| =  |E(G[F^{e_1}])| +5 $. In both cases,  $|E(G[C^{1}])|+|E(G[C^{2}])| \le 6$.
        Also, since $\ell_1+\ell_2 = 4$ and $Y = V_{e_2}$, we have that $|E(G[C^{1'}])| = |E(G[F^{e_1}])|$ and $|E(G[C^{2'}])| = |E(G[F^{e_2}])| +8 $.
        Therefore, $|E(G[C^{1'}])|+|E(G[C^{2'}])| > |E(G[C^1])|+|E(G[C^2])|$.

        \medskip 
        
        \noindent\textbf{Case 2. $\boldsymbol{\ell_1+\ell_2> 4}$.} Let $F^{e_1}$ and $F^{e_2}$ be the edge-gadgets such that $V(F^{e_1}) \subseteq C_1$ and $V(F^{e_2}) \subseteq C_2$. 
        We modify the sets $C^1$ and $C^2$ as follows:
        \begin{itemize}
            \item We replace $C^1$ with $C^{1'} = C^1\cup V_{e_1}$ and
            \item we replace $C^2$  with $C^{2'} = C^2\setminus V_{e_2} \cup Y$ where $Y \subseteq \{ w_{L1}^e,w_{R1}^e \}$ and $|Y| = \ell_1+\ell_2 - 4$.
        \end{itemize}
        Indeed, it suffices to have $Y \subseteq \{ w_{L1}^e,w_{R1}^e \}$ as $2 \le \ell_1,\ell_2 \le 3 $, and thus, $\ell_1+\ell_2 - 4 <3$. We need to consider two cases, either $\ell_1+\ell_2 = 5$ or $\ell_1+\ell_2=6$
        
        \smallskip 
        
        \noindent\textbf{Case 2.a. $\boldsymbol{\ell_1+\ell_2 = 5}$.} In this case, we have that one of $\ell_1, \ell_2$ is equal to $ 2$ while the other is equal to $3$. W.l.o.g. let $\ell_1 =2$.
        By the construction of $G$, we get that $|E(G[C^1])| \le |E(G[F^{e_1}])| +3$ (as all edge-gadgets have the same number of edges). Also, $|E(G[C^1])| \le |E(G[F^{e_1}])| +5$.
        Now observe that $|E(G[C^{1'}])| = |E(G[F^{e_1}])| +8 $ and $|E(G[C^{1'}])| = |E(G[F^{e_1}])| +1 $.
        Therefore, $|E(G[C^{1'}])|+|E(G[C^{2'}])| > |E(G[C^1])|+|E(G[C^2])|$.

        \smallskip 
        
        \noindent\textbf{Case 2.b. $\boldsymbol{\ell_1+\ell_2 = 6}$.} In this case, we have that one of $\ell_1= \ell_2=3$.
        By the construction of $G$ we obtain that $|E(G[C^1])| = |E(G[C^2])| = |E(G[F^{e_1}])| +5$ (as all edge-gadgets have the same number of edges).
        We also have that $|E(G[C^{1'}])| = |E(G[F^{e_1}])| +8 $ and $|E(G[C^{1'}])| = |E(G[F^{e_1}])| +3 $ (since $Y = \{ w_{L1}^e,w_{R1}^e \}$ in this case).
        Therefore, $|E(G[C^{1'}])|+|E(G[C^{2'}])| > |E(G[C^1])|+|E(G[C^2])|$.
        
        To sum up, we have that $|\mathcal{X}_1|+|\mathcal{X}_2|+|\mathcal{X}_3| \le 1 $. We will now show that  $|\mathcal{X}_1|+|\mathcal{X}_2|+|\mathcal{X}_3| =0 $ if $x\ mod\ 4=0$ and $|\mathcal{X}_1|+|\mathcal{X}_2|+|\mathcal{X}_3| =1 $ otherwise.
        
        Assume that $x\ mod\ 4=0$. Notice that $4|\mathcal{X}_4|+ 3|\mathcal{X}_3|+2|\mathcal{X}_2|+|\mathcal{X}_1| +0 |\mathcal{X}_0| = x$; therefore 
        $
        (4|\mathcal{X}_4|+ 3|\mathcal{X}_3|+2|\mathcal{X}_2|+|\mathcal{X}_1| +0 |\mathcal{X}_0|)\ mod\ 4 = 0 \implies
        (3|\mathcal{X}_3|+2|\mathcal{X}_2|+|\mathcal{X}_1|)\ mod\ 4 = 0
        $.
        This implies that $|\mathcal{X}_1|+|\mathcal{X}_2|+|\mathcal{X}_3| =0 $. Indeed, assuming otherwise we get that $|\mathcal{X}_1|+|\mathcal{X}_2|+|\mathcal{X}_3| =1 $, and thus $(3|\mathcal{X}_3|+2|\mathcal{X}_2|+|\mathcal{X}_1|)\ mod\ 4 =i$, for an $i \in [3]$.
        This is a contradiction to $(3|\mathcal{X}_3|+2|\mathcal{X}_2|+|\mathcal{X}_1|)\ mod\ 4 = 0$.
        
        Next, assume that $x\ mod\ 4=i$ for $i \in [3]$. Then we have that
        $
        4|\mathcal{X}_4|+ 3|\mathcal{X}_3|+2|\mathcal{X}_2|+|\mathcal{X}_1| +0 |\mathcal{X}_0| = x \implies
        (4|\mathcal{X}_4|+ 3|\mathcal{X}_3|+2|\mathcal{X}_2|+|\mathcal{X}_1| +0 |\mathcal{X}_0|)\ mod\ 4 = i \implies
        $
        $
        (3|\mathcal{X}_3|+2|\mathcal{X}_2|+|\mathcal{X}_1|)\ mod\ 4 = i
        $.
        If $|\mathcal{X}_1|+|\mathcal{X}_2|+|\mathcal{X}_3| =0 $ then the previous implies that $i = 0$ which is a contradiction. This finishes the proof of this lemma.
    \end{proof}

    It follows that the maximum value of $\max_{S'\subseteq V_E, |S'|=x} |E(G[U_E \cup S'])|$ is
    \begin{itemize}
        \item $m_e |E(H)| +8 x/4$, when $x\ mod\ 4=0$, or
        \item $m_e |E(H)| +8 (x-i)/4 + x_i$, when $x\ mod\ 4=i$,
    \end{itemize}
    where $x_1 = 1$, $x_2 = 3$, $x_3=5$.
    Notice that $\max_{S'\subseteq V_E, |S'|=x} |E(G[U_E \cup S'])| \le m_e |E(H)| +2 x$
    where the equality holds only when $x\ mod\ 4=0$.

    Thus, we have the following:
    \begin{corollary} \label{treedepth:observation:max:partition}
        Given that $x = |S|$, the maximum value of $\Par$ is:
        $v(\Par) \le m_v|V(H)|  + m_\ell |V(H)|^2  + m_e |E(H)|   + 10 A  |V(H)| + 8 |E(H)|$.
        Also, this can be achieved only when $x\ mod\ 4=0$ and 
        $x = |V_E|-y = \sum_{v \in H(v)} 2\alpha_v$.
    \end{corollary}

\medskip 

We are finally ready to prove our main theorem. 
\thmTreedepth*
\begin{proof}
        Let $(H,L)$ be the input of the \textsc{General Factors} problem, let $G$ be the graph constructed by $(H,L)$ as described above, and let $\Par$ be an optimal $\C$-partition of $G$. We will prove that the following two statements are equivalent:
    \begin{itemize}
        \item $v(\Par)= m_v|V(H)|  + m_\ell |V(H)|^2  + m_e |E(H)|   + 10\sizeOfListGadget |V(H)| + 8 |E(H)| $
        \item $(H, L)$ is a yes-instance of the \textsc{General Factors} problem.
    \end{itemize}

    Assume that $(H, L)$ is a yes-instance of the \textsc{General Factors} problem and let $E'\subseteq E(H)$ be the edge set such that,
    for any vertex $v\in V(H)$, we have $d_{H-E'}(v) \in L(v)$. We will create a $\C$-partition of $G$ that has value $m_v|V(H)|  + m_\ell |V(H)|^2  + m_e |E(H)|   + 10\sizeOfListGadget |V(H)| + 8 |E(H)| $.
    For each edge $e \in E(H) \setminus E'$ we create a set $C_e = V(F^e)$ and for each $e \in E'$ we create a set $C_e = V(F^e)\cup V_e $.
    For each $v \in V(H)$, let $F$ be the list-gadget that represents the value  $d_{H-E'}(v)$. The existence of such a list-gadget is guaranteed since $d_{H-E'}(v) \in L(v)$. Also, let $U_{v}$ be the subset of $V_E$ such that, for any $u \in U_v$, there exists an edge $e\in E(H) \setminus E'$
    such that $u \in V_e$ and $u$ is incident to the vertices of $V_1(F^v)$ (this means that $e$ is incident to $v$ in $H$). 
    Notice that the vertices in $U_v$ have not been included in any set $C_e$, for $e\in E(H)$, that we have created this far.
    Now, for each $v \in V(H)$, we create a set $C_v = V(F^v) \cup V(F) \cup U_v$.
    It remains to deal with the list-gadgets that have not yet been included in any set.
    We create the sets $C_1, \ldots, C_{n(n-1)}$, one for each one of them.
    We claim that $\Par = \{C_e\mid e \in E(H) \} \cup \{C_v \mid v\in V(H)\} \cup \{C_1, \ldots, C_{n(n-1)}\}$ is a
    $\C$-partition of $G$ and $v(\Par) = m_v|V(H)|  + m_\ell |V(H)|^2  + m_e |E(H)|   + 10\sizeOfListGadget |V(H)| + 8 |E(H)| $. Notice that any of the sets $C \in \{C_e\mid e \in E(H) \} \cup \{C_1, \ldots, C_{n(n-1)}\}$
    have size at most $\C$ as they are either vertex sets of a list-gadget or a subset of $V(F^e) \cup V_e$, for some $e \in E(H)$.
    Thus we only need to show that $|C_v| \le \C$ for all $v\in V(H)$. 
    We have that $|V(F^v) \cup F| \subseteq C_v$ where $F$ is the list-gadget that represents the value $d_{H-E'}(v) \in L(v)$.
    Therefore,  $|V(F^v) \cup F|  = \C - 2d_{H-E'}(v) $. We claim that $|U_v| = 2d_{H-E'}(v)$. Recall that $U_v$ contains the vertices of $V_E$ for which
    there exists an edge $e\in E(H) \setminus E'$ such that $u \in V_e$ and $u$ is incident to the vertices of $V_1(F^v)$.
    Actually, there are exactly $d_{H-E'}(v)$ edges incident to $v$ from $E(H) \setminus E'$. Also, for each such edge $e$,
    two vertices of $V_e$ are incident to $V_1(F^v)$ (the vertices $w_{L1}^e,w_{L2}^e$ if $v \in V_L$ and $w_{R1}^e,w_{R2}^e$ if $v \in V_R$). Thus $|U_v| = 2d_{H-E'}(v)$ and $|C_v| = \C$.

    It remains to argue that $v(\Par) = m_v|V(H)|  + m_\ell |V(H)|^2  + m_e |E(H)|   + 10\sizeOfListGadget |V(H)| + 8 |E(H)| $. First, notice that the vertex set $V(F)$ of any gadget $F$ belongs to one set.
    Thus, every edge of $E(G[V(F)])$ contributes in the value of $\Par$. This give us $m_v$ edges for each vertex-gadget, $m_e$ edges for each edge-gadget and $m_\ell$ edges for each list-gadget.
    Since we have $|V(H)|$ vertex-gadgets, $|E(H)|$ edge-gadgets and $|V(H)|^2$ list-gadgets, this gives $m_v|V(H)|  + m_\ell |V(H)|^2 + m_e |E(H)| $ edges (up to this point).
    
    We also need to compute the number of edges in $E(\Par)$ that do not belong in any set $E(G[V(F)])$, for any gadget $F$.
    Let $S$ be the set $E(\Par) \setminus \bigcap_{F \text{ is any gadget}} E(G[V(F)])$.
    Notice that for any $C \in \Par$, we have $S \cap C \neq \emptyset$ if and only if there exists a vertex or edge-gadget $F$ such that $V(F) \subseteq C$.
    
    First we consider a set $C$ that includes a vertex-gadget. By construction, we have that $C$ includes the vertices of a vertex-gadget $F^v$, the vertices of a list-gadget $F$ that represents an integer $\alpha_v$ and $2\alpha_v$ vertices from the set $V_E$.
    There are exactly $|U^v| \cdot |V_1(F)|$ edges between $F^v$ and $F$. Also, for any vertex $u \in C\cap V_E$, we have that $N(u) \cap C = V_1(F^v)$.
    Thus, we have   $|U^v| \cdot |V_1(F)| + |C\cap V_E| \cdot |V_1(F^v)|= 2(5 \sizeOfListGadget- 2\alpha) + 2\alpha \cdot 4 = 10 \sizeOfListGadget+ 4\alpha_v$ edges.
    Also, by construction, $C\cap V_E$ is an independent set. Thus we have no other edges to count.
    This gives us $10 n \sizeOfListGadget + 4 \sum_{v \in V(H)}\alpha_v $.

    Now we consider a set $C$ that includes an edge-gadget. By the construction of $C$ we have that there exists an edge $e \in E(H)$
    such that either $C=V(F^e)$ or $C=V(F^e) \cup V_e$. Therefore, if $e \in E'$ then $C = V(F^e) \cup V_e$ 
    and $E(G[C])$ includes $8$ edges incident to vertices of $V_E$, while if $e \notin E'$ then $C =V(F^e)$ and $E(G[C])$ does not include edges
    incident to vertices of $V_E$. This gives us $8|E'|$ extra edges.

    In order to complete the calculation of $|S|$ we need to observe that the values $\alpha_v$, $v\in V(H)$ and $|E'|$ are related. 
    In particular, by the selection of $\alpha_v$, we have that $\sum_{v\in V(H)} d_{H-E'}(v) = \sum_{v\in V(H)} \alpha_v = 2 |E(H)\setminus E'| $.
    It follows that:
    $
    |S| = 10 n \sizeOfListGadget + 4 \sum_{v \in V(H)}\alpha_v + 8|E'| = 10n \sizeOfListGadget + 8 |E(H)\setminus E'| + 8 |E'|=  10n \sizeOfListGadget + 8 |E(H)|$.

    In total, $
    |E(\Par)| = m_v|V(H)|  + m_\ell |V(H)|^2  + m_e |E(H)|   + 10\sizeOfListGadget |V(H)| + 8 |E(H)|$.
    
\smallskip

    For the reverse direction, assume that we have a $\C$-partition $\Par$ of $G$ such that $v(\Par) =m_v|V(H)|  + m_\ell |V(H)|^2  + m_e |E(H)|   + 10\sizeOfListGadget |V(H)| + 8 |E(H)| $.
    By the calculated upper bounds, we have that
    \begin{itemize}
        \item $|E(\Par_E)| = m_e |E(H)| + 2x$ and
        \item $|E(\Par_V)| = m_v n + m_\ell n^2 +  10n \sizeOfListGadget + 4 \sum_{v \in V(H)} y_v - 4\sum_{v \in V(H)} \alpha_v$.
    \end{itemize}
    Also, in order to achieve $m_v|V(H)|  + m_\ell |V(H)|^2  + m_e |E(H)|   + 10\sizeOfListGadget |V(H)| + 8 |E(H)| $, we have that $\sum_{v \in V(H)}\alpha_v = (|V_E| - x) /2$.
    Recall that in order to achieve the maximum value, for any edge $e \in E(H)$, either $V_e \subset \Par_V$ or $V_e \subset \Par_E$.
    Let $E' = \{e \in E(H) \mid V_e \subset \Par_E \}$. 
    We claim that for any $v \in V(H)$, we have $d_{H-E'}(v) \in L(v)$. Let $V(F^v) \subseteq C$, for some $C\in \Par$, 
    $F$ be the list-gadget such that $F\subseteq C$ and $|C \cap V_E| = x_v $.
    By Corollary~\ref{treedepth:observation:max:partition} we obtain that $2 \alpha_v =x_v$ where $a_v$ is the value represented by $F$, if the partition is of optimal value. 
    Observe that, for any edge $e \in E(H) \setminus E'$ incident to $v$, two vertices of $V_e$ are in $C \cap V_E$.
    Thus, $2 d_{H-E'}(v)  = x_v$. Since $2\alpha_v = x_v$ and $\alpha_v \in L(v)$ (by the construction of list-gadgets)
    we obtain that $d_{H-E'}(v) = \alpha_v \in L(v)$. Thus $(H,L)$ is a yes-instance of the \textsc{General Factors} problem.

    \paragraph*{The tree-depth of $G$ is bounded}

    The only thing that remains to be shown is that the tree-depth of $G$ is bounded by a computable function of $m$.
    Recall that $m$ is the size of one of the bipartitions of $H$. W.l.o.g., assume that $|V_L| = m $.
    We start by deleting the set $V_1(F^v)$, for all $v \in V_L$. This means that we have deleted $4m$ vertices.
    Now, we will calculate an upper bound of the tree-depth of the remaining graph.
    In the new graph, there are connected components that include vertices from vertex-gadgets $F^v$, for $v \in V_R$, but no connected component includes two such gadgets. For each such a component, we delete the vertices $V_1(F^v)$, for each $v \in V_R$.
    Since these deletion are in different components, they are increasing the upper bound of the tree-depth of the original graph 
    by $4$. 
    Also, after these deletions, any connected component that remains is:
    \begin{itemize}
        \item either a list-gadget $F$,
        \item or isomorphic to $G[V_e \cup V(F^e)]$ (for any $e \in E(H)$),
        \item or isomorphic to $G[V_2(F^v) \cup V_3(F^v)]$ (for any $e \in E(H)$).
    \end{itemize}
    We claim that in any of these cases, the tree-depth of this connected component is at most $\bO(m)$.
    Consider a list-gadget $F$. Any $G[V(F)]$ had tree-depth at most $m+ 1$. This holds because if we remove $V_1(F)$ we remain with a set of independent vertices plus a matching. 
    Consider a connected component isomorphic to $G[V_e \cup V(F^e)]$. Observe that $G[V_e \cup V(F^e)]$ has tree-depth $2m+5$ because removing $V_1(F^e) \cup V_3(F^e) \cup V_e$ results to an independent set and $|V_1(F^e) \cup V_3(F^e) \cup V_e| = 2m +5$ for all $e \in E(H)$.
    Finally, consider a connected component isomorphic to $G[V_2(F^v) \cup V_3(F^v)]$. 
    In this case the tree-depth of this component is upperly bounded by $2m$ since deleting $V_3(F^v)$ results to an independent set. 

    In total the tree-depth of $G$ is upperly bounded by $3m+9$. This completes the proof.
\end{proof}
\fi

\subsection{Graphs of bounded vertex cover number}
Next, we consider a slightly relaxed parameter, the vertex cover number of the input graph. We begin with establishing that \CCFw admits FPT algorithm by the vertex cover number of the input graph. It contrasts \Cref{thm:TwC-FPT,thm:treedepth} since we can remove the dependence on $\mathcal{C}$ in the FPT algorithm in the following theorem, even for the weighted case. 

\thmVcFPT*

\begin{proof}
Let $U$ be a vertex cover of $G$ of size $\vc$ and let $I$ be the  independent set $V\setminus U$. If such a vertex cover is not provided as input, we can compute one in time $2^{\vc}n^{\bO(1)}$ time~\cite{CyganFKLMPPS15}. First, observe that there can be at most $\vc$ many coalitions in $G$ which can have a positive contribution (since the contribution comes from edges and each edge in $G$ is incident to some vertex in $U$). Next, we guess $\Par'=\{C_1,\dots,C_p\}$ (here, $p\leq \vc$), the intersection of the sets of an optimal $\C$-partition of $G$ with $U$; let $W=v(\Par')$. Notice that we can enumerate all $\vc^{\bO(\vc)}$ partitions of $U$ in $\vc^{\bO(\vc)}$ time.

Next, for each $\Par'$ we do the following (in $n^{\bO(1)}$ time). We create a new graph $G'$ as follows. First, we create the vertex sets $S_i$, where $|S_i|=\C-C_i$ for each $i\in[p]$.  Then, we add all the edges between the vertices of $x\in I$ and $S_i$ if $v\in N(C_i)$, for every $i\in[p]$. Formally, $S_i=\{u^i_1, \ldots, u^i_{S_i}\}$ for $i\in [p]$, $V(G')= \bigcup_{i\in[p]} S_i$, and $E(G') = \{u^i_jx~|~x\in N(C_i)\cap I, i\in [p], j\in [S_i]\}$. Finally, for every edge $xy$, with $x\in S_i$ and $y\in I$, we set the weight $w(xy)=\sum_{u\in C_i \land uy\in E} w(uy)$, i.e., to be equal to how much $W$ would increase if $y$ was added to $C_i$. Now, observe that in order to compute an optimal $\C$-partition of $G$ whose intersection with $U$ is $\Par'$, it suffices to find a maximum weighted matching of $G'$, which can be done in polynomial time~\cite{E65}. Since we do this operation for each possible intersection of the $\C$-partition with $U$, all of which can be enumerated in $\vc^{\bO(\vc)}$ time, we can compute an optimal $\C$-partition of $G$ in time $\vc^{\bO(\vc)}n^{\bO(1)}$.
\end{proof}

In the rest of this subsection, we establish that our algorithm from \Cref{thm:vc-FPT}, which enumerates all possible intersections of a smallest size vertex cover with an optimal solution, is asymptotically optimal when assuming the ETH. Notably, we design a tight lower bound result to match both \Cref{thm:TwC-FPT,thm:vc-FPT} by establishing that it is highly unlikely that \CCF admits an algorithm with running time $(\mathcal{C}\vc)^{o(\vc+\mathcal{C})}n^{\mathcal{O}(1)}$. This is achieved through a reduction from a restricted version of the $3$-\textsc{SAT} problem.

\Pb{\textsc{\textsc{R$3$-SAT}}}{A \textsc{$3$-SAT} formula $\phi$ defined on a set of variables $X$ and a set of clauses $C$. Additionally, each variable appears at most four times in $C$ and the variable set $X$ is partitioned into $X_1\cup X_2\cup X_3$, such that every clause includes at most one variable from each one of the sets $X_1$, $X_2$ and $X_3$.
    }{Question}
{Does there exist a truth assignment to the variables of $X$ that satisfies $\phi$?}

 First, we establish that \textsc{R$3$-SAT} is unlikely to admit a $2^{o(n+m)}$  time algorithm. 
\begin{lemma}
    The \textsc{\textsc{R$3$-SAT}} problem is \np-hard. Also, under the ETH, there is no algorithm that solves this problem in time $2^{o(n+m)}$, where $n = \max \{ |X_1,|X_2|,|X_3|\}$ and $m$ is the number of clauses.
\end{lemma}
\begin{proof}
    The reduction is from \textsc{$3$-SAT}. First we make sure that each variable appears at most four times. 
    Assume that variable $x$ appears $k>3$ times. We create $k$ new variables $x_1,\ldots,x_k $ and
    replace the $i$-th appearance of $x$ with $x_i$. Finally, we add the clauses
    $(x_1\lor \neg x_2 )\land(x_2\lor \neg x_3 ) \ldots (x_k\lor \neg x_1 )$. 
    This procedure is repeated until there is no variable that appear more than $3$ times.

    Next, we create an instance where the variables are partitioned in the wanted way.
    First, we fix the order that the variables appear in each clause.
    Let $x$ be any variable that appears in the formula. If $x$ appears only in the $i$-th position of every clause it is part of (for some $i \in [3]$),
    then we add $x$ into $X_i$.
    Otherwise, we create three new variables $x_1,x_2,x_3$ and, for each clause $c\in C$,
    if $x$ appears in the $i$-th position of $c$, we replace it with $x_i$. Notice that, at the moment, $x_i$ appears at most twice for each $i\in[3]$. We add $x_i$ in the set $X_i$, for all $i\in [3]$.
    Also, we add the clauses $(x_1\lor \neg x_2 )\land(x_2\lor \neg x_3 ) \land (x_3\lor \neg x_1 )$. Thus, in any satisfying assignment of the formula, the variables $x_1$, $x_2$ and $x_3$ have the same assignment.
    Notice that in each one of the original clauses, the $i$-th literal contains a variable from $X_i$.
    Therefore, each one of the original clauses have at most one variable from $X_i$ for each $i\in[3]$.
    This is also true for all the clauses that were added during the construction.

    It is easy to see that the constructed formula is satisfiable if and only if $\phi$ is also satisfiable.

    Finally, notice that the number of variables and clauses that were added is linear in regards to $n+m$.
    Therefore, we cannot have an algorithm that runs in $2^{o(n+m)}$ and decides whether the
    new instance is satisfiable unless the ETH is false.
\end{proof}

Once more, we begin by describing the construction of our reduction, we continue with a high-level idea of the reduction, prove the set of properties that should be verified by any optimal $\C$-partition of the constructed graph and finish with the reduction.

\medskip 

\noindent\textbf{The construction.} Let $(X,C)$ be an instance of the \textsc{\textsc{R$3$-SAT}} problem, and let $X=X_1\cup X_2\cup X_3$ be the partition of $X$ as it is defined above. We may assume that $|X_1|= |X_2|=|X_3|=n = 2^k$ for some $k \in \mathbb{N} $. If this is not the case, we can add enough dummy variables that are not used anywhere just to make sure that this holds. We can also assume that $k$ is an even number; if not, we can double the variables to achieve that. Notice that the number of additional dummy variables is at most $2max\{|X_1|,|X_2|,|X_3| \}$, so that the number of variables still remains linear in regards to $\max \{ |X_1,|X_2|,|X_3|\}+m$. 

\begin{figure}[!t]
\centering
\includegraphics[scale=1]{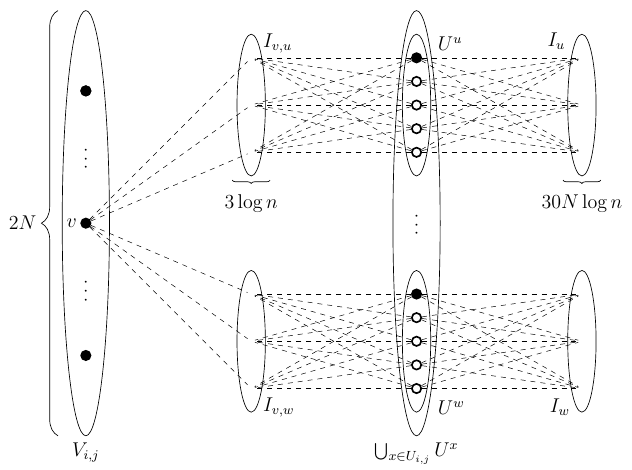}
\caption{The gadget $G_{i,j}$ used in the construction of Theorem~\ref{thm:vc-to-the-vc}.}\label{fig:vc-to-the-vc}
\end{figure}

We start by partitioning each variable set $X_i$ in to $k = \log n $ sets $X_{i,1}\ldots,X_{i,k}$, with $|X_{i,j}|\leq \lceil n / \log n \rceil$ for every $j\in[k]$.
For each set $X_{i,j}$, we construct a variable gadget $G_{i,j}$ (illustrated in Figure~\ref{fig:vc-to-the-vc}) as follows:

\begin{itemize}
\item First, we create a vertex set $V_{i,j}$ with $2N = 2 \lceil n / \log^2 n \rceil $ vertices. Each vertex in $V_{i,j}$ represents at most $\frac{\log n}{2}$ variables.
To see that we have enough vertices to achieve this, observe that $X_{i,j}$ represents a set of $\lceil n / \log n \rceil $ variables. Thus:
$ \Bigg\lceil \frac{\lceil \frac{n}{\log n} \rceil}{2 \lceil \frac{n}{\log^2 n} \rceil} \Bigg\rceil = 
\Big \lceil \frac{ n }{2 \log n \lceil \frac{n}{\log^2 n} \rceil} \Big \rceil \le  
\Big \lceil \frac{ n }{2 \log n \frac{n}{\log^2 n}} \Big\rceil = \Big\lceil \frac{ \log n }{2} \Big\rceil = \frac{ \log n }{2}
$
where the last equality holds because $\log n$ is assumed to be an even number.
Hereafter, let $X(v)$ be the variable set that is represented by $v$.
Also notice that $X(v) \subseteq X_{i,j} \subseteq X_{i}$ for all $v\in V_{i,j}$.

\item Then we create the set of \emph{assignment vertices} $U_{i,j} = \{u_\ell \mid \ell \in [\sqrt{n}] \}$. 
Now, for each vertex $v \in V_{i,j}$ and each assignment over the variable set $X(v)$,
we want to have a vertex of $U_{i,j}$ represent this assignment.
Since $|X(v)| \le \frac{\log n}{2}$,
there are at most $2^{\frac{\log n}{2}}$ different assignments over the variable set $X(v)$.
Therefore, we can select the variables of $U_{i,j}$ to represent the assignments 
over $X(v)$ in such a way that
each assignment is represented by at least one vertex and no vertex represents more than one assignment. Notice that $U_{i,j}$ contains enough vertices to achieve this since $|U_{i,j}|=\sqrt{n}$. 
We are doing the same for all vertices in $V_{i,j}$.

\item We proceed by creating four copies $u^1, \ldots, u^{\Copy}$ of each vertex $u \in U_{i,j}$.
For each assignment vertex $u$, let $U^{u}$ be the set $\{u,u^1, \ldots, u^{\Copy}\}$.
For each set $U^u$, 
we add an independent set $I_u$ of size $\sizeIu N \log n $. Then,
for each vertex $v\in I_u$ we add all the edges between $v$ and the vertices of $U^u$. 
\item Finally, for each pair $(v,u) \in V_{i,j}\times U_{i,j}$, we create an independent set
$I_{v,u}$ of $\sizeIvu \log n $ vertices and, for all $w \in I_{v,u}$ and $x \in U^u \cup \{v\}$, we add the edge $wx$. 
\end{itemize}

This concludes the construction of $G_{i,j}$, which corresponds to the set $X_{i,j}$. 


Let $V_C$ be the set of \textit{clause vertices}, which contains a vertex $v_c$ for each $c\in C$. We add the vertices of $V_C$ to the graph we are constructing. The edges incident to the vertices of $V_C$ are added as follows.

Let $c\in C$, $l$ be a literal that appears in $c$, $x$ be the variable that appears in $l$
and $v$ the variable vertex such that $x \in X(v)$. We first add the edge $v_cv$. 
Now, consider the $(i,j) \in [3]\times [\log n]$ such that  $v \in V_{i,j}$. 
For each vertex $u$ in $U_{i,j}$ we add the edge $uv_c$ if and only if $l$ 
becomes true by the assignment over $X(v)$ represented by $u$.  

Let $G$ be the resulting graph. Finally, set $\C = \sizeCar N \log n  $ to be the capacity of the coalitions, where, recall, $N=\lceil n / \log^2 n \rceil $. This finishes our construction.

\medskip 

\noindent\textbf{High-level description.} We establish that, due to the structural characteristics of $G$ and the specific value chosen for $\C$, any optimal $\C$-partition $\Par$ of $G$ has a very particular structure. In particular, we know that any set $S$ of $\Par$ contains exactly one assignment vertex $u$. Moreover, if $S$ contains a clause vertex $v_c$ (representing the clause $c$) and the $v(\Par)$ is greater than some threshold, then it also contains a vertex $v$ from a variable set, such that the assignment over $X(v)$ corresponding to $u$ satisfies $c$.


\medskip 

\noindent\textbf{Properties of optimal $\C$-partitions of $G$.}
First, we identify the structural properties of any optimal $\C$-partition of $G$ that are going to be used in the reduction. 
In particular we will show that for any optimal $\C$-partition $\{C_1,\ldots, C_p\}$ of $G$, we have that:
\begin{itemize}
    \item  for any $k \in [p]$, if $\{u\}\subseteq C_k$ for some assignment vertex $u$, then $U^u \subseteq C_k$;
    \item for any $k \in [p]$, it holds that $|C_k \cap \bigcup_{(i,j) \in [3]\times [\log n]} U_{i,j}| \le 1$;
    \item for any $(i,j) \in [3]\times [\log n]$ and $v \in V_{i,j}$, if $v \in C$ then $C \cap U_{i,j}\neq \emptyset$.
\end{itemize}

\begin{lemma}\label{lemma_assignment_with_its_copy}
    Let $\Par = \{C_1,\ldots, C_p\}$ be an optimal $\C$-partition of $G$. Let $ u\in U_{i,j}$ for some $(i,j) \in [3]\times [\log n]$ and $u \in C_k$ for some $k \in [p]$. Then, $U^u \cap C_k = U^u$.
\end{lemma}

\begin{proof} 
    Assume that, for some $(i,j) \in [3]\times [\log n]$, there exists a set $U^u$ for an assignment vertex $u \in U_{i,j}$ such that $u \in C_k$, for some $k \in [p]$, and $U^u \cap C_k \neq U^u$.
    We will show that, in this case, $\Par$ is not an optimal $\C$-partition of $G$.
    Indeed, consider the following $\C$-partition of $G$. First set $C= U^u \cup N(U^u) \setminus V_C$. Then, let
    $\Par' = \{C, C_1\setminus C, \ldots, C_p \setminus C\}$. Notice that $\Par'$ is a $\C$-partition. Indeed,
    $|C| = 5+ \sizeIu N \log n + 2N \sizeIvu \log n  \le \sizeCar N\log n = \C$
    and $|C_i\setminus C|\le |C_i|\le  \C $ as $C_i \in \Par$ for all $i \in [p]$. 
    
    We will now show that $v(\Par') > v(\Par)$. First observe that for every $v\in C$ we have that:
    \begin{itemize}
        \item $v \in  U^u$, or
        \item $v \in V_u $ where $ V_u= \{v \mid N(v) =  U^u \}$, or
        \item $v \in V'_u $ where $V'_u = \{v \mid N(v) =  U^u\cup \{v\} \text{ for some }v \in V_{i,j} \}$.
    \end{itemize}
    By construction, we know that $|V_u| = \sizeIu N \log n$ and $|V'_u| = 2\cdot\sizeIvu N \log n$. 

    We now consider $\Par$. Observe that the vertices of $U^u$ are assigned to different components of $\Par$. Thus, we have that:
    \begin{itemize}
        \item at most $4 \cdot \sizeIu N \log n = 120 N \log n$ of the edges incident to vertices of $V_u$ are included the $E(\Par)$, and
        \item at most $5  \cdot 2 N \cdot \sizeIvu \log n = 30 N \log n$ of the edges incident to vertices of $V'_u$ are included in $E(\Par)$.
    \end{itemize}
    
    Also, since $|N(u) \cap V_C| \le 4 N \log n$ and $|N(u^i) \cap V_C| =0$, for all $i\in [4]$, we have that $E(\Par)$ contains at most $4 N \log n$ edges between $U^u$ and $V_C$.
    Therefore, by removing $C$ for all $C_i$, $i \in [p]$, we have reduced the value of $\Par$ by at most
    $(120 + 30 +4 ) N \log n = 154 N \log n$. 
    Let us now count the number of edges in $G[C]$. Since $U^u \cup V_u \subseteq C$, we have that $E(G[C])$ includes all the
    $150 N \log n$ edges between vertices of $U^u$ and $V_u$. 
    Also, we have that $E(G[C])$ contains $\frac{5}{6}$ of the edges incident to $V'_u$. Indeed, $N(V'_u)\cap C =  U_u$. This gives us another $5 \cdot 2N \sizeIvu \log n = 30 N \log n$ edges.
    Furthermore, no other edge appears in $G[C]$. Thus, $E(G[C])$ contains $(150 + 30) N \log n= 180 N \log n$ edges.
    Therefore, we have that $v(\Par') \ge 26 N \log n + v(\Par)$. This is a contradiction to the optimality of $\Par$, as $v(\Par')> v(\Par)$.
\end{proof}

\begin{lemma} \label{lemma_no_two_assignment_vertices}
    Let $\Par = \{C_1,\ldots, C_p\}$ be an optimal $\C$-partition of $G$. For any $k \in p$, there is no pair $(u,u')$ of vertices such that:
    \begin{itemize}
        \item $ u\in U_{i,j}$ for some $(i,j) \in [3]\times [\log n]$,
        \item $ u'\in U_{i',j'}$ for some $(i',j') \in [3]\times [\log n]$ (it is not necessary that $(i,j) \neq (i',j')$) and 
        \item $\{u,u'\} \subseteq C_k$.
    \end{itemize}
\end{lemma}
\begin{proof}
    Assume that this is not true and let $k \in [p]$ be an index for which such a pair $(u,u')$ exists in $C_k$.
    By the optimality of $\Par$ and Lemma~\ref{lemma_assignment_with_its_copy}, we have that $U^{u} \cup U^{u'} \subseteq C_k$.
    By construction, we have that $|I_{u}|=|I_{u'}| =\sizeIu N \log n$.
    Since $u$ and $u'$ belong in the same $C_k$ and $\C = \sizeCar N \log n$, we know that there are at least $(2 \cdot \sizeIu - \sizeCar)N \log n = 18 N \log n$ vertices from the sets $I_u$ and $I_{u'}$ that do not belong in $C_k$. Notice that these vertices do not contribute at all to the value of $\Par$ as they are not in the same partition as any of their neighbors.
    Consider the sets $C^1 = U^u\cup I_{u}$, $C^2= U^{u'} \cup I_{u'}$ and $C= C^1\cup C^2$. 
    We create the $\C$-partition $\Par' = \{C^1, C^2, C_1\setminus C, \ldots, C_p \setminus C\}$. 
    Notice that $\Par'$ is indeed a $\C$-partition as $|C^1|= |C^2| = 5+ \sizeIu  N \log n \le \sizeCar  N \log n$ and $|C_i\setminus C|\le |C_i|\le  \sizeCar N \log n$ as $C_i \in \Par$ for all $i \in [p]$. We will show that $v(\Par)< v(\Par')$.
    
    First, we will deal with the edges incident to vertices of $I_{u}$ and $I_{u'}$.
    Notice that $C^1$ and $C^2$ include all the edges between $U^u$ and $I_{u}$ as well as the edges between $U^{u'}$ and $I_{u'}$. Therefore, $|E(\Par') \setminus E(\Par)| \ge 5 (2 \cdot \sizeIu- \sizeCar)  N \log n = 90 N \log n $. Indeed, each vertex of $I_{u} \cup I_{u'}$ has exactly five neighbors in the set $C^1\cup C^2$ and at least
    $18  N \log n $ edges do not contribute any value to $\Par$. Now we consider the edges incident to vertices in $W=N(U^u \cup U^{u'}) \setminus (I_{u} \cup I_{u'})$. Observe that, in the worst case, all the edges between vertices of $W$ and $U^u \cup U^{u'}$ are included in $E(\Par)$ while none of them is included in $\Par'$. Also, any edge that is not incident to  $U^u \cup U^{u'}$ is either included in both $E(\Par)$ and $E(\Par')$ or in none of them.
    Notice that any vertex in $U^u$ (respectively in $U^{u'}$) has $2 N \cdot \sizeIvu  \log n $ neighbors in $V(G_{i,j})\setminus I_{u}$ (resp. in $V(G_{i',j'})\setminus I_{u'}$). Furthermore, $u$ (resp. $u'$) has at most $4 N \log n$ neighbors in $V_C$.
    Also, there are no other neighbors of these vertices to be considered.
    Therefore, in the worst case, $|E(\Par) \setminus E(\Par')| = 2N \cdot \sizeIvu \log n + 8 N \log n = 68 N \log n$. Since $90 N \log n > 68 N \log n$ we have that $v(\Par')>v(\Par)$,  contradicting the optimality of $\Par$.
\end{proof}

\begin{lemma}\label{lemma_assignments_have_their_neighbors}
    Let $\Par = \{C_1,\ldots, C_p\}$ be an optimal $\C$-partition of $G$.
    For any $(i,j) \in [3]\times [\log n]$ and $u \in U_{i,j}$, if $u \in C$ then any $v\in N(u) \cap V(G_{i,j})$ also belongs in $C$. 
\end{lemma}
\begin{proof}
    Assume that for an $(i,j) \in [3]\times [\log n]$ there exists a $u \in U_{i,j}$ and a $w \in (N(u) \cap V(G_{i,j}) )$ such that $u \in C_k$ and $ w \notin  C_k$. We will show that $\Par$ is not optimal.
    
    It follows from Lemma~\ref{lemma_assignment_with_its_copy} that $U^u \subseteq C_k$. We will distinguish the following two cases: either $|C_k|< \C$ or not.

    \medskip
    
    \noindent\textbf{Case 1: $\boldsymbol{|C_k| < \C}$}.
    In this case, either $w\in I_{u}$ or $w \in I_{v,u}$ for some $v \in V_{i,j}$.
    Since $w$ has at most one neighbor that does not belong in $C_k$, moving $w$ to the partition of $C_k$ will create a
    $\C$-partition that includes more edges than $\Par$. This is a contradiction to the optimality of $\Par$.

    \medskip 
    
    \noindent\textbf{Case 2: $\boldsymbol{|C_k| = \C}$}. In this case, it is safe to assume that $G[C_k]$ is connected as otherwise we can partition it into its connected components. This does not change the value of the partition, and the resulting set that contains $u$ has a size less than $\C$.
    We proceed by considering two sub-cases, either $C_k \cap V_C \neq \emptyset$ or not.

    \smallskip 
    
    \noindent\textbf{Case 2.a: $\boldsymbol{C_k \cap V_C \neq \emptyset}$.} 
    We claim that either there exists a vertex $c \in C_k \cap V_C$ such that $u \notin N(c)$ or $G[C_k]$ has a leaf $x$ such that $u \notin N(x)$. Indeed, in the second case, $|C_k \cap V_C|$ (by construction) and the existence of $w$ is guaranteed by the fact that no other assignment vertex can be in $C_k$. In the first case we set $y = c$ while in the second $y=x$.
    We create a new partition as follows:
    \begin{itemize}
        \item we remove $y$ from $C_k$,
        \item move $w$ from its set to $C_k$ and
        \item add a new set $C= \{c\}$ in the partition. 
    \end{itemize}
    Let $\Par'$ be this new $\C$-partition. We have that $v (\Par')> v(\Par) $. Indeed, $w$ has at most one neighbor that does not belong in $C_k$. Therefore, moving $w$ to $C_k$ increases the number of edges by at least $4$ ($w$ is adjacent to all vertices of $U^u$ and $U^u \subseteq C_k$). We consider the case where $y$ is a vertex $c \in C_k \cap V_C$ such that $u \notin N(c)$. 
    Since $u$ is the only assignment vertex in $C_k$, and there are at most $3$ edges connecting $c$ to variable vertices, removing $c$ from $C_k$ reduces the value of $\Par$ by at most $3$. Therefore, $v (\Par')- v(\Par) = |E(\Par')|- |E(\Par)|\ge 1$. This is a contradiction to the optimality of $\Par$.
    Similarly, in the case where $y$ a leaf such that $u \notin N(y)$, removing $y$ from $C_k$ reduces the value of $\Par$ by at most $1$. This again contradicts the optimality of $\Par$.

    \smallskip 
    
    \noindent\textbf{Case 2.b: $\boldsymbol{C_k \cap V_C = \emptyset}$.}
    Since $G[C_k]$ is connected, $|N(u) \cup V_{i,j}| < \C$ and $C_k \cap V_C = \emptyset$, there exists a pair $(v,x) \in V_{i,j} \times U_{i,j} $
    such that $x \neq u$ and $C_k \cap I_{v,x} \neq \emptyset$. Also, by Lemmas~\ref{lemma_no_two_assignment_vertices} and~\ref{lemma_assignment_with_its_copy}, we have that $U^x \cap C_k = \emptyset$. Therefore, any vertex $y \in C_k \cap I_{v,x}$ contributes at most one edge in $E(\Par)$. We create a new partition as follows:
    \begin{itemize}
        \item select a vertex $y\in C_k \cap I_{v,x}$ and remove it from $C_k$,
        \item move $w$ for its set to $C_k$ and
        \item add a new set $C= \{y\}$ in the partition. 
    \end{itemize}
    This is a contradiction to the optimality of $\Par$, since the removal of $y$ from $C_k$ reduces the value of the partition by at most $1$, while moving $w$ to $C_k$ increases the value by at least $4$. 
\end{proof}

Summing up the previous lemmas, we can observe that in any optimal $\C$-partition $\Par$ of $G$,
there is one component for each vertex $u \in \bigcup_{(i,j) \in [3]\times [\log n]} U_{i,j} $ and
if $u \in C$, for some $C \in \Par$, then $N(u) \setminus V_C \subseteq C$. 

\begin{lemma}\label{lemma_variables_have_assignments}
    Let $\Par = \{C_1,\ldots, C_p\}$ be an optimal $\C$-partition of $G$. 
    For any $(i,j) \in [3]\times [\log n]$ and $v \in V_{i,j}$, if $v \in C_k$, for some $k \in [p]$, then $|C_k \cap U_{i,j}|=1$.
\end{lemma}
\begin{proof}
Recall that by Lemma~\ref{lemma_no_two_assignment_vertices}, $|C_k \cap U_{i,j}|$ is either $1$ or $0$.
Assume that there exist $(i,j) \in [3]\times [\log n]$ and $v \in V_{i,j}$ such that $v \in C_k$ and $|C_k \cap U_{i,j}|=0$.
By this assumption and Lemma~\ref{lemma_assignments_have_their_neighbors}, we can conclude that $N(v) \cap C_k \subseteq V_C$.
Also, since each variable has at most $4$ appearances and $v$ represents at mots $\frac{\log n }{2}$ variables, 
we have that $|N(v) \cap C_k| \le 2 \log n$.

Let $u \in U_{i,j}$ be an arbitrary assignment vertex.
Also, let $C_\ell \neq C_k$ be the set of $\Par$ such that $u \in C_\ell$.
By Lemma~\ref{lemma_assignments_have_their_neighbors}, we know that $C_\ell \cap I_{v,u} = I_{v,u} $.
Now we distinguish two cases: either $|C_\ell|<\C$ or $|C_\ell|= \C$. 

\medskip 

\noindent\textbf{Case 1: $\boldsymbol{|C_\ell|<\C}$.} We create a new partition as follows:
\begin{itemize}
    \item remove $v$ for $C_k$ and 
    \item add $v$ for $C_\ell$.
\end{itemize}
Let $\Par'$ be the new partition; notice that this is a $\C$-partition as $|C_\ell \cup \{v\}| <\C+1$. Also, the removal of $v$ from $C_k$ reduces the value of the partition by at most $2\log n$ while the addition of $v$ to $C_k$ increases the value by $\sizeIvu \log n$. This is a contradiction to the optimality of $\Par$.

\medskip 

\noindent\textbf{Case 2: $\boldsymbol{|C_\ell|=\C}$.} Similarly to the proof of Lemma~\ref{lemma_assignments_have_their_neighbors}, we assume that $G[C_\ell]$ is connected.
Also, since any set $I_{v',x}$, for $(v',x) \in V_{i,j}\times U_{i,j}$ is a subset of the set of the partition that includes $x$, we have that $C_\ell \cap V_C \neq \emptyset$. Indeed, assuming otherwise we get that either $|C_\ell|<\C$ or $G[C_\ell]$ is not connected. We create a new partition as follows:
    \begin{itemize}
        \item select (arbitrarily) a vertex $c\in C_\ell \cap V_C$ and remove it from $C_\ell$,
        \item move $v$ from $C_k$ to $C_\ell$ and
        \item add a new set $C= \{c\}$ in the partition. 
    \end{itemize}
    We will show that the value of the new partition is greater than the original. First, notice that $c$ has at most four neighbors in $C_\ell$, as $C_\ell$ can include only one assignment vertex, and $v$ has at most $2\log n$ neighbors in $C_k$, as $N(v) \cap C_k \subseteq V_C$). Therefore, removing $c$ from $C_\ell$ and $v$ from $C_k$ reduces the value of the partition by at most $2 \log n + 4$. Also, since $u\in C_\ell$ and by Lemma~\ref{lemma_assignments_have_their_neighbors}, we get that $I_{v,u}\subseteq C_\ell$. Thus, moving $v$ into $C_\ell$ increases the value of the partition by $|I_{v,u}|= \sizeIvu \log n > 2\log n +4$. This is a contradiction to the optimality of $\Par$.
\end{proof}

Next, we compute the minimum and maximum values that any optimal $\C$-partition of $G$ can admit. 

\begin{lemma}\label{lemma_value_of_optimal_partition}
    Let $\Par = \{C_1,\ldots, C_p\}$ be an optimal $\C$-partition of $G$. 
    We have that $ 3N \log^2 n (180 \sqrt{n}+ 6 )  \le v(\Par) \le  3N \log^2 n(180 \sqrt{n}+ 6 )  + 2m $, where $m=|V_C|$.
    Furthermore, if a vertex $c \in V_C$ belongs to $C \in \Par$ and $|N(c) \cap C|=2$, then $N(c) \cap C= \{v,u\}$, where $v\in V_{i,j}$ and $u\in U_{i,j}$, for some $(i,j)\in[3]\times [\log n]$.
\end{lemma}

\begin{proof}
    First, we calculate the number of edges that $E(\Par)$ includes from any $G_{i,j}$. Notice that $W = V_{i,j} \cup U_{i,j}$ is a vertex cover of $G_{i,j}$ and no edge is incident to two vertices of this set. Therefore, we can compute $|E(\Par) \cap E(G_{i,j})|$ by counting the edges of $\Par$ that are incident to a vertex of $W$. First, for any vertex $u \in U_{i,j}$, if $u \in C$, for a $C \in \Par$, we have that $N(u)\cap V_C \subseteq N(u) \cap C$. Also, we know that $U^u \subseteq C$. Therefore, all the edges that are incident to vertices in $U^u$ are in $E(\Par)$. So, for each $u\in U_{i,j}$ we have $5 (\sizeIu + 2 \cdot \sizeIvu)N \log n = 180 N \log n$ edges in $E(\Par)$ that are incident to vertices in $U^u$. Also, it follows by Lemma~\ref{lemma_variables_have_assignments} that for any vertex $v \in V_{i,j}$, there exists a (unique) $u \in U_{i,j} $
    such that $\{v,u\} \subseteq C$ for some $C \in \Par$. Furthermore, by Lemma~\ref{lemma_assignments_have_their_neighbors}, we have that $N(v)\cap C\subseteq I_{v,u}$. Thus, for each $v \in V_{i,j}$, the set $E(\Par)$ includes $\sizeIvu \log n$ edges and no other edge (from $E(G_{i,j})$) is incident to it.
    Since we have not counted any edge more than once, we have that
    $|E(\Par) \cap E(G_{i,j})| = (180 \sqrt{n} +6)  N \log n$ for any $(i,j) \in [3]\times [\log n]$. 
    Therefore, we have that
    $|  E(\Par) \cap \bigcup_{(i,j) \in [3]\times [\log n]} E(G_{i,j}) | =  3 N \log^2 n (180 \sqrt{n} +6)$.

    Since there are no edges between $V(G_{i,j})$ and $V(G_{i',j'})$ for $(i,j) \neq (i',j')$, 
    it remains to count the edges incident to vertices of $V_C$. For any $(i,j) \in [3]\times [\log n]$ and any $c \in V_C$, we have that
    $|N(c) \cap V_{i,j}| \le 1$ as the clause represented by $c$ has at most one variable from the vertex set $X_i$ and the vertices of any $V_{i,j}$ represent variables from $X_i$. Assume that $c \in C$, for $C \in \Par$. If $C \cap U_{i,j} = \emptyset$ for all $(i,j) \in [3]\times [\log n]$, then $c$ has no neighbors in $C$. Indeed, by Lemma~\ref{lemma_variables_have_assignments} we have that any variable vertex appears in the same set as one assignment vertex. Now, assume that $C$ includes a $u \in U_{i,j}$ for some $(i,j) \in [3]\times [\log n]$.
    By Lemma~\ref{lemma_no_two_assignment_vertices}, there is no other assignment vertex in $C$.
    Also, by Lemma~\ref{lemma_variables_have_assignments}, only variable vertices from $V_{i,j}$ can be in $C$. Therefore, $c$ has at most $2$ neighbors in $C$ (one variable vertex and one assignment vertex).
    Since the sets of edges are disjoint, we have at most $2$ extra edges per clause vertex $c \in V_C$. This concludes the proof of this lemma.
\end{proof}

We are now ready to prove our result.
\thmVcToTheVc*

\begin{proof}

    Let $\phi$ be the formula that is given in the input of the \textsc{R$3$-SAT} problem, and let $G$ be the graph constructed from $\phi$ as described before. We will show that $\phi$ is satisfiable if and only if $G$ has a $\C$-partition of value
    $ 3 N \log^2 n (180 \sqrt{n} +6) + 2m $, where $\C = \sizeCar N \log n  $ and $N=\lceil n / \log^2 n \rceil $.
    
    Assume that $\phi$ is satisfiable and let $\alpha: X \rightarrow \{true, false\}$ be a satisfying assignment. We will construct a $\C$-partition of $G$ of the wanted value.

    First, for each assignment vertex $u$, create a set $C_u = U^u \cup (N(u)\setminus V_C)$.
    We then extend these sets as follows.
    Consider a variable vertex $v$ and restrict the assignment $\alpha$ on the vertex set $X(v)$.
    By construction, there exists an assignment vertex $u$ that represents this restriction of $\alpha$. 
    Notice that there may exist more than one such vertices; in this case we select one of them arbitrarily.
    We add $v$ into the set $C_u$ that corresponds to $u$.
    We repeat the process for all variable vertices. 
    Next, we consider the vertices in $V_C$. Let $c\in V_C$ be a vertex that represents a clause in $\phi$.
    Since $\alpha$ is a satisfying assignment, there exists a literal in this clause that is set to true by $\alpha$.
    Let $x$ be the variable of this literal. 
    We find the set $C_u$ such that $v \in C_u$ and $x \in X(v)$. 
    We add $c$ in $C_u$, and we repeat this for the rest of the vertices in $V_C$.

    We claim that the partition $ \Par =  \{C_u \mid u$ is an assignment vertex$\}$ is an optimal $\C$-partition of $G$.
    We first show that this is indeed a $\C$-partition. By construction, for any $C \in \Par$ we have a pair $(i,j) \in [3]\times [\log n]$ and a vertex $u \in U_{i,j}$ such that $C \subseteq V_{i,j} \cup N[U^u] \cup V_C$.
    Notice that $|V_{i,j} \cup N[U^u]| = 2N + 2N \cdot \sizeIvu \log n + \sizeIu N \log n +5 $.
    We now calculate $|C \cap V_C|$. By construction, if $c \in C \cap V_C$, there exists a vertex $v \in V_{i,j}$ such that $v \in C$.
    Therefore, $N(V_{i,j}) \cap V_C \supseteq C\cap V_C$.
    Since each $v \in V_{i,j}$ represents $\frac{\log n}{2}$ variables and each variable appears in at most $4$ clauses, we have that
    $|N(V_{i,j}) \cap V_C| \le |V_{i,j}| 2 \log n  \le 4 N \log n$.
    Thus $|C| \le 2N + 2N \sizeIvu \log n + \sizeIu N \log n + 4N \log n +5 < 42 N \log n =\C $ for sufficiently large $n$.

    We now need to argue about the optimality of $\Par$. Using the same arguments as in Lemma~\ref{lemma_value_of_optimal_partition}, we can show that $E(\Par) \cap E(G_{i,j})$ includes exactly $3 N \log n (180 \sqrt{n}  +6) $ edges. Thus,
    $|  E(\Par) \cap \bigcup_{(i,j) \in [3]\times [\log n]} E(G_{i,j}) | =  3 N \log^2 n (180 \sqrt{n}  +6)$.
    Therefore, we need to show that there are $2m$ additional edges in $E(\Par)$ that are incident to vertices of $V_C$.
    Notice that, for any $c \in V_C$, there exists a $C_u$ such that $c \in C_u$ and there exist vertices $v,u$ in $C_u$ that are both incident to $c$ (which holds by the selection of $C_u$), with $v$ being a variable vertex and $u$ an assignment vertex.
    Finally, by construction,  there are at most $2$ edges incident to $c$ in $E(\Par)$.
    Therefore, $v(\Par)=3 N \log^2 n (180 \sqrt{n} +6 )  + 2m$.

    For the reverse direction, assume that we have a $\C$-partition $\Par$ of $G$, with $v(\Par)=3 N \log^2 n [180 \sqrt{n} +6 ]  + 2m$.
    By Lemma~\ref{lemma_value_of_optimal_partition} we have that each vertex $c \in V_C$ must be in a set
    $C \in \Par$ such that: 
    \begin{itemize}
        \item  $|N(c) \cap C|=2$ and 
        \item there exist $(i,j) \in [3]\times [\log n]$ such that $v \in V_{i,j} \cap C$, $u \in U_{i,j} \cap C$  and $\{v,u\}\subseteq N(c)$.
    \end{itemize}
    We construct an assignment $\alpha$ of $\phi$ that corresponds to this partition as follows. For each variable $x$, consider the variable vertex $v$ such that $x \in X(v)$.
    By Lemma~\ref{lemma_variables_have_assignments} there exists a unique assignment vertex $u$ such that $v$ and $u$ belong in the same component of $\Par$. Let $\sigma_{v,u}$ be the assignment represented by $u$ for $X(v)$.
    We set $\alpha (x)=\sigma_{v,u}(x)$. 
    Notice that each variable appears in the set of one variable vertex and for each such vertex we have selected a unique assignment (represented by the assignments vertex in its set). Therefore the assignment we create in this way it is indeed unique.

    We claim that $\alpha$ is a satisfying assignment.
    Consider a clause of $\phi$ and assume that $c$ is the corresponding clause vertex in $V_C$.
    Assume that $c \in C$ for some $C\in \Par$. By Lemma~\ref{lemma_value_of_optimal_partition} we have that $|N(c) \cap C|=2$ and
    there exist $(i,j) \in [3]\times [\log n]$ such that $v \in V_{i,j} \cap C$, $u \in U_{i,j} \cap C$  and $\{v,u\}\subseteq N(c)$. 
    Since $v \in N(c)$, we know that there exists a variable $x\in X(v)$ that appears in a literal $l$ of the clause represented by $c$.
    Observe that $v$ is unique. Moreover, since $u,v\in V(G_{i,j})$, and $u\in N(c)$, we have that $\sigma_{v,u}(l)=\alpha(l)$ satisfies the clause represented by $c$.
    This finishes the reduction.

    \medskip

    The last thing that remains to be done is to bound $\vc(G)=\vc$, \emph{i.e.}, the size of the vertex cover number of $G$, appropriately. Notice that the vertex set containing the $V_{i,j}$s, the $U_{i,j}$s and the copies of the vertices in the $U_{i,j}$s, for every $(i,j) \in [3]\times [\log n]\}$, is a vertex cover of the graph. Therefore, $\vc \le 3 \log n (2N + 5\sqrt{n})\in \bO(\frac{n}{\log n}) $. Additionally, $\C \in \bO(\frac{n}{\log n})$.

    To sum up, if we had an algorithm that computed an optimal solution of the \CCF~problem in time $(\C \vc)^{o (\C+\vc)}$, we would also solve the \textsc{\textsc{R$3$-SAT}} problem in time $\big(\frac{n}{\log n}\big)^{o(\frac{n}{\log n})}$.
    This contradicts the ETH since
    $ \Big( \frac{n}{\log n}\Big)^{o(\frac{n}{\log n})} = 2^{ (\log n - \log \log n ) o(\frac{n}{\log n}) }= 2^{ o( n - \frac{n\log \log n}{\log n}) }=2^{ o( n ) }. $
\end{proof}

\section{Kernelization}
In this section, we consider the kernelization complexity of \CCF and \CCFw parameterized by the vertex cover number of the input graph and establish that following contrasting result. While the weighted and unweighted versions of this problem have  similar asymptotic behavior from a parameterized complexity point of view, for the parameters $\tw$ and $\vc$, the  kernelization complexity exhibits a stark contrast between the two versions. This signifies that weights present a barrier from the kernelization complexity point of view. In particular, we establish that while \CCF parameterized by $\vc+\C$ admits a  $\mathcal{O}(\vc^2\C)$ vertex kernel, \CCFw parametrized by $\vc+\C$ cannot admit any polynomial kernel unless the polynomial hierarchy collapses.  

\subsection{Polynomial Kernel for \CCF}
In this section, we establish that \CCF~admits a polynomial kernel parameterized by $\vc+\C$. We will use an auxiliary bipartite graph $H$ that we construct as follows. Let $U$ be a vertex cover in $G$ and let $I=V(G)\setminus U$.  Then, $V(H)$ contains two partitions $X$ and $Y$ such that $V(Y) = I$ and for each $u\in U$, we add $t = \vc\times \C + \C$ many vertices $u_1,\ldots,u_t$. Moreover, if $uv \in E(G)$ such that $u\in U$ and $v\in I$, we add the edge $u_iv$ in $H$ for each $i\in [t]$. Now, we compute a maximum matching $\mathcal{M}$ in $H$. Let $Y'\subseteq Y$ be the set of vertices that are not matched in $\mathcal{M}$. We have the following reduction rule (RR).

\smallskip
\noindent\textbf{(RR):} Delete an arbitrary vertex $w\in Y'$ from $I$.
\smallskip

\begin{lemma}\label{L:safe}
    RR is safe.
\end{lemma}
\begin{proof}
    First, observe that in any $\C$-partition $\Par$ of $G$, at most $\vc\times \C$ many vertices can participate 
    in sets $C\in \Par$ such that $C\cap U \neq \emptyset$ and these are the only vertices of $I$ that can contribute 
    in the value of $\Par$. 

    Now, let $G' = G[V(G)\setminus \{w\}]$. 
    Since any $\C$-partition $\Par'$ of $G'$ can be easily extended to a $C$-partition of $G$ by adding to it a 
    singleton set $C=\{w\}$, it suffices to show that the value of the optimal partition of $G$ and value of 
    the optimal partition of $G'$ are equal.
    
    Let $\Par = \{C_1,\ldots,C_p\}$ be an optimal $\C$-partition of $G$. 
    We claim that there exists a $\C$-partition $\Par^*$ such that $\{w\} \in \Par^*$ 
    and $v(\Par^*) = v(\Par)$. 
    Notice that, by proving that $\Par^*$ exists, we also prove that any optimal $\C$-partition of $G'$ has the same value as 
    any optimal $\C$-partition of $G$. Indeed, $\Par^*\setminus \{w\}$ is a $\C$-partition of $G'$ and $v(\Par^*) = v(\Par^*\setminus \{w\}) = v(\Par)$; thus $\Par^*\setminus \{w\}$ is a $\C$-partition of $G'$ (otherwise, $\Par$ is not an optimal $\C$-partition of $G$). 
    It remains to prove that such a $\C$-partition exists.
    
    In the case that $\{w\}$ is a singleton in $\Par$ then $\Par^* = \Par$. 
    Therefore, we assume that $w$ participates in some set $C \neq \{w\}$ of $\Par$. 
    Let $x\in U\cap C$ such that $xw\in E(G)$. Then, observe that $x_1,\ldots,x_t$ are matched to 
    $t = \vc \times \C +\C$ many vertices of $I$ by the maximum matching in $H$ 
    (as $w \notin Y$); let $S_x$ be the set of these vertices. 

    Observe that at least $\C$ of the vertices in $S_x$ are not contributing to the value of $\Par$ 
    (since at most $\vc\times \C$ many vertices can participate in sets $C\in \Par$ such that $C\cap U \neq \emptyset$). 
    We create a new set $C'$ by moving $\C-1$ of these vertices (which contain vertices that are connected to $x$) into $C'$ 
    and move $x$ from $C$ to $C'$. 
    Observe that after this step, we have that $v(\mathcal{P}') \geq v(\mathcal{P})$ as we remove at most $\C-1$ edges incident on $x$ in $C$ and add exactly $\C-1$ edges incident on $x$ in $C'$. 
    We can keep repeating this step until $w$ is no longer connected to any vertex in $C$, and at this point, we can delete $w$ from $C$ (since its contribution is 0) and add it as a singleton. Hence, we get the $\mathcal{P}^*$.
This completes our proof.
\end{proof}

It is straightforward to see that once we cannot apply RR anymore, $|V(G)| =\mathcal{O} (\vc^2 \C)$. This, along with \Cref{L:safe},  imply the following theorem.
\thmPolyKernel*



\subsection{Incompressibility of \CCFw}
Unfortunately, the above approach only works in the unweighted case. In fact, in the rest of this subsection we show that the existence of a polynomial kernel parameterized by $\vc$ and $\C$ is highly unlikely for the general weighted case. This is achieved through a reduction from the \textsc{$k$-Multicolored Clique} problem.

\begin{figure}[!t]
\centering
\includegraphics[scale=0.7]{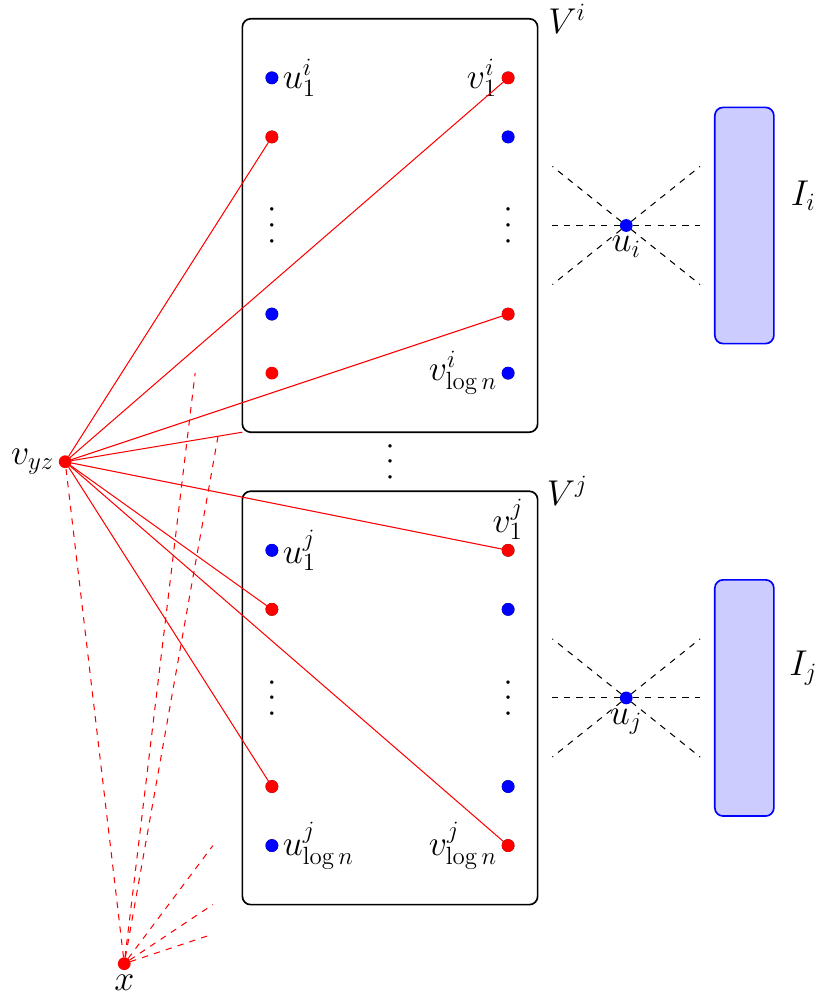}
\caption{The graph $G$ constructed in the proof of \Cref{thm:no-kernel-vc}. The vertex $v_{yz}$ represents the edge of $H$ with endpoints the vertices $y$ and $z$ such that $y\in V_i$ and $z\in V_j$ (for some $i<j\in[k]$). The red edges joining the vertex $v_{yz}$ to the red vertices of $V^i$ and $V^j$ are according to the selection set of $y$ and $z$ respectively. That is, if $y$ is the $k^{th}$ vertex of $V_i$ (according to the arbitrary enumeration that was chosen), then the first digit of the binary representation of $k$ would be a $1$ (there is an edge towards $v^i_1$), the second a $0$ (there is an edge towards $u_2^i$) and so on. The color blue is used to indicate which vertices would join the coalitions defined by the $u_i$s.}\label{fig:no-poly-kernel}
\end{figure}

\iflong

\Pb{\textsc{$k$-Multicolored Clique}}{A graph $H=(V,E)$ and a partition $(V_1,\ldots,V_k)$ of $V$ into $k$ independent sets}{Question}
{Does there exists a set $S\subseteq V$ such that $G'[S]$ is a clique?}

We may additionally assume that $|V_1|= \ldots =|V_k|=n = 2^m$ for some $m \in \mathbb{N} $ as, otherwise, we can add independent vertices in each set. It is known that \textsc{$k$-MC} does not admit a kernel of size $poly (k + \log n)$, unless the Polynomial Hierarchy collapses~\cite{HermelinKSWW15}. We proceed by describing the construction of our reduction, provide a high-level idea of the proof, show the needed properties and finish by proving the theorem. 

\medskip

\noindent\textbf{The construction.} Given a graphs $H$ as an input for the \textsc{$k$-Multicolored Clique} problem, we construct an instance of \textsc{$\C$-Coalition Formation}, for $\C = {k\choose 2} +k\log n+1$, as follows (illustrated in Figure~\ref{fig:no-poly-kernel}).
For each set $V_i$, we first create a clique of $2\log n$ vertices $V^i =\{u_j^i, v_j^i \mid j \in[\log n] \}$. 
We proceed by creating a vertex $u_i$ and an independent set $I_i$ of size $\C - \log n -1$. 
Finally, we add all edges between $u_i$ and vertices from $V^i \cup I_i$.

Before we continue, we will relate the vertices for each set $V_i$ with a subset of vertices of $V^i$. Let $v_i,\ldots, v_n$ be an enumeration 
of the vertices in $V_i$. We assign to each $v_j \in V_i$ a binary string of length $\log n$ such that the string assigned to $v_j$ represents the number $j$ in binary form. Let $s(v)$ be the string assigned to a vertex $v$ of the original graph. 
Also, for each $i$ and $v \in V_i$, we use $s(v)$ in order to define a set $S(v) \subseteq V^i$ as follows; for each $\ell \in [\log n]$,
\begin{itemize}
    \item if the $\ell$-th letter of $s(v)$ is $0$, we add $u_{\ell}^i$ in $S(v)$,
    \item otherwise, we add $v_{\ell}^i$ in $S(v)$.
\end{itemize}

We continue by creating one vertex $v_e$ for each edge $e \in E$. We call this set $V_e$. 
We add edges incident to $v_e$ as follows: 
Let $e = uv$ where $u \in V_i$ and  $v\in V_j$ for some $i,j \in [k]$; notice that $i \neq j$. 
We add all edges $v_e w$, where $w \in S(u) \cup S(v)$. 

Finally, we add a vertex $x$ and we add all the edges between $x$ and the vertices of $V_e \cup  \bigcup_{i\in [k]} V^i$.
We will call this new graph $G$.

We complete the construction by defining the weight function $w:E(G) \rightarrow \mathbb{N}$ as follows. 
\begin{itemize}
    \item For any edge $e = u_i v$ where $v \in I_i$, we set $w(e) = \weightStarsRight$. 
    \item For any edge $e = u_i v$ where $v \in V^i$, we set $w(e) = \weightStarsColor$. 
    \item For any edge $e = x v$ where $v \in V^i$, we set $w(e) = \weightXColor$. 
    \item For any other edge $e$ we set $w(e) = 1$.
\end{itemize}


\medskip 

\noindent\textbf{High-level description.} The rest of the proof consists in showing that there is a clique of order $k$ in $H$ if and only if there exists a $\C$-partition of $G$ with $v(\Par)$ above some threshold. Intuitively, the weights of the edges incident to the $u_i$s are chosen so that every set of $\Par$ contains at most one of the $u_i$s. Also, if a set of $\Par$ contains a $u_i$, then it also contains $\log n$ vertices from $V^i$; the remaining $\log n$ vertices of $V^i$ are the encoding of a vertex $v\in V_i$. We then build an equivalence between locating the vertices in the $V_i$s that can form a clique in $H$ and having a set in $\Par$ that contains $x$, $\ell={k\choose 2}$ ``edge'' vertices $e_{u_1,v_1},\dots,e_{u_\ell,v_\ell}$ and the $k\log n$ vertices contained in $S(u_1),\dots,S(u_\ell)$ and $S(v_1),\dots,S(v_\ell)$ that correspond to the encodings of the vertices $u_1,\dots,u_\ell$ and $v_1,\dots,v_\ell$.

\medskip 

\noindent\textbf{Properties of optimal $\C$-partitions of $G$.} Before we proceed with the reduction, we will prove some properties for any optimal partition of $G$.

\begin{lemma} \label{lemma:stars:together}
    Let $\Par = \{C_1\ldots, C_p\}$ be an optimal $\C$-partition of $G$. For any $i \in [k]$, there exists a $j \in [p]$ 
    such that $ u_i \cup I_i \subseteq C_j $.
\end{lemma}
\begin{proof}
    Assume that there exists an $i\in [k]$ such that $ C_j \cap (u_i \cup I_i) \neq ( u_i \cup I_i ) $, for any $j \in [p]$.
    Let $C$ be the set $\{u_i\} \cup I_i$. We claim that the partition $\Par' = \{C_1\setminus C, \ldots C_p\setminus C, C\}$ 
    has higher value than $\Par$. 
    Indeed, by separating $C$ from the rest of the partition, we only lose the weights of the edges incident to $u_i$ and vertices of $V^i$ (as all the other neighbors of $u_i$ are in $C$). Since all these edges have weight $\weightStarsColor$, we are reducing the value of the partition by at most $\C \weightStarsColor$.
    On the other hand, notice that there exists at least one edge $e = u_i v$ for some $v \in I_i$ such that $w(e)$ is counted in $\Par'$ but not in $\Par$.
    Since $\weightStarsRight > \C \weightStarsColor$, we have that $v(\Par') > v(\Par)$, which contradicts the optimality of $\Par$.
\end{proof}

\begin{lemma} \label{lemma:stars:have:log:colors}
    Let $\Par = \{C_1\ldots, C_p\}$ be an optimal $\C$-partition of $G$. If $u_i \in C_j$, for some $(i,j)\in [k]\times[p]$, then $| C_j \cap V^i | = \log n$. 
\end{lemma}
\begin{proof}
    Assume that there exists a $u_i \in C_j$, for some $(i,j)\in [k]\times[p]$, such that $| C_j \cap V^i | < \log n$. 
    Notice that $| C_j \cap V^i | $ is at most $ \log n$ by Lemma~\ref{lemma:stars:together}. 
    Select (arbitrarily) a set $U$ such that $V^i \supset U\supset C_j \cap V^i$ and $|U| = \log n$. 
    We set $C= U \cup I_i \cup \{u_i\}$. 
    Them, we create the partition $\Par' = \{C_1\setminus C,\ldots, C_p\setminus C, C \}$. We claim that $v(\Par')> v(\Par)$.
    Indeed, there is at least one edge $u_i v \in E(\Par') \setminus E(\Par)$. Also, this edge has weight $\weightStarsColor$. Now, consider an edge in $e \in E(\Par) \setminus E(\Par')$. It is not hard to see that $w(e)=1$ or $w(e) = \weightXColor$. 
    Also, since any edge with weight $\weightXColor$ is incident to $x$, we may have less than 
    $\C-1$ such edges in $E(\Par) \setminus E(\Par')$. Thus, the total weight of the edges in $E(\Par) \setminus E(\Par')$ is less than $\C \weightXColor + {{\C} \choose {2}} < \weightStarsColor$. Therefore, $v(\Par')> v(\Par)$, which contradicts the optimality of $\Par$.
\end{proof}

\begin{lemma} \label{lemma:x:with:the:remaining}
    Let $\Par = \{C_1\ldots, C_p\}$ be an optimal $\C$-partition of $G$ and $x \in C_\ell$, for an $\ell \in [p]$. 
    Then $| C_\ell \cap V^i | = \log n$ for all $i \in [k]$.
\end{lemma}
\begin{proof}
    
    It follows from Lemmas~\ref{lemma:stars:together} and~\ref{lemma:stars:have:log:colors} that, for each $i\in [k]$, we have exactly $\log n$ vertices from $V_i$ that are in the same set as $\{u_i\} \cup I_i$. Let $C^i\in \Par$ be the set that includes $\{u_i\}$. 
    Observe that $|C^i|=\C$. Thus, no other vertex has been included to $C_i$. Therefore, for all $i \in [k]$, there are exactly $\log n $ vertices that are not in the same set as $u_i$; let $S_i$ be this subset of $V^i$. That is, $S_i = V^i \setminus C^i$. 
    We will show that, for all $i \in [k]$, we have that $S_i \subseteq C_\ell$. 
    Assume that there exists an $i$ such that $S_i 	\nsubseteq C_{\ell}$ and let $u \in S_i \setminus C_\ell$. 
    We consider two cases, either $|C_{\ell}|<\C $ or not.

    \smallskip 
    
    \noindent\textbf{Case 1 $\boldsymbol{|C_{\ell}|<\C }$:} Then we create the following partition.
    \begin{itemize}
        \item Remove $u$ from its current set and
        \item add $u$ to $C_{\ell}$.
    \end{itemize}
    Since $u$ was not in the same set as $u_i$ or $x$ in $\Par$, any edge $e \in E(\Par)$ that is incident to $u$ has weight $1$. 
    Also, since $\Par $ is a $\C$-partition, we have at most $\C-1$ neighbors of $u$ in the same set as $u$ in $\Par$. Thus, moving $u$ to a different set reduces the value of the partition by at most $\C-1$. On the other hand, in $E(\Par') $ we have at least included the edge $xu$ and $w(xu) = \weightXColor > \C$. This contradicts the optimality of $\Par$.

    \smallskip 
    
    \noindent\textbf{Case 2 $\boldsymbol{|C_{\ell}|=\C }$:} Then we have at least one edge vertex $v_e$ in $\C_{\ell}$. We create a new $\C$-partition by swapping $u$ and $v_e$. 
    Again, moving $u$ to a different set reduces the value of the partition by at most $\C-1$. Also, recall that by construction, $d(v_e)=2\log n + 1$ and all of theses edges are of weight $1$. Therefore, moving $v_e$ to a different set reduces the value of the partition by at most $2\log n + 1$. 
    The fact that $E(\Par') $ includes at least the edge $xu$ and $w(xu) = \weightXColor > \C +2\log n + 1$,  contradicting to the optimality of $\Par$.  
\end{proof}

We are now ready to prove our result. 

\thmNoKernelVc*
\begin{proof}
Let $H$ be the input graph of the \textsc{$k$-Multicolored Clique} problem, and $G$ be the graph constructed from $H$ as described above. We will show that $H$ has a clique of order $k$ if and only if any optimal $\C$-partition $\Par$ for $G$ has value at least
$v(\Par) = k \big( \weightStarsRight(\C - \log n-1) + (\weightStarsColor + \weightXColor ) \log n + 2 \binom{\log n}{2} \big) +  \binom{k}{2} (2\log n +1) $.

Assume that $H$ has a clique of order $k$ and let $v^i$ be the vertex of this clique that also belongs to $V_i$, for each $i \in [k]$.
For each $i \in [k]$, we create the set $C_i = \{u_i\}\cup I_i \cup (V_i \setminus S(v^i))$. 
Then we create a set $C = \{x\} \cup \bigcup_{i \in [k]}S(v^i) \cup \{v_e \mid e = v^iv^j \text{ for all } 1\le i < j \le k \}$. 
Finally we add one set for each remaining vertex $v_e$.
Let $\Par= \{C_1,\ldots,C_p\}$, $p>k+1$, be the resulting partition. 
We claim that $v(\Par) = k \big( \weightStarsRight(\C - \log n-1) + (\weightStarsColor + \weightXColor ) \log n + 2 \binom{\log n}{2} \big) +  \binom{k}{2} (2\log n +1) $.

Indeed, for any $i \in [k]$, the sum of the weights of the 
edges of $G [ C_i ]$ is exactly $\weightStarsRight(\C - \log n-1) + 
\weightStarsColor \log n +  \binom{\log n}{2}$. 
Also, by construction, the sum of the weights of the edges of 
$G [ C ]$ is exactly 
$ k \weightXColor \log n +  k \binom{\log n}{2} + \binom{k}{2} (2\log n +1)$. 
Finally, all the other sets are singletons. Thus $v(\Par) = k \big( \weightStarsRight(\C - \log n-1) + (\weightStarsColor + \weightXColor ) \log n + 2 \binom{\log n}{2} \big) +  \binom{k}{2} (2\log n +1) $.

For the reverse direction, assume that we have a partition $\Par$ that has value $v(\Par) = k \big( \weightStarsRight(\C - \log n-1) + (\weightStarsColor + \weightXColor ) \log n + 2 \binom{\log n}{2} \big) +  \binom{k}{2} (2\log n +1) $.
By Lemmas~\ref{lemma:stars:together} and~\ref{lemma:stars:have:log:colors}, we know that, for each $i \in [k]$, there exists a set $C \in \Par$ such that 
$C \supseteq \{u_i\} \cup I_i$ and $C \setminus (\{u\} \cup I_i) \subseteq V^i$. Let us reorder the sets of $\Par$ such that $\Par =\{C_1,\ldots , C_p\}$ and $u_i \in C_i$, for all $i \in [k]$.
First, we calculate the maximum value of $\{C_1,\ldots, C_k\}$. 
Notice that for any $i$, the set $C_i$ includes exactly $\log n$ vertices from $V^i$ and the set $\{u_i\} \cup I_i$. 
Therefore, we need to take into account:
\begin{itemize}
    \item $\binom{\log n }{2}$ edges of weight $1$, between the vertices of $V^i$, 
    \item $\log n$ edges of weight $\weightStarsColor$, between the vertices of $V^i$ and $u_i$ and 
    \item $\C - \log n $ edges of weight $\weightStarsRight$ between the vertices of $I_i$ and $u_i$. 
\end{itemize}
In total, this gives us a value of $\binom{\log n }{2} + \log n \weightStarsColor + (\C - \log n) \weightStarsRight$, and this holds for all $i \in [k]$.

By Lemma~\ref{lemma:x:with:the:remaining}, we also know that there exists a set $C$ in $\Par$ that includes the vertex $x$ together with the remaining 
vertices from the sets $V^i$, $i \in [k]$. Notice that $C$ may also include up to $\binom{k}{2}$ vertices from $V_e$.
Actually, $C$ must include all these vertices, as otherwise the edges incident to them will contribute nothing to the value of $\Par$. 

We will calculate the value of the edges in $E(C\setminus V_e)$. Notice that these edges are either between two vertices in the same set $V^i$ 
or between a set $V^i$ and $x$. Since for each $i \in [k]$ there are $\log n$ vertices from $V^i$, we have:
\begin{itemize}
    \item $\binom{\log n }{2}$ edges of weight $1$, between the vertices of $V^i$, for each $i\in [k]$ and
    \item $\log n$ edges of weight $\weightXColor$, between the vertices of $V^i$ and $x$, for each $i\in [k]$. 
\end{itemize}
Therefore, by adding these with the value from the sets $C_i$, $i \in [k]$, we calculate a value of $k \big( \weightStarsRight(\C - \log n) + (\weightStarsColor + \weightXColor ) \log n + 2 \binom{\log n}{2} \big)$. 

Observe that the assumed value of $\Par$ is higher than the one that we have calculate for the moment, by $\binom{k}{2} (2\log n +1) $.
Notice also that this extra value can be added only by the vertices from $V_e$ that can be in the same set as $x$.
Finally, any vertex $v \in V_e \cap C$ can contribute at most $2\log n +1$ since $d(v) = 2\log n +1$ and all these edges have weight $1$.
Therefore, in order to achieve the wanted value, we have that $|C \cap V_e| = \binom{k}{2}$ and $N(v) \subseteq C$ for each vertex $v \in C \cap V_e$.

Next, we will show that there is no pair $(i,j)$ for which there exist two edges $e,e'$ such that $\{v_e, v_{e'}\} \subseteq C$, $e  = uv$,  where $u \in V_i$ and $v \in V_j$, $e'  = u'v'$,  where $u' \in V_i$ and $v' \in V_j$. Notice that $N(v_e)\subseteq C$ and $N(v_e)= S(u) \cup S(v) \cup x$. Therefore, $C \cap V^i = S(u)$ and $C \cap V^j = S(v)$.
Since the same holds for $v_{e'}$, we can conclude that $S(u) = S(u')$ and $S(v)=S(v')$. Thus, $e = e'$. This cannot happen because these vertices represent edges of $H$ and there are no parallel edges in $H$. 
We can conclude that no two of vertices $v_e$ and $v_e'$ in $C$ can represent edges between vertices of the same sets. 
Also, since we have $\binom{k}{2}$ such vertices, for each pair $(i,j)$ we have a vertex $v_{e}$ that represents an edge $uv$ where $u \in V_i$ and $v \in V_j$.

Now, consider the set of vertices $U= \{v \in V(H) \mid S(v) = C\cap V^i \text{ for some } i \in [k] \}$. We claim that $U$ is a clique of order $k$ in $H$.
We will first show that for each $i \in [k]$, we have that $C\cap V^i = S(v)$ for some $v \in V_i$. 
As we mentioned, for each pair $(i,j)$ there exists one $e  = uv$  where $u \in V_i$, $v \in V_j$ and $v_e \in C$. 
Also, $N(v_e)\subseteq C$ and $N(v_e)= S(u) \cup S(v) \cup x$. Therefore, $C \cap V^i = S(u)$. 
Since this holds for any $i \in [k]$, we have that $U$ indeed represents a set of $k$ vertices in $H$.
It remains to show that $U$ induces a clique. Consider two vertices $u,v \in U$ and let $u \in V_i$ and $V_j$. 
Recall that for each pair $(i,j)$, we have a vertex $v_e\in C$ such that $e= u'v'$,  $u \in V_i$ and $v \in V_j$. 
Also, we have shown that $S(u') = C\cap V^i $ and $S(v') = C\setminus V^j$. Therefore, $S(u')= S(u)$ and $S(v')= S(v)$, from which follows that $e = uv$. 
Thus, there exists an edge between the two vertices. Since we have selected $u$ and $v$ 
arbitrarily, we have that $U$ is indeed a clique.

To fully prove the statement, it remains to be shown that the parameter that we are considering is bounded by a polynomial of $k+\log n$. 
Notice that the set $U = \{x\} \cup \bigcup_{i \in [k]} (V^i \cup \{u_i\})$ is a vertex cover of $G$. 
Also, $|V^i|=2\log n$ for all $i \in [k]$. Therefore, we have that $|U| \in \bO(k \log n)$. 
Recall that we have set $\C= 1 + \binom{k}{2} + k \log n$. Therefore, 
$ \vc+\C  \in poly (k+\log n)$.
\end{proof}
\fi

\section{Additional Structural Parameters}
In this section, we consider two additional parameters, vertex integrity and twin-cover number, of the input graph. In particular, we establish that while \CCF admit an FPT algorithm parameterized by the vertex integrity of the input graph, it remains \W[1]-hard when parameterized by the twin-cover number of the input graph.

\subsection{FPT Algorithm parameterized by vertex integrity}
In this section, we establish that \CCF admits an FPT algorithm when parameterized by the vertex integrity of the input graph in the following theorem.
\thmVi*
\ifshort
\begin{sketch}
Once we have a $U$ according to the definition of the vertex integrity, we can guess $\Par'$, the intersection of the sets of an optimal $\C$-partition of $G$ with $U$. Then, we organize the connected components of $G[V\setminus U]$ according to their structural characteristics; the number of different such types is bounded by a function of $k$. We reduce the problem of extending $\Par'$ into an optimal $\C$-partition of $G$ to solving an ILP with a number of variables that is bounded by a function of $k$. A solution of this ILP can be obtained in FPT time, parameterized by $k$ (by running for example, the Lenstra algorithm~\cite{Len83}).
\end{sketch}
\fi
\iflong
\begin{proof}
Let $U\subseteq V$ be such that $|U|=k'\leq k$ and $S_1,\ldots, S_m$ be the vertex sets of the connected components of $G[V\setminus U]$. It follows that $|S_j|\leq k$, $j\in[m]$. Let $\Par'=\{C'_1,\dots,C'_p\}$ be the strict restriction\footnote{a restriction is \textit{strict} if it only contains non-empty sets.} of an optimal $\C$-partition $\Par$ of $G$ on the set $U$ (there are at most $|U|^{|U|} \le k^k$ possible restrictions of $\Par$ on $U$). We will extend $\Par'$ into an optimal $\C$-partition of $G$. To do so, we will organize the connected components of $G[V\setminus U]$ into a bounded number of different types, and run an ILP.

We begin by defining the types. Two graphs $G_i=G[U \cup S_i]$ and $G_j=G[U \cup S_j]$, $i,j \in [m]$, are of the same \textit{type} 
if there exists a bijection\footnote{Recall that a function $f:A\rightarrow B$ is a \emph{bijection} if, for every $a_1,a_2\in A$ with $a_1\neq a_2$, we have that $f(a_1)\neq f(a_2)$ and for every $b\in B$, there exists an $a\in A$ such that $f(a)=b$. Recall also that the \emph{inverse} function of $f$, denoted as $f^{-1}$, exists if and only if $f$ is a bijection, and is such that $f^{-1}:B\rightarrow A$ and for each $b\in B$ we have that $f^{-1}(b)=a$, where $f(a)=b$.} 
$f: U \cup S_i \rightarrow U \cup S_j$ such that $f(u)=u$ for all $u\in U$ and $N_{G_i}(u) = \{ f^{-1}(v) \mid v \in N_{G_j}(f(u))\}$ for all $u \in S_i$. Note that if such a function exists, then $G_i$ is isomorphic to $G_j$.

Let $\mathcal{T}_1,\dots,\mathcal{T}_\ell$ be the different types that were defined. Observe that $\ell$ is at most a function of $k$ since $|U|\leq k$. For each $i\in[\ell]$, we define the \textit{representative of} $\mathcal{T}_i$ to be any connected component of $G[V\setminus U]$ that is contained in a graph of type $\mathcal{T}_i$; we will denote this graph by $G_{\mathcal{T}_i}$. For each $i\in[\ell]$, we will store a set of vectors $\tau^i_j$, for $j\in[q]$, which contain all possible orderings of all possible partitions of $V(G_{\mathcal{T}_i})$ into $p+k$ sets (some of which may be empty). If $G_{\mathcal{T}_i}$ \emph{follows} the vector $\tau^i_j=(\alpha_1,\dots,\alpha_{p+1},\dots,\alpha_{p+k})$, then $\alpha_1,\dots,\alpha_{p+k}$ is a partition of $V(G_{\mathcal{T}_i})$, and $\Par_j^i=\{C'_1\cup \alpha_1,\dots,C'_p\cup \alpha_p, \alpha_{p+1},\dots,\alpha_{p+k}\}$ is a possible extension of $\Par'$ including the vertices that belong in any component of type $i$, according to the vector $\tau_j^i$. 

For every $i\in [\ell]$ and $j\in[q]$, let $E^i_j=\{E(\Par^i_j)\setminus E(\Par')\}$ (\textit{i.e.}, the edges of the subgraph of $G$ induced by $\Par_j^i$) be the \emph{important edges according to} $\tau^i_j$. All that remains to be done is to search through these vectors and find the optimal ones among those that result in $\C$-partitions. This is achieved through the following ILP. 

\bigskip

\par\noindent\rule{0.88\textwidth}{0.5pt}

\smallskip 

\noindent Variables

\begin{tabular}{lll}

\vspace{5pt}

$x_i$ & $ i \in [\ell]$ & 
\begin{tabular}{@{}c@{}}
number of components of type $i$
\end{tabular} \\

\vspace{5pt}

$y_{i,j}$ & $i\in [\ell], j\in [q]$ &
\begin{tabular}{@{}c@{}}
number of important edges according to $\tau^i_j$
\end{tabular}\\

\vspace{5pt}

$v_{i,j,l}$ & $i\in [\ell], j\in [q], l\in [p]$ &
\begin{tabular}{@{}c@{}}
number of vertices in the $l^{th}$ position of vector $\tau^j_i$
\end{tabular}\\

\vspace{5pt}

$z_{i,j}$ & $i\in [\ell], j\in [q]$ &
\begin{tabular}{@{}c@{}}
number of components of type $i$ following the vector $\tau_j^i$
\end{tabular}
\end{tabular}

\noindent Constants

\begin{tabular}{lll}
$w_{l}$ & $l\in[p]$ & number of vertices in $\C'_l$\\
\end{tabular}

\noindent Objective
\begin{align}
\max \sum_{i=1}^{\ell}\sum_{j=1}^{q} y_{i,j}z_{i,j}
\end{align}

\noindent Constraints
\begin{align}
&\sum_{j=1}^{q}z_{i,j}=x_i && \forall i\in[\ell]\label{subtypes-create-types}\\
&\sum_{i=1}^\ell \sum_{j=1}^{q}v_{i,j,z}z_{i,j}+w_l\leq \C && \forall l\in[p]\label{capacity-is-respected}
\end{align}

\par\noindent\rule{0.88\textwidth}{0.5pt}

\bigskip 

In the above model, the constraint~\ref{subtypes-create-types} is used to make sure that every component of type $i$ follows exactly one vector $\tau^i_j$. Then, the constraint~\ref{capacity-is-respected} is used to make sure that the resulting partition is indeed a $\C$-partition. 
Finally, since the number of variables of the model is bounded by a function of $k$, we can and obtain a solution in FPT time, parameterized by $k$ (by running for example the Lenstra algorithm~\cite{Len83}).
\end{proof}
\fi

\subsection{Intractability for twin-cover number}
We close this section by proving that \CCF is $\W[1]$-hard when parameterized by the twin-cover of the input graph. This result is achieved through a reduction from the following problem, which was shown to be $\W[1]$-hard when parameterized by $k$ in~\cite{JKMS13}.

\Pb{\textsc{Unary Bin Packing}}{a set of items $A=\{a_1,\dots,a_n\}$, a \emph{size function} $s: A\rightarrow \mathbb{N}$ which returns the size of each item in unary encoding, and two integers $B$ and $k$.}{Question}
{Is there a way to fit items of $A$ into $k$ bins, so that every bin contains items of total size exactly $B$, and every item is assigned to exactly one bin?}

\thmTwinCover*

\iflong
\begin{proof}
Let $(A,s,B,k)$, where $A=\{a_1,\dots,a_n\}$, be an instance of \textsc{Unary Bin Packing}. We construct an instance of \CCF~as follows: for each $j\in[n]$, construct the clique $K^j$, which is of order $s(a_j)$. Then, for each $i\in[k]$, add one vertex $b_i$ and all the edges between $b_i$ and all the vertices of the cliques $K^j$, for all $j\in[n]$. Let $G$ be the resulting graph, and set $\C=B+1$. Observe that the twin-cover number of $G$ is at most $k$, as the set $\{b_1,\dots,b_k\}$ is a twin-cover of $G$. We will show that any optimal partition $\Par$ of $(G,\C)$ has value $v(\Par)=\sum_{j=1}^n\frac{s(a_j)(s(a_j)-1)}{2}+kB$ if and only if $(A,s,B,k)$ is a yes-instance of \textsc{Unary Bin Packing}.

For the first direction of the reduction, let $(A,s,B,k)$ be a yes-instance of \textsc{Unary Bin Packing} and let $f:A\rightarrow [k]$ be the returned function assigning items to bins, such that every bin contains items with total size exactly equal to $B$. We define a partition $\Par$ of $V(G)$ into $k$ sets $C_1,\dots,C_k$ as follows. For every $i\in[k]$, the set $C_i$ contains $b_i$ and all the vertices of the clique $K^j$ such that $f(a_j)=i$, for all $j\in[n]$. Clearly, $|C_i|= B+1=\C$ for every $i\in[k]$ and, thus, $\Par$ is a $\C$-partition of $G$. Moreover, $E(\Par)$ contains all the edges that belong in the clique $K^j$, for every $j\in[n]$, and exactly $B$ edges incident to $b_i$, for each $i\in[k]$. In total, $v(\Par)=|E(\Par)|= \sum_{j=1}^n\frac{s(a_j)(s(a_j)-1)}{2}+kB$.

For the reverse direction, let $G$ be an instance of \CCF~and $\Par=C_1,\dots,C_p$ be a $\C$-partition of $G$ with value $v(\Par)=\sum_{j=1}^n\frac{s(a_j)(s(a_j)-1)}{2}+kB$. Let $G'=G-\{b_1,\dots,b_k\}$ and observe that $|E(G')|=\sum_{j=1}^n\frac{s(a_j)(s(a_j)-1)}{2}$. We have the following two claims.

\begin{claim}\label{cl:each-bag-unique}
For each $i\in[k]$, there exists a unique $\ell\in [p]$ such that $b_i\in C_\ell$. Moreover, $p=k$ and $|C_\ell|=B+1$.
\end{claim}
\begin{proofclaim}
In order for $v(\Par)$ to have the correct value, and by the construction of $G$, each one of the vertices $b_1,\dots,b_k$ contributes exactly $B$ edges to $v(\Par)$. Indeed, since $|C_i|\leq B+1$, $i\in [p]$, no vertex can contribute more than $B$ edges towards $v(\Par)$.
Assume now that there exist $i<i'\in[k]$ and $\ell\in[p]$ such that $b_i$ and $b_{i'}$ both belong to $C_\ell$. Then, since $\C=B+1$ and by the construction of $G$, we have that $C_{\ell}$ contains at most $B-1$ edges incident to $b_i$ and $b_{i'}$, which is a contradiction. Finally, for all $i\in [p]$, if $C_i$ contains a vertex from $\{b_1,\dots,b_k\}$, then $|C_i|=B+1$. It also follows that $p=k$.
\end{proofclaim}

\begin{claim}\label{cl:each-item-unique} 
For each $j\in[n]$, all the vertices of $K^j$ belong in the same set of $\Par$.
\end{claim}
\begin{proofclaim}
In order for $v(\Par)$ to have the correct value, and by the construction of $G$, we have that for each $j\in[n]$, each vertex of $K^j$ contributes all of its incident edges in $G'$ towards $v(\Par)$.   
\end{proofclaim}

We are now ready to show that $(A,s,B,k)$ is a yes-instance of the \textsc{Unary Bin Packing} problem. Let $\Par$ be an optimal $\C$-partition of $G$. It follows from Claim~\ref{cl:each-bag-unique} that $\Par$ consists of $k$ sets $\{C_1,\dots,C_k\}$. We create the bins $B_1,\dots,B_k$. For each $j\in [n]$, we insert the item $a_j$ in the bin $B_i$, $i\in[k]$, if and only if $K^j\subseteq C_i$. It follows from Claim~\ref{cl:each-item-unique} that each item of $A$ has been assigned to exactly one bin. Recall that for each $j\in[n]$, the item $a_j$ has size equal to the order of $K^j$ (by construction). Moreover, for each $j\in[n]$, the item $a_j$ corresponds exactly to the clique $K^j$. Thus, from Claim~\ref{cl:each-item-unique}, we have that for each $i\in k$, $|C_i|$ is equal to the sum of the orders of the cliques contained in $C_i$, which is exactly equal to $B$. 

It remains to show that $\sum_{a_j \in B_\ell} s(a_j) = B$ for all $\ell \in [k]$.
Recall that $|V(K^j)| = s(a_j)$, for $j\in[n]$. Let $\ell\in[k]$. We have that  $\sum_{a_j \in B_\ell} s(a_j) = \sum_{a_j \in B_\ell} |V(K^j)| $. Also, $|C_\ell| = \sum_{a_j \in B_\ell} |V(K^j)| +1$ since $C_\ell$ contains the cliques that correspond to the items contained in $B_\ell$ and one vertex from $\{b_1,\dots,b_k\}$. Thus, $\sum_{a_j \in B_\ell} s(a_j)=|C_\ell|-1=B$.  
\end{proof}
\fi

\section{Conclusion}
In this paper, we studied the \textsc{$\C$-Coalition Formation} problem, considering both its weighted and unweighted versions, through the lens of parameterized complexity. The main takeaway message is that the problems behave relatively well in regards to many widely used parameters, despite the multiple intractability results that we provided. Moreover, our intractability results provide motivation towards a more heuristic-oriented approach. 
From the theoretical point of view, there are many rather interesting questions that are born from our research. 
First notice that, when we consider instance of \CCFw where $\C<\tw$ we do not have any guarantee that the running time of the algorithm presented in Theorem~\ref{thm:TwC-FPT} cannot been improved. Indeed, the lower bound we presented in Theorem~\ref{thm:vc-to-the-vc} holds for instances where $\C\ge \vc \ge \tw$. Therefore, it would be interesting to investigate if, by applying more advanced techniques (like the Cut $\&$ Count technique~\cite{CyganFKLMPPS15}), we can solve \CCFw in time $\C^{\mathcal{O}(\tw)}n^{\mathcal{O}(1)}$ (which is faster than the existing algorithm when $\C<\tw$) or prove that the lower bound of Theorem~\ref{thm:vc-to-the-vc} holds even when $\tw>\C$.
Finally, we wonder about the existence of an FPT algorithm for \CCFw~parameterized by other interesting parameters that are left untouched by our work, such as the vertex integrity, the neighborhood diversity and the feedback vertex number of the input graph. 

\section{Acknowledgments}


  






  FF and NM acknowledge the support by the European Union under the project Robotics and advanced industrial production (reg.\ no.\ CZ.02.01.01/00/22\_008/0004590) and the CTU Global postdoc fellowship program. NM is also partially supported by Charles Univ. projects UNCE 24/SCI/008 and PRIMUS 24/SCI/012, and by the project 25-17221S of GAČR.  HG acknowledges the support by the IDEX-ISITE initiative CAP 20-25 (ANR-16-IDEX-0001), the International Research Center ``Innovation Transportation and Production Systems'' of the I-SITE CAP 20-25, and the ANR project GRALMECO (ANR-21-CE48-0004).

\bibliographystyle{ACM-Reference-Format}
\bibliography{main}


\begin{thebibliography}{66}


\ifx \showCODEN    \undefined \def \showCODEN     #1{\unskip}     \fi
\ifx \showDOI      \undefined \def \showDOI       #1{#1}\fi
\ifx \showISBNx    \undefined \def \showISBNx     #1{\unskip}     \fi
\ifx \showISBNxiii \undefined \def \showISBNxiii  #1{\unskip}     \fi
\ifx \showISSN     \undefined \def \showISSN      #1{\unskip}     \fi
\ifx \showLCCN     \undefined \def \showLCCN      #1{\unskip}     \fi
\ifx \shownote     \undefined \def \shownote      #1{#1}          \fi
\ifx \showarticletitle \undefined \def \showarticletitle #1{#1}   \fi
\ifx \showURL      \undefined \def \showURL       {\relax}        \fi
\providecommand\bibfield[2]{#2}
\providecommand\bibinfo[2]{#2}
\providecommand\natexlab[1]{#1}
\providecommand\showeprint[2][]{arXiv:#2}

\bibitem[Aloisio et~al\mbox{.}(2020)]%
        {aloisio2020impact}
\bibfield{author}{\bibinfo{person}{Alessandro Aloisio}, \bibinfo{person}{Michele Flammini}, {and} \bibinfo{person}{Cosimo Vinci}.} \bibinfo{year}{2020}\natexlab{}.
\newblock \showarticletitle{The impact of selfishness in hypergraph hedonic games}. In \bibinfo{booktitle}{\emph{Proceedings of the AAAI Conference on Artificial Intelligence}}, Vol.~\bibinfo{volume}{34}. \bibinfo{pages}{1766--1773}.
\newblock


\bibitem[Aziz et~al\mbox{.}(2019)]%
        {aziz2019fractional}
\bibfield{author}{\bibinfo{person}{Haris Aziz}, \bibinfo{person}{Florian Brandl}, \bibinfo{person}{Felix Brandt}, \bibinfo{person}{Paul Harrenstein}, \bibinfo{person}{Martin Olsen}, {and} \bibinfo{person}{Dominik Peters}.} \bibinfo{year}{2019}\natexlab{}.
\newblock \showarticletitle{Fractional hedonic games}.
\newblock \bibinfo{journal}{\emph{ACM Transactions on Economics and Computation (TEAC)}} \bibinfo{volume}{7}, \bibinfo{number}{2} (\bibinfo{year}{2019}), \bibinfo{pages}{1--29}.
\newblock


\bibitem[Aziz et~al\mbox{.}(2013)]%
        {aziz2013computing}
\bibfield{author}{\bibinfo{person}{Haris Aziz}, \bibinfo{person}{Felix Brandt}, {and} \bibinfo{person}{Hans~Georg Seedig}.} \bibinfo{year}{2013}\natexlab{}.
\newblock \showarticletitle{Computing desirable partitions in additively separable hedonic games}.
\newblock \bibinfo{journal}{\emph{Artificial Intelligence}}  \bibinfo{volume}{195} (\bibinfo{year}{2013}), \bibinfo{pages}{316--334}.
\newblock


\bibitem[Bachrach et~al\mbox{.}(2013)]%
        {approxGraphical}
\bibfield{author}{\bibinfo{person}{Yoram Bachrach}, \bibinfo{person}{Pushmeet Kohli}, \bibinfo{person}{Vladimir Kolmogorov}, {and} \bibinfo{person}{Morteza Zadimoghaddam}.} \bibinfo{year}{2013}\natexlab{}.
\newblock \showarticletitle{Optimal coalition structure generation in cooperative graph games}. In \bibinfo{booktitle}{\emph{Proceedings of the AAAI Conference on Artificial Intelligence}}, Vol.~\bibinfo{volume}{27}. \bibinfo{pages}{81--87}.
\newblock


\bibitem[B{\"a}ckstr{\"o}m et~al\mbox{.}(2012)]%
        {backstrom2012complexity}
\bibfield{author}{\bibinfo{person}{Christer B{\"a}ckstr{\"o}m}, \bibinfo{person}{Yue Chen}, \bibinfo{person}{Peter Jonsson}, \bibinfo{person}{Sebastian Ordyniak}, {and} \bibinfo{person}{Stefan Szeider}.} \bibinfo{year}{2012}\natexlab{}.
\newblock \showarticletitle{The complexity of planning revisited—a parameterized analysis}. In \bibinfo{booktitle}{\emph{Proceedings of the AAAI Conference on Artificial Intelligence}}, Vol.~\bibinfo{volume}{26}. \bibinfo{pages}{1735--1741}.
\newblock


\bibitem[Barrot et~al\mbox{.}(2019)]%
        {barrot2019unknown}
\bibfield{author}{\bibinfo{person}{Nathana{\"e}l Barrot}, \bibinfo{person}{Kazunori Ota}, \bibinfo{person}{Yuko Sakurai}, {and} \bibinfo{person}{Makoto Yokoo}.} \bibinfo{year}{2019}\natexlab{}.
\newblock \showarticletitle{Unknown agents in friends oriented hedonic games: Stability and complexity}. In \bibinfo{booktitle}{\emph{Proceedings of the AAAI Conference on Artificial Intelligence}}, Vol.~\bibinfo{volume}{33}. \bibinfo{pages}{1756--1763}.
\newblock


\bibitem[Barrot and Yokoo(2019)]%
        {barrot2019stable}
\bibfield{author}{\bibinfo{person}{Nathana{\"e}l Barrot} {and} \bibinfo{person}{Makoto Yokoo}.} \bibinfo{year}{2019}\natexlab{}.
\newblock \showarticletitle{Stable and Envy-free Partitions in Hedonic Games.}. In \bibinfo{booktitle}{\emph{IJCAI}}. \bibinfo{pages}{67--73}.
\newblock


\bibitem[Bessiere et~al\mbox{.}(2008)]%
        {bessiere2008parameterized}
\bibfield{author}{\bibinfo{person}{Christian Bessiere}, \bibinfo{person}{Emmanuel Hebrard}, \bibinfo{person}{Brahim Hnich}, \bibinfo{person}{Zeynep Kiziltan}, \bibinfo{person}{Claude~Guy Quimper}, {and} \bibinfo{person}{Toby Walsh}.} \bibinfo{year}{2008}\natexlab{}.
\newblock \showarticletitle{The parameterized complexity of global constraints}. In \bibinfo{booktitle}{\emph{AAAI Conference on Artificial Intelligence}}. AAAI Press, \bibinfo{pages}{235--240}.
\newblock


\bibitem[Bil{\`o} et~al\mbox{.}(2022)]%
        {bilo2022hedonic}
\bibfield{author}{\bibinfo{person}{Vittorio Bil{\`o}}, \bibinfo{person}{Gianpiero Monaco}, {and} \bibinfo{person}{Luca Moscardelli}.} \bibinfo{year}{2022}\natexlab{}.
\newblock \showarticletitle{Hedonic games with fixed-size coalitions}. In \bibinfo{booktitle}{\emph{Proceedings of the AAAI Conference on Artificial Intelligence}}.
\newblock


\bibitem[Bodlaender(1996)]%
        {Bodlaender96}
\bibfield{author}{\bibinfo{person}{Hans~L. Bodlaender}.} \bibinfo{year}{1996}\natexlab{}.
\newblock \showarticletitle{A Linear-Time Algorithm for Finding Tree-Decompositions of Small Treewidth}.
\newblock \bibinfo{journal}{\emph{{SIAM} J. Comput.}} \bibinfo{volume}{25}, \bibinfo{number}{6} (\bibinfo{year}{1996}), \bibinfo{pages}{1305--1317}.
\newblock
\urldef\tempurl%
\url{https://doi.org/10.1137/S0097539793251219}
\showDOI{\tempurl}


\bibitem[Bodlaender(1998)]%
        {B98}
\bibfield{author}{\bibinfo{person}{Hans~L. Bodlaender}.} \bibinfo{year}{1998}\natexlab{}.
\newblock \showarticletitle{A partial $k$-arboretum of graphs with bounded treewidth}.
\newblock \bibinfo{journal}{\emph{Theoretical Computer Science}} \bibinfo{volume}{209}, \bibinfo{number}{1} (\bibinfo{year}{1998}), \bibinfo{pages}{1--45}.
\newblock
\urldef\tempurl%
\url{https://doi.org/10.1016/S0304-3975(97)00228-4}
\showDOI{\tempurl}


\bibitem[Boehmer and Elkind(2020)]%
        {boehmer2020individual}
\bibfield{author}{\bibinfo{person}{Niclas Boehmer} {and} \bibinfo{person}{Edith Elkind}.} \bibinfo{year}{2020}\natexlab{}.
\newblock \showarticletitle{Individual-based stability in hedonic diversity games}. In \bibinfo{booktitle}{\emph{Proceedings of the AAAI Conference on Artificial Intelligence}}, Vol.~\bibinfo{volume}{34}. \bibinfo{pages}{1822--1829}.
\newblock


\bibitem[Bogomolnaia and Jackson(2002)]%
        {bogomolnaia2002stability}
\bibfield{author}{\bibinfo{person}{Anna Bogomolnaia} {and} \bibinfo{person}{Matthew~O Jackson}.} \bibinfo{year}{2002}\natexlab{}.
\newblock \showarticletitle{The stability of hedonic coalition structures}.
\newblock \bibinfo{journal}{\emph{Games and Economic Behavior}} \bibinfo{volume}{38}, \bibinfo{number}{2} (\bibinfo{year}{2002}), \bibinfo{pages}{201--230}.
\newblock


\bibitem[Brandt et~al\mbox{.}(2023)]%
        {brandt2023reaching}
\bibfield{author}{\bibinfo{person}{Felix Brandt}, \bibinfo{person}{Martin Bullinger}, {and} \bibinfo{person}{Ana{\"e}lle Wilczynski}.} \bibinfo{year}{2023}\natexlab{}.
\newblock \showarticletitle{Reaching individually stable coalition structures}.
\newblock \bibinfo{journal}{\emph{ACM Transactions on Economics and Computation}} \bibinfo{volume}{11}, \bibinfo{number}{1-2} (\bibinfo{year}{2023}), \bibinfo{pages}{1--65}.
\newblock


\bibitem[Brandt et~al\mbox{.}(2016)]%
        {handbook}
\bibfield{author}{\bibinfo{person}{Felix Brandt}, \bibinfo{person}{Vincent Conitzer}, \bibinfo{person}{Ulle Endriss}, \bibinfo{person}{J{\'e}r{\^o}me Lang}, {and} \bibinfo{person}{Ariel~D Procaccia}.} \bibinfo{year}{2016}\natexlab{}.
\newblock \bibinfo{booktitle}{\emph{Handbook of computational social choice}}.
\newblock \bibinfo{publisher}{Cambridge University Press}.
\newblock


\bibitem[Bredereck et~al\mbox{.}(2017)]%
        {bredereck2017parliamentary}
\bibfield{author}{\bibinfo{person}{Robert Bredereck}, \bibinfo{person}{Jiehua Chen}, \bibinfo{person}{Rolf Niedermeier}, {and} \bibinfo{person}{Toby Walsh}.} \bibinfo{year}{2017}\natexlab{}.
\newblock \showarticletitle{Parliamentary voting procedures: Agenda control, manipulation, and uncertainty}.
\newblock \bibinfo{journal}{\emph{Journal of Artificial Intelligence Research}}  \bibinfo{volume}{59} (\bibinfo{year}{2017}), \bibinfo{pages}{133--173}.
\newblock


\bibitem[Bullinger and Kober(2021)]%
        {bullinger2021loyalty}
\bibfield{author}{\bibinfo{person}{Martin Bullinger} {and} \bibinfo{person}{Stefan Kober}.} \bibinfo{year}{2021}\natexlab{}.
\newblock \showarticletitle{Loyalty in Cardinal Hedonic Games.}. In \bibinfo{booktitle}{\emph{IJCAI}}. \bibinfo{pages}{66--72}.
\newblock


\bibitem[Chen et~al\mbox{.}(2023)]%
        {CCRS23}
\bibfield{author}{\bibinfo{person}{Jiehua Chen}, \bibinfo{person}{Gergely Cs{\'{a}}ji}, \bibinfo{person}{Sanjukta Roy}, {and} \bibinfo{person}{Sofia Simola}.} \bibinfo{year}{2023}\natexlab{}.
\newblock \showarticletitle{Hedonic Games With Friends, Enemies, and Neutrals: Resolving Open Questions and Fine-Grained Complexity}. In \bibinfo{booktitle}{\emph{Proceedings of the 2023 International Conference on Autonomous Agents and Multiagent Systems, {AAMAS} 2023, London, United Kingdom, 29 May 2023 - 2 June 2023}}.
\newblock


\bibitem[Cseh et~al\mbox{.}(2019)]%
        {cseh2019pareto}
\bibfield{author}{\bibinfo{person}{Agnes Cseh}, \bibinfo{person}{Tam{\'a}s Fleiner}, {and} \bibinfo{person}{Petra Harj{\'a}n}.} \bibinfo{year}{2019}\natexlab{}.
\newblock \showarticletitle{Pareto Optimal Coalitions of Fixed Size}.
\newblock \bibinfo{journal}{\emph{Journal of Mechanism and Institution Design}}  \bibinfo{volume}{4} (\bibinfo{year}{2019}), \bibinfo{pages}{1}.
\newblock


\bibitem[Cygan et~al\mbox{.}(2015)]%
        {CyganFKLMPPS15}
\bibfield{author}{\bibinfo{person}{Marek Cygan}, \bibinfo{person}{Fedor~V. Fomin}, \bibinfo{person}{\L{}ukasz Kowalik}, \bibinfo{person}{Daniel Lokshtanov}, \bibinfo{person}{D{\'{a}}niel Marx}, \bibinfo{person}{Marcin Pilipczuk}, \bibinfo{person}{Micha\l{} Pilipczuk}, {and} \bibinfo{person}{Saket Saurabh}.} \bibinfo{year}{2015}\natexlab{}.
\newblock \bibinfo{booktitle}{\emph{Parameterized Algorithms}}.
\newblock \bibinfo{publisher}{Springer}.
\newblock
\showISBNx{978-3-319-21274-6}
\urldef\tempurl%
\url{https://doi.org/10.1007/978-3-319-21275-3}
\showDOI{\tempurl}


\bibitem[Darmann et~al\mbox{.}(2018)]%
        {schedule}
\bibfield{author}{\bibinfo{person}{Andreas Darmann}, \bibinfo{person}{Edith Elkind}, \bibinfo{person}{Sascha Kurz}, \bibinfo{person}{J{\'e}r{\^o}me Lang}, \bibinfo{person}{Joachim Schauer}, {and} \bibinfo{person}{Gerhard Woeginger}.} \bibinfo{year}{2018}\natexlab{}.
\newblock \showarticletitle{Group activity selection problem with approval preferences}.
\newblock \bibinfo{journal}{\emph{International Journal of Game Theory}}  \bibinfo{volume}{47} (\bibinfo{year}{2018}), \bibinfo{pages}{767--796}.
\newblock


\bibitem[Deng and Papadimitriou(1994)]%
        {deng1994complexity}
\bibfield{author}{\bibinfo{person}{Xiaotie Deng} {and} \bibinfo{person}{Christos~H Papadimitriou}.} \bibinfo{year}{1994}\natexlab{}.
\newblock \showarticletitle{On the complexity of cooperative solution concepts}.
\newblock \bibinfo{journal}{\emph{Mathematics of operations research}} \bibinfo{volume}{19}, \bibinfo{number}{2} (\bibinfo{year}{1994}), \bibinfo{pages}{257--266}.
\newblock


\bibitem[Diestel(2012)]%
        {D12}
\bibfield{author}{\bibinfo{person}{Reinhard Diestel}.} \bibinfo{year}{2012}\natexlab{}.
\newblock \bibinfo{booktitle}{\emph{Graph Theory, 4th Edition}}. \bibinfo{series}{Graduate texts in mathematics}, Vol.~\bibinfo{volume}{173}.
\newblock \bibinfo{publisher}{Springer}.
\newblock
\urldef\tempurl%
\url{https://doi.org/10.1007/978-3-662-53622-3}
\showDOI{\tempurl}


\bibitem[Downey and Fellows(2013a)]%
        {DowneyF13}
\bibfield{author}{\bibinfo{person}{Rodney~G. Downey} {and} \bibinfo{person}{Michael~R. Fellows}.} \bibinfo{year}{2013}\natexlab{a}.
\newblock \bibinfo{booktitle}{\emph{Fundamentals of Parameterized Complexity}}.
\newblock \bibinfo{publisher}{Springer}.
\newblock
\showISBNx{978-1-4471-5558-4}
\urldef\tempurl%
\url{https://doi.org/10.1007/978-1-4471-5559-1}
\showDOI{\tempurl}


\bibitem[Downey and Fellows(2013b)]%
        {DF13}
\bibfield{author}{\bibinfo{person}{Rodney~G. Downey} {and} \bibinfo{person}{Michael~R. Fellows}.} \bibinfo{year}{2013}\natexlab{b}.
\newblock \bibinfo{booktitle}{\emph{Fundamentals of Parameterized Complexity}}.
\newblock \bibinfo{publisher}{Springer}.
\newblock
\urldef\tempurl%
\url{https://doi.org/10.1007/978-1-4471-5559-1}
\showDOI{\tempurl}


\bibitem[Drange et~al\mbox{.}(2016)]%
        {DDH16}
\bibfield{author}{\bibinfo{person}{P{\aa}l~Gr{\o}n{\aa}s Drange}, \bibinfo{person}{Markus~S. Dregi}, {and} \bibinfo{person}{Pim van~'t Hof}.} \bibinfo{year}{2016}\natexlab{}.
\newblock \showarticletitle{On the Computational Complexity of Vertex Integrity and Component Order Connectivity}.
\newblock \bibinfo{journal}{\emph{Algorithmica}} \bibinfo{volume}{76}, \bibinfo{number}{4} (\bibinfo{year}{2016}), \bibinfo{pages}{1181--1202}.
\newblock


\bibitem[Dreze and Greenberg(1980)]%
        {dreze1980hedonic}
\bibfield{author}{\bibinfo{person}{Jacques~H Dreze} {and} \bibinfo{person}{Joseph Greenberg}.} \bibinfo{year}{1980}\natexlab{}.
\newblock \showarticletitle{Hedonic coalitions: Optimality and stability}.
\newblock \bibinfo{journal}{\emph{Econometrica: Journal of the Econometric Society}} (\bibinfo{year}{1980}), \bibinfo{pages}{987--1003}.
\newblock


\bibitem[Edmonds(1965)]%
        {E65}
\bibfield{author}{\bibinfo{person}{Jack Edmonds}.} \bibinfo{year}{1965}\natexlab{}.
\newblock \showarticletitle{Paths, Trees, and Flowers}.
\newblock \bibinfo{journal}{\emph{Canadian Journal of Mathematics}}  \bibinfo{volume}{17} (\bibinfo{year}{1965}), \bibinfo{pages}{449–467}.
\newblock
\urldef\tempurl%
\url{https://doi.org/10.4153/CJM-1965-045-4}
\showDOI{\tempurl}


\bibitem[Fanelli et~al\mbox{.}(2021)]%
        {fanelli2021relaxed}
\bibfield{author}{\bibinfo{person}{Angelo Fanelli}, \bibinfo{person}{Gianpiero Monaco}, \bibinfo{person}{Luca Moscardelli}, {et~al\mbox{.}}} \bibinfo{year}{2021}\natexlab{}.
\newblock \showarticletitle{Relaxed core stability in fractional hedonic games}. In \bibinfo{booktitle}{\emph{Proceedings of the Thirtieth International Joint Conference on Artificial Intelligence}}. \bibinfo{pages}{182--188}.
\newblock


\bibitem[Fellows(2006)]%
        {kernelApplication}
\bibfield{author}{\bibinfo{person}{M.~R. Fellows}.} \bibinfo{year}{2006}\natexlab{}.
\newblock \showarticletitle{The Lost Continent of Polynomial Time: Preprocessing and Kernelization} \emph{(\bibinfo{series}{IWPEC'06})}. \bibinfo{publisher}{Springer-Verlag}, \bibinfo{address}{Berlin, Heidelberg}, \bibinfo{pages}{276–277}.
\newblock
\showISBNx{3540390987}
\urldef\tempurl%
\url{https://doi.org/10.1007/11847250_25}
\showDOI{\tempurl}


\bibitem[Fioravantes et~al\mbox{.}(2025)]%
        {ourAAAISmallCars}
\bibfield{author}{\bibinfo{person}{Foivos Fioravantes}, \bibinfo{person}{Harmender Gahlawat}, {and} \bibinfo{person}{Nikolaos Melissinos}.} \bibinfo{year}{2025}\natexlab{}.
\newblock \showarticletitle{Exact Algorithms and Lower Bounds for Forming Coalitions of Constrained Maximum Size}. In \bibinfo{booktitle}{\emph{Proceedings of the AAAI Conference on Artificial Intelligence}}, Vol.~\bibinfo{volume}{39}. \bibinfo{pages}{13847--13855}.
\newblock


\bibitem[Flammini et~al\mbox{.}(2018)]%
        {onlineAAMAS}
\bibfield{author}{\bibinfo{person}{Michele Flammini}, \bibinfo{person}{Gianpiero Monaco}, \bibinfo{person}{Luca Moscardelli}, \bibinfo{person}{Mordechai Shalom}, {and} \bibinfo{person}{Shmuel Zaks}.} \bibinfo{year}{2018}\natexlab{}.
\newblock \showarticletitle{Online coalition structure generation in graph games}. In \bibinfo{booktitle}{\emph{Proceedings of the 17th International Conference on Autonomous Agents and MultiAgent Systems}}. \bibinfo{pages}{1353--1361}.
\newblock


\bibitem[Flum and Grohe(2006)]%
        {FlumG06}
\bibfield{author}{\bibinfo{person}{J{\"{o}}rg Flum} {and} \bibinfo{person}{Martin Grohe}.} \bibinfo{year}{2006}\natexlab{}.
\newblock \bibinfo{booktitle}{\emph{Parameterized Complexity Theory}}.
\newblock \bibinfo{publisher}{Springer}.
\newblock
\showISBNx{978-3-540-29952-3}
\urldef\tempurl%
\url{https://doi.org/10.1007/3-540-29953-X}
\showDOI{\tempurl}


\bibitem[Fomin et~al\mbox{.}(2019)]%
        {FLSZ19}
\bibfield{author}{\bibinfo{person}{Fedor~V. Fomin}, \bibinfo{person}{Daniel Lokshtanov}, \bibinfo{person}{Saket Saurabh}, {and} \bibinfo{person}{Meirav Zehavi}.} \bibinfo{year}{2019}\natexlab{}.
\newblock \bibinfo{booktitle}{\emph{Kernelization: Theory of Parameterized Preprocessing}}.
\newblock \bibinfo{publisher}{Cambridge University Press}.
\newblock
\urldef\tempurl%
\url{https://doi.org/10.1017/9781107415157}
\showDOI{\tempurl}


\bibitem[Ganian(2011)]%
        {G11}
\bibfield{author}{\bibinfo{person}{Robert Ganian}.} \bibinfo{year}{2011}\natexlab{}.
\newblock \showarticletitle{Twin-Cover: Beyond Vertex Cover in Parameterized Algorithmics}. In \bibinfo{booktitle}{\emph{Parameterized and Exact Computation - 6th International Symposium, {IPEC} 2011}} \emph{(\bibinfo{series}{Lecture Notes in Computer Science}, Vol.~\bibinfo{volume}{7112})}. \bibinfo{publisher}{Springer}, \bibinfo{pages}{259--271}.
\newblock
\urldef\tempurl%
\url{https://doi.org/10.1007/978-3-642-28050-4\_21}
\showDOI{\tempurl}


\bibitem[Ganian et~al\mbox{.}(2023)]%
        {ganian2023hedonic}
\bibfield{author}{\bibinfo{person}{Robert Ganian}, \bibinfo{person}{Thekla Hamm}, \bibinfo{person}{Du{\v{s}}an Knop}, \bibinfo{person}{{\v{S}}imon Schierreich}, {and} \bibinfo{person}{Ond{\v{r}}ej Such{\`y}}.} \bibinfo{year}{2023}\natexlab{}.
\newblock \showarticletitle{Hedonic diversity games: A complexity picture with more than two colors}.
\newblock \bibinfo{journal}{\emph{Artificial Intelligence}}  \bibinfo{volume}{325} (\bibinfo{year}{2023}), \bibinfo{pages}{104017}.
\newblock


\bibitem[Gao(2009)]%
        {gao2009data}
\bibfield{author}{\bibinfo{person}{Yong Gao}.} \bibinfo{year}{2009}\natexlab{}.
\newblock \showarticletitle{Data reductions, fixed parameter tractability, and random weighted d-CNF satisfiability}.
\newblock \bibinfo{journal}{\emph{Artificial Intelligence}} \bibinfo{volume}{173}, \bibinfo{number}{14} (\bibinfo{year}{2009}), \bibinfo{pages}{1343--1366}.
\newblock


\bibitem[Guo and Niedermeier(2007)]%
        {guo2007invitation}
\bibfield{author}{\bibinfo{person}{Jiong Guo} {and} \bibinfo{person}{Rolf Niedermeier}.} \bibinfo{year}{2007}\natexlab{}.
\newblock \showarticletitle{Invitation to data reduction and problem kernelization}.
\newblock \bibinfo{journal}{\emph{ACM SIGACT News}} \bibinfo{volume}{38}, \bibinfo{number}{1} (\bibinfo{year}{2007}), \bibinfo{pages}{31--45}.
\newblock


\bibitem[Gutin et~al\mbox{.}(2012)]%
        {GKSSY12}
\bibfield{author}{\bibinfo{person}{Gregory~Z. Gutin}, \bibinfo{person}{Eun~Jung Kim}, \bibinfo{person}{Arezou Soleimanfallah}, \bibinfo{person}{Stefan Szeider}, {and} \bibinfo{person}{Anders Yeo}.} \bibinfo{year}{2012}\natexlab{}.
\newblock \showarticletitle{Parameterized Complexity Results for General Factors in Bipartite Graphs with an Application to Constraint Programming}.
\newblock \bibinfo{journal}{\emph{Algorithmica}} \bibinfo{volume}{64}, \bibinfo{number}{1} (\bibinfo{year}{2012}), \bibinfo{pages}{112--125}.
\newblock


\bibitem[Hanaka et~al\mbox{.}(2023)]%
        {HIO23}
\bibfield{author}{\bibinfo{person}{Tesshu Hanaka}, \bibinfo{person}{Airi Ikeyama}, {and} \bibinfo{person}{Hirotaka Ono}.} \bibinfo{year}{2023}\natexlab{}.
\newblock \showarticletitle{Maximizing Utilitarian and Egalitarian Welfare of Fractional Hedonic Games on Tree-Like Graphs}. In \bibinfo{booktitle}{\emph{Combinatorial Optimization and Applications - 17th International Conference, {COCOA} 2023, Hawaii, HI, USA, December 15-17, 2023, Proceedings, Part {I}}}.
\newblock


\bibitem[Hanaka et~al\mbox{.}(2019)]%
        {HKMO19}
\bibfield{author}{\bibinfo{person}{Tesshu Hanaka}, \bibinfo{person}{Hironori Kiya}, \bibinfo{person}{Yasuhide Maei}, {and} \bibinfo{person}{Hirotaka Ono}.} \bibinfo{year}{2019}\natexlab{}.
\newblock \showarticletitle{Computational Complexity of Hedonic Games on Sparse Graphs}. In \bibinfo{booktitle}{\emph{{PRIMA} 2019: Principles and Practice of Multi-Agent Systems - 22nd International Conference, Turin, Italy, October 28-31, 2019, Proceedings}}.
\newblock


\bibitem[Hanaka and Lampis(2022)]%
        {lampis2022hedonic}
\bibfield{author}{\bibinfo{person}{Tesshu Hanaka} {and} \bibinfo{person}{Michael Lampis}.} \bibinfo{year}{2022}\natexlab{}.
\newblock \showarticletitle{Hedonic Games and Treewidth Revisited}. In \bibinfo{booktitle}{\emph{30th Annual European Symposium on Algorithms, ESA 2022}}. Schloss Dagstuhl-Leibniz-Zentrum fur Informatik GmbH, Dagstuhl Publishing.
\newblock


\bibitem[Hermelin et~al\mbox{.}(2015)]%
        {HermelinKSWW15}
\bibfield{author}{\bibinfo{person}{Danny Hermelin}, \bibinfo{person}{Stefan Kratsch}, \bibinfo{person}{Karolina Soltys}, \bibinfo{person}{Magnus Wahlstr{\"{o}}m}, {and} \bibinfo{person}{Xi Wu}.} \bibinfo{year}{2015}\natexlab{}.
\newblock \showarticletitle{A Completeness Theory for Polynomial (Turing) Kernelization}.
\newblock \bibinfo{journal}{\emph{Algorithmica}} \bibinfo{volume}{71}, \bibinfo{number}{3} (\bibinfo{year}{2015}), \bibinfo{pages}{702--730}.
\newblock


\bibitem[Igarashi et~al\mbox{.}(2019)]%
        {igarashi2019robustness}
\bibfield{author}{\bibinfo{person}{Ayumi Igarashi}, \bibinfo{person}{Kazunori Ota}, \bibinfo{person}{Yuko Sakurai}, {and} \bibinfo{person}{Makoto Yokoo}.} \bibinfo{year}{2019}\natexlab{}.
\newblock \showarticletitle{Robustness against agent failure in hedonic games}. In \bibinfo{booktitle}{\emph{Proceedings of the 28th International Joint Conference on Artificial Intelligence}}. \bibinfo{pages}{364--370}.
\newblock


\bibitem[Impagliazzo and Paturi(2001)]%
        {IP01}
\bibfield{author}{\bibinfo{person}{Russell Impagliazzo} {and} \bibinfo{person}{Ramamohan Paturi}.} \bibinfo{year}{2001}\natexlab{}.
\newblock \showarticletitle{On the Complexity of k-SAT}.
\newblock \bibinfo{journal}{\emph{J. Comput. Syst. Sci.}} \bibinfo{volume}{62}, \bibinfo{number}{2} (\bibinfo{year}{2001}), \bibinfo{pages}{367--375}.
\newblock
\urldef\tempurl%
\url{https://doi.org/10.1006/JCSS.2000.1727}
\showDOI{\tempurl}


\bibitem[Jansen et~al\mbox{.}(2013)]%
        {JKMS13}
\bibfield{author}{\bibinfo{person}{Klaus Jansen}, \bibinfo{person}{Stefan Kratsch}, \bibinfo{person}{D{\'{a}}niel Marx}, {and} \bibinfo{person}{Ildik{\'{o}} Schlotter}.} \bibinfo{year}{2013}\natexlab{}.
\newblock \showarticletitle{Bin packing with fixed number of bins revisited}.
\newblock \bibinfo{journal}{\emph{J. Comput. System Sci.}} \bibinfo{volume}{79}, \bibinfo{number}{1} (\bibinfo{year}{2013}), \bibinfo{pages}{39--49}.
\newblock
\urldef\tempurl%
\url{https://doi.org/10.1016/j.jcss.2012.04.004}
\showDOI{\tempurl}


\bibitem[Kloks(1994)]%
        {Kloks94}
\bibfield{author}{\bibinfo{person}{Ton Kloks}.} \bibinfo{year}{1994}\natexlab{}.
\newblock \bibinfo{booktitle}{\emph{Treewidth, Computations and Approximations}}. \bibinfo{series}{Lecture Notes in Computer Science}, Vol.~\bibinfo{volume}{842}.
\newblock \bibinfo{publisher}{Springer}.
\newblock
\showISBNx{3-540-58356-4}
\urldef\tempurl%
\url{https://doi.org/10.1007/BFb0045375}
\showDOI{\tempurl}


\bibitem[Korhonen and Lokshtanov(2023)]%
        {KorhonenL23}
\bibfield{author}{\bibinfo{person}{Tuukka Korhonen} {and} \bibinfo{person}{Daniel Lokshtanov}.} \bibinfo{year}{2023}\natexlab{}.
\newblock \showarticletitle{An Improved Parameterized Algorithm for Treewidth}. In \bibinfo{booktitle}{\emph{Proceedings of the 55th Annual {ACM} Symposium on Theory of Computing, {STOC} 2023}}. \bibinfo{publisher}{{ACM}}, \bibinfo{pages}{528--541}.
\newblock
\urldef\tempurl%
\url{https://doi.org/10.1145/3564246.3585245}
\showDOI{\tempurl}


\bibitem[Lee(2017)]%
        {soda}
\bibfield{author}{\bibinfo{person}{Euiwoong Lee}.} \bibinfo{year}{2017}\natexlab{}.
\newblock \showarticletitle{Partitioning a graph into small pieces with applications to path transversal}. In \bibinfo{booktitle}{\emph{Proceedings of the Twenty-Eighth Annual ACM-SIAM Symposium on Discrete Algorithms}}. SIAM, \bibinfo{pages}{1546--1558}.
\newblock


\bibitem[Lenstra~Jr.(1983)]%
        {Len83}
\bibfield{author}{\bibinfo{person}{H.~W. Lenstra~Jr.}} \bibinfo{year}{1983}\natexlab{}.
\newblock \showarticletitle{Integer Programming with a Fixed Number of Variables}.
\newblock \bibinfo{journal}{\emph{Mathematics of Operations Research}} \bibinfo{volume}{8}, \bibinfo{number}{4} (\bibinfo{year}{1983}), \bibinfo{pages}{538--548}.
\newblock


\bibitem[Levinger et~al\mbox{.}(2023)]%
        {levinger2023social}
\bibfield{author}{\bibinfo{person}{Chaya Levinger}, \bibinfo{person}{Amos Azaria}, {and} \bibinfo{person}{Noam Hazon}.} \bibinfo{year}{2023}\natexlab{}.
\newblock \showarticletitle{Social Aware Coalition Formation with Bounded Coalition Size}. In \bibinfo{booktitle}{\emph{Proceedings of the 2023 International Conference on Autonomous Agents and Multiagent Systems}}. \bibinfo{pages}{2667--2669}.
\newblock


\bibitem[Li(2021)]%
        {L21}
\bibfield{author}{\bibinfo{person}{Fu Li}.} \bibinfo{year}{2021}\natexlab{}.
\newblock \showarticletitle{Fractional Hedonic Games With a Limited Number of Coalitions}. In \bibinfo{booktitle}{\emph{Proceedings of the 22nd Italian Conference on Theoretical Computer Science, Bologna, Italy, September 13-15, 2021}}.
\newblock


\bibitem[Maniu et~al\mbox{.}(2019)]%
        {maniu2019experimental}
\bibfield{author}{\bibinfo{person}{Silviu Maniu}, \bibinfo{person}{Pierre Senellart}, {and} \bibinfo{person}{Suraj Jog}.} \bibinfo{year}{2019}\natexlab{}.
\newblock \showarticletitle{An experimental study of the treewidth of real-world graph data}. In \bibinfo{booktitle}{\emph{ICDT 2019--22nd International Conference on Database Theory}}.
\newblock


\bibitem[Monaco and Moscardelli(2023)]%
        {monaco2023nash}
\bibfield{author}{\bibinfo{person}{Gianpiero Monaco} {and} \bibinfo{person}{Luca Moscardelli}.} \bibinfo{year}{2023}\natexlab{}.
\newblock \showarticletitle{Nash Stability in Fractional Hedonic Games with Bounded Size Coalitions}. In \bibinfo{booktitle}{\emph{International Conference on Web and Internet Economics}}.
\newblock


\bibitem[Monaco et~al\mbox{.}(2021)]%
        {monaco2021additively}
\bibfield{author}{\bibinfo{person}{Gianpiero Monaco}, \bibinfo{person}{Luca Moscardelli}, {and} \bibinfo{person}{Yllka Velaj}.} \bibinfo{year}{2021}\natexlab{}.
\newblock \showarticletitle{Additively Separable Hedonic Games with Social Context}.
\newblock \bibinfo{journal}{\emph{Games}} \bibinfo{volume}{12}, \bibinfo{number}{3} (\bibinfo{year}{2021}), \bibinfo{pages}{71}.
\newblock


\bibitem[Niedermeier(2006)]%
        {Niedermeier06}
\bibfield{author}{\bibinfo{person}{Rolf Niedermeier}.} \bibinfo{year}{2006}\natexlab{}.
\newblock \bibinfo{booktitle}{\emph{Invitation to Fixed-Parameter Algorithms}}.
\newblock \bibinfo{publisher}{Oxford University Press}.
\newblock
\showISBNx{9780198566076}
\urldef\tempurl%
\url{https://doi.org/10.1093/ACPROF:OSO/9780198566076.001.0001}
\showDOI{\tempurl}


\bibitem[Niedermeier and Rossmanith(2000)]%
        {niedermeier2000general}
\bibfield{author}{\bibinfo{person}{Rolf Niedermeier} {and} \bibinfo{person}{Peter Rossmanith}.} \bibinfo{year}{2000}\natexlab{}.
\newblock \showarticletitle{A general method to speed up fixed-parameter-tractable algorithms}.
\newblock \bibinfo{journal}{\emph{Inform. Process. Lett.}} \bibinfo{volume}{73}, \bibinfo{number}{3-4} (\bibinfo{year}{2000}), \bibinfo{pages}{125--129}.
\newblock


\bibitem[Ohta et~al\mbox{.}(2017)]%
        {ohta2017core}
\bibfield{author}{\bibinfo{person}{Kazunori Ohta}, \bibinfo{person}{Nathana{\"e}l Barrot}, \bibinfo{person}{Anisse Ismaili}, \bibinfo{person}{Yuko Sakurai}, {and} \bibinfo{person}{Makoto Yokoo}.} \bibinfo{year}{2017}\natexlab{}.
\newblock \showarticletitle{Core Stability in Hedonic Games among Friends and Enemies: Impact of Neutrals.}. In \bibinfo{booktitle}{\emph{IJCAI}}. \bibinfo{pages}{359--365}.
\newblock


\bibitem[Olsen(2009)]%
        {social}
\bibfield{author}{\bibinfo{person}{Martin Olsen}.} \bibinfo{year}{2009}\natexlab{}.
\newblock \showarticletitle{Nash stability in additively separable hedonic games and community structures}.
\newblock \bibinfo{journal}{\emph{Theory of Computing Systems}}  \bibinfo{volume}{45} (\bibinfo{year}{2009}), \bibinfo{pages}{917--925}.
\newblock


\bibitem[Peters(2016a)]%
        {hedonicTreewidth}
\bibfield{author}{\bibinfo{person}{Dominik Peters}.} \bibinfo{year}{2016}\natexlab{a}.
\newblock \showarticletitle{Graphical hedonic games of bounded treewidth}. In \bibinfo{booktitle}{\emph{Proceedings of the AAAI Conference on Artificial Intelligence}}, Vol.~\bibinfo{volume}{30}.
\newblock


\bibitem[Peters(2016b)]%
        {P16}
\bibfield{author}{\bibinfo{person}{Dominik Peters}.} \bibinfo{year}{2016}\natexlab{b}.
\newblock \showarticletitle{Towards Structural Tractability in Hedonic Games}. In \bibinfo{booktitle}{\emph{Proceedings of the AAAI Conference on Artificial Intelligence}}.
\newblock


\bibitem[Peters and Elkind(2015)]%
        {peters2015simple}
\bibfield{author}{\bibinfo{person}{Dominik Peters} {and} \bibinfo{person}{Edith Elkind}.} \bibinfo{year}{2015}\natexlab{}.
\newblock \showarticletitle{Simple causes of complexity in hedonic games}. In \bibinfo{booktitle}{\emph{Proceedings of the 24th International Conference on Artificial Intelligence}}. \bibinfo{pages}{617--623}.
\newblock


\bibitem[Saad et~al\mbox{.}(2010)]%
        {wireless}
\bibfield{author}{\bibinfo{person}{Walid Saad}, \bibinfo{person}{Zhu Han}, \bibinfo{person}{Tamer Basar}, \bibinfo{person}{M{\'e}rouane Debbah}, {and} \bibinfo{person}{Are Hjorungnes}.} \bibinfo{year}{2010}\natexlab{}.
\newblock \showarticletitle{Hedonic coalition formation for distributed task allocation among wireless agents}.
\newblock \bibinfo{journal}{\emph{IEEE Transactions on Mobile Computing}} \bibinfo{volume}{10}, \bibinfo{number}{9} (\bibinfo{year}{2010}), \bibinfo{pages}{1327--1344}.
\newblock


\bibitem[Sliwinski and Zick(2017)]%
        {sliwinski2017learning}
\bibfield{author}{\bibinfo{person}{Jakub Sliwinski} {and} \bibinfo{person}{Yair Zick}.} \bibinfo{year}{2017}\natexlab{}.
\newblock \showarticletitle{Learning Hedonic Games.}. In \bibinfo{booktitle}{\emph{IJCAI}}. \bibinfo{pages}{2730--2736}.
\newblock


\bibitem[van Rooij et~al\mbox{.}(2013)]%
        {RNB13}
\bibfield{author}{\bibinfo{person}{Johan M.~M. van Rooij}, \bibinfo{person}{Marcel~E. van Kooten~Niekerk}, {and} \bibinfo{person}{Hans~L. Bodlaender}.} \bibinfo{year}{2013}\natexlab{}.
\newblock \showarticletitle{Partition Into Triangles on Bounded Degree Graphs}.
\newblock \bibinfo{journal}{\emph{Theory Comput. Syst.}} \bibinfo{volume}{52}, \bibinfo{number}{4} (\bibinfo{year}{2013}), \bibinfo{pages}{687--718}.
\newblock


\bibitem[Weihe(1998)]%
        {weihe1998covering}
\bibfield{author}{\bibinfo{person}{Karsten Weihe}.} \bibinfo{year}{1998}\natexlab{}.
\newblock \showarticletitle{Covering trains by stations or the power of data reduction}.
\newblock \bibinfo{journal}{\emph{Proceedings of Algorithms and Experiments, ALEX}} (\bibinfo{year}{1998}), \bibinfo{pages}{1--8}.
\newblock


\end{thebibliography}


\end{document}